\documentclass[%
preprint,
amsmath,amssymb,
aps,
pra,
]{revtex4-2}

\usepackage{graphicx}
\usepackage{xcolor}
\usepackage{amsfonts}
\usepackage{gensymb}
\usepackage{extarrows}
\usepackage{stackengine}
\usepackage{tikz} 
\usetikzlibrary{calc,fadings,decorations.pathreplacing,arrows,mindmap}
\usepackage[export]{adjustbox}
\usepackage{amsthm} 
\usepackage{chemmacros}  
\usepackage{chemfig,siunitx} 

\usepackage{parskip}
\usepackage{romannum}
\usepackage{amsthm} 
\usepackage{caption}
\usepackage{subcaption}
\usepackage{yhmath}
\usepackage[ruled]{algorithm2e}
\usepackage{multirow}
\usepackage{placeins}
\usepackage{mathrsfs}

\makeatletter
\tikzset{circle split part fill/.style  args={#1,#2}{%
		alias=tmp@name, 
		postaction={%
			insert path={
				\pgfextra{%
					\pgfpointdiff{\pgfpointanchor{\pgf@node@name}{center}}%
					{\pgfpointanchor{\pgf@node@name}{east}}%
					\pgfmathsetmacro\insiderad{\pgf@x}
					\fill[#1] (\pgf@node@name.base) ([xshift=-\pgflinewidth]\pgf@node@name.east) arc
					(0:180:\insiderad-\pgflinewidth)--cycle;
					\fill[#2] (\pgf@node@name.base) ([xshift=\pgflinewidth]\pgf@node@name.west)  arc
					(180:360:\insiderad-\pgflinewidth)--cycle;            
}}}}}  
\makeatother

\newtheorem{theorem}{Theorem}[section]

\newtheorem{remark}{Remark}

\graphicspath{ {./figures2/} }

\def\bC{{\bf C}}
\def\bK{{\bf K}}
\def\bQ{{\bf Q}}

\def\tbC{\widetilde{{\bf C}}}

\def\PP{{\mathbb P}}
\def\RR{{\mathbb R}}

\def\J{{\mathscr J}}

\begin{document}


\title{Data-Driven Optimal Closures for Mean-Cluster Models: Beyond the Classical Pair Approximation}



\author{Avesta Ahmadi}
\affiliation{School of Computational Science \& Engineering, McMaster University, Hamilton, Ontario, Canada, L8S 4L8}

\author{Jamie M.~Foster}
\affiliation{School of Mathematics \& Physics, University of Portsmouth, Portsmouth, Hampshire, UK, PO1 2UP}

\author{Bartosz Protas}
\email{bprotas@mcmaster.ca}
\homepage{https://ms.mcmaster.ca/bprotas/}
\affiliation{Department of Mathematics \& Statistics, McMaster University, Hamilton, Ontario, Canada, L8S 4L8}

\date{\today}

\begin{abstract}
  This study concerns the mean-clustering approach to modelling the
  evolution of lattice dynamics. Instead of tracking the state of
  individual lattice sites, this approach describes the time evolution
  of the concentrations of different cluster types. It leads to an
  infinite hierarchy of ordinary differential equations which must be
  closed by truncation using a so-called closure condition. This
  condition approximates the concentrations of higher-order clusters
  in terms of the concentrations of lower-order ones. The pair
  approximation is the most common form of closure. Here, we consider
  its generalization, termed the ``optimal approximation'', which we
  calibrate using a robust data-driven strategy. To fix attention, we
  focus on a recently proposed structured lattice model for a
  nickel-based oxide, similar to that used as cathode material in
  modern commercial Li-ion batteries. The form of the obtained optimal
  approximation allows us to deduce a simple sparse closure model. In
  addition to being more accurate than the classical pair
  approximation, this ``sparse approximation'' is also physically
  interpretable which allows us to a posteriori refine the hypotheses
  underlying construction of this class of closure models. Moreover,
  the mean-cluster model closed with this sparse approximation is
  linear and hence analytically solvable such that its parametrization
  is straightforward. On the other hand, parametrization of the
  mean-cluster model closed with the pair approximation is shown to
  lead to an ill-posed inverse problem.
\end{abstract}


\maketitle

\section{Introduction}
\label{sec:intro}

Evolution of particles on a structured lattice is typically described
by discrete lattice models rather than continuous space models.  These
models are usually not solvable exactly and have to be studied through
computer simulations. One approach to describing the evolution of
particles on a structured lattice is to keep track of all interacting
particles as is done in various Monte-Carlo techniques such as
simulated annealing. However, these methods are costly as they
determine the lattice structure which is unnecessary in many
applications. What is often sufficient is knowledge of
the type and the number of different clusters in the lattice, which
can then be used for model fitting purposes along with experimental
measurements such as, e.g., Nuclear Magnetic Resonance (NMR) data
\cite{Harris2017}.  Hence, as an alternative to Monte-Carlo methods,
one can develop a simplified description of particle interactions in
terms of evolving probabilities of particle clusters of different
types in the form of a dynamical system which is sufficient for many
applications.  These approaches are referred to as ``mean-field
clustering methods'' and find applications in many areas of science
and engineering. The Ising model, as a canonical application of
mean-field methods, is a model of ferromagnetism describing the
evolution of magnetic moments in a lattice. Both Monte-Carlo methods
\cite{hasenbusch2001monte} and mean-field methods
\cite{strecka2015brief} have been employed to study this problem.
Another example of the application of such models is the contact
process which is a stochastic process describing the growth of a
population on a structured or unstructured lattice. Cluster
approximations are used to find mean-field properties of such
systems. Population dynamics in ecology \cite{Matsuda1992,Harada1994}
is one example of such processes.  Another example is the disease
spread in epidemiology that has been widely studied on structured
networks
\cite{Silva2020,Keeling2011,Eames2002,Pedro2015,Sato1994,Keeling1997,Bauch2005}
and complex networks \cite{DeOliveira2019,Mata2013}.  Failure
propagation \cite{Lin2020} and emergence of marriage networks
\cite{Pei2017} are some other examples of contact processes.

The focus of the present study is on cluster-based modelling of
systems of interacting particles on two-dimensional (2D) structured
lattices. The specific application which motivates the present study
is related to prediction of the structure of materials used in
Lithium-ion batteries \cite{Harris2017}. Using a cluster approximation
method, one can construct a hierarchical dynamical system describing
the evolution of concentrations of different clusters in the lattice
during a real annealing process.  In other words, the evolution of
concentrations of clusters of size $n$ involves concentrations of
clusters of size ($n+1$).  To solve this system of equations one is
required to close it by prescribing the evolution of concentrations of
($n+1$) clusters, which in turn will be determined by probabilities of
clusters of a still higher order.  This process therefore gives rise
to an infinite hierarchy of equations which is exact but is
intractable both analytically and computationally.  Thus, one needs to
truncate and close this infinite hierarchy of equations. Various
moment closure approximations have been used for this
purpose. Ben-Avraham et al.~\cite{Ben-Avraham1992} proposed a class of
approximations for 1D lattices with extensions developed in higher
dimensions, namely, the mean-field and pair approximations. These
techniques take into account local interactions between
{neighbouring} elements only and completely neglect interactions
between non-nearest neighbours on a lattice.  Applications of
mean-field and pair approximation methods to various problems in
science and engineering can be found in
\cite{DeOliveira2019,Pedro2015,Sugden2007,Joo2004} and in
\cite{Silva2020,Keeling2011,DeOliveira2019,Pedro2015,Mata2014,Mata2013,Sugden2007,Joo2004},
respectively. Some extensions of the pair approximation technique are
also introduced in \cite{Filipe2003} where interactions between
different elements are considered to be generic functions of distance.
In the present study our goal is to develop and validate a general
data-driven methodology that will allow us to optimally close (in a
mathematically precise sense) the infinite hierarchy of equations. We
will refer to this approach as the ``optimal approximation''.  This
approach leads to a general simple and mathematically interpretable
closure model.

As an emerging application of lattice dynamics, Harris et
al.~\cite{Harris2017} used a simulated annealing approach to
investigate the crystalline structure of cathode materials used in
state-of-the-art Lithium-ion batteries. More precisely, they focused
on layers of NMC (Nickel-Manganese-Cobalt) used in most modern 
commercial Li-ion batteries. These cathodes are described by the chemical
formula Li(NMC)O$_2$, where 2D layers of Lithium, Oxygen and NMC are
stacked on top of each other.  The capacity enhancement observed in
such materials is attributed to changes in the local microscopic
structure of the cathode layers \cite{Seo2016,Shukla2015}, however,
important aspects of this structure are not yet completely understood.
Hence, further refinement of this battery technology requires more
information about the arrangement of elements inside these layers.  In
\cite{Harris2017} simulated annealing was used to generate statistical
information about arrangements of different species on the lattice in
the NMC layer of a cathode, which was very costly and did not scale up
to large lattice sizes. The model developed in the present study aims
to address this limitation. While the proposed approach is general and
can be applied to many lattice systems, to fix attention, we will
develop it here for the problem from \cite{Harris2017} as an
example. Other applications of approaches based on lattice dynamics in
physics and chemistry include organic synthesis reactions in the
fields of heterogeneous catalysis and materials engineering
\cite{Lisiecki2021}, adsorption models of binary mixtures
\cite{Sanchez-Varretti2021} and microstructure mapping of perovskite
materials \cite{Chen2021}.

In this work, we use the mean-clustering approach to build a
hierarchical system of equations for the evolution of concentrations
of different clusters inside a structured lattice of the NMC cathode
layer. We assume a triangular lattice compatible with the structure of
the NMC layer \cite{Harris2017}. This spatial structure is important
in detecting the rotational symmetries of the system. A dynamical
system is constructed to describe reactions between different species
which are limited to swaps between nearest-neighbour elements. The
underlying principle is that as the ``temperature'' decreases the
lattice converges to a certain equilibrium state through a series of
element swaps, controlled by specific rate constants. Our new approach
consists of two distinct steps: first, the truncated hierarchical
dynamical system is closed using an optimal approximation whose
parameters are inferred from simulated annealing data; it is
demonstrated that such an optimal closure is in fact both simpler and
more accurate than the nearest-neighbour approximation proposed in
\cite{Ben-Avraham1992}. Additionally, robustness of the predictive
performance of the obtained model is demonstrated based on problems
with different stoichiometries.  Second, the reaction rates
parameterizing the dynamical system with the three types of closure,
i.e., pair approximation, optimal approximation and sparse
approximation, are inferred from the simulated annealing data using a
Bayesian approach which also allows us to estimate the uncertainty of
these reconstructions; this will show that the model with the optimal
closure is also less prone to calibration uncertainty than the model
closed with the nearest-neighbour approximation.

The paper is organized as follows: further details about our model
problem are presented in Section \ref{sec:problem}; then, in Section
\ref{sec:cluster} we introduce a dynamical system governing the
evolution of the concentrations of different clusters and in Section
\ref{sec:closure} we describe the closures we consider which are the
pair approximation and the optimal closure; Bayesian approach for
estimation of the reaction rates is introduced in Section
\ref{sec:Bayes}; computational results are presented in Section
\ref{sec:results} together with a justification for a suitably
sparse approximation, whereas discussion and conclusions are
deferred to Section \ref{sec:final}. Some technical material is
collected in Appendix \ref{sec:sym}.

\section{Model Problem}
\label{sec:problem}

In this section we provide some details about a lattice evolution
problem that will serve as our test case. In \cite{Harris2017} Harris
et al.~used a simulated annealing method to identify an evolving
arrangement of particles on the lattice and keep track of their
interactions. One material similar to the materials actually
used in Li-ion batteries is Li[Li$_{1/3}$Mn$_{2/3}$]O$_2$, where 2D
sheets of an oxygen layer, transition metal layer and Lithium layer
are stacked on top of each other, as shown in Figure \ref{fig:LiNMCO2}. 
Transition metal layer consists of Manganese and Lithium.
\begin{figure}
	\centering
	\begin{subfigure}[b]{0.35\textwidth}
		\centering
		\includegraphics[width=\textwidth]{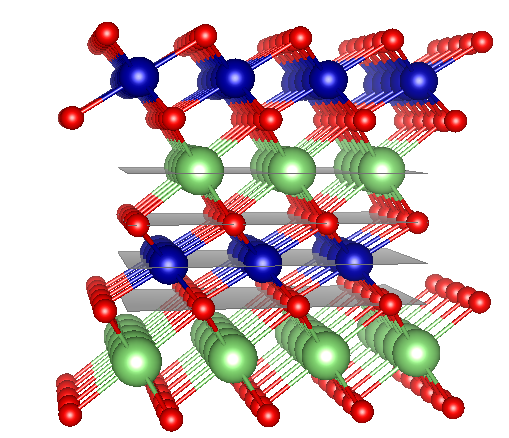}
		\caption{}
	\end{subfigure}\qquad
	\begin{subfigure}[b]{0.35\textwidth}
		\centering
		\includegraphics[width=\textwidth]{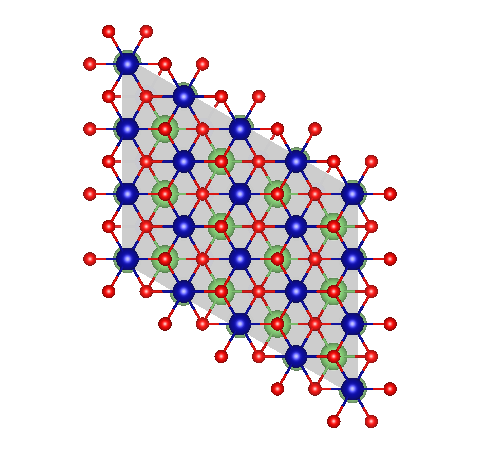}
		\caption{}
	\end{subfigure}
	\hfill
	\caption{The Li[Li$_{1/3}$Mn$_{2/3}$]O$_2$ lattice considered
          in \cite{Harris2017} and shown here in (a) a 3D view and (b)
          a 2D view. The red elements are oxygen atoms, green elements
          are lithium and the blue elements represent the transition
          metal layer elements which can be either lithium or
          manganese.}
	\label{fig:LiNMCO2}
\end{figure}

In the simulated annealing method, the energy of the system is
calculated by considering the local charge neutrality at oxygen sites.
Each oxygen element is surrounded by six nearest neighbours, cf.~Figure \ref{fig:LiNMCO2}.
The energy of each oxygen site is then determined by considering the
charge contributions of the neighbouring sites to its charge balance
with the goal of achieving neutrality.  The simulated annealing
approach attempts to find a 2D lattice configuration minimizing the
total energy of the system $E = \sum_i E_i$ corresponding to a
specific ``temperature'', where $E_i$ is the energy over each oxygen
site. This is a probabilistic approach to finding global optima in a
discrete space which mimics the annealing process applied to actual
materials. These materials are annealed at a high temperature,
followed by quenching to the desired temperature. Higher energy levels
of the system occur at higher temperatures and the evolution of the
system from higher temperatures to lower ones is controlled by a
probabilistic rule via the Boltzmann distribution.  This distribution
gives the probability that the system is in a certain state, given the
temperature and its energy level. In simulated annealing, state updates
occur through random swaps of the elements on the lattice
which destroys the ordering of elements and can lower the energy level
of the lattice.  The acceptance probability of an element swap is
\begin{equation}
\begin{aligned}
P = \begin{cases}
1, & \Delta E \le 0,\\
\exp\left( \frac{-\Delta E}{T} \right), & \Delta E > 0,
\end{cases}
\end{aligned}
\label{eq:boltzmann}
\end{equation}
where $T$ plays the role of the ``temperature''. Swaps leading to more
probable lower-energy states are always accepted, however, swaps
producing higher-energy states may also be accepted with some
probability in order to prevent rapid quenching of the system which
might result in convergence to a local minimum. We note that the
  concept of ``temperature'' used here is not equivalent to the
  thermodynamic temperature of the system. The choice of this
  pseudo-temperature, which controls the annealing protocol of the
  system, requires knowledge of the change of the system energy
  resulting from an element swap. In other words, the magnitude of the
  exponent $\frac{-\Delta E}{T}$ determines the quenching rate of the
  annealing process such that $T$ needs to be suitably adjusted.  This
  effective parameter $T$ is therefore related to the term $k_B\theta$
  in the Boltzmann distribution, where $k_B$ is the Boltzmann constant
  and $\theta$ the thermodynamic temperature of the system. The
choice of how the temperature is decreased is in principle arbitrary,
however, the equilibrium state must be reached at the end of the
annealing process for every arbitrarily chosen temperature profile.
The details of this approach can be found in \cite{Harris2017}.

In the crystal structure of the annealed metal layer of
Li[Li$_{1/3}$Mn$_{2/3}$]O$_2$ each triangle consists of two Mn
elements and one Li element. In this structure, the energy $E_i$ over
each oxygen site becomes zero and the total energy of the system will
be zero accordingly, as shown in Figure \ref{fig:lattice}b.  In
the simulated annealing study of this structure the temperature was
reduced in a stepwise manner, cf.~Figure \ref{fig:SA}a, and enough
time was allowed for the structure to stabilize at an equilibrium at
each intermediate temperature, cf.~Figure \ref{fig:SA}b.  The results
obtained for the system with Li$_{1/3}$Mn$_{2/3}$ are shown in the
form of the final lattice structure in Figure \ref{fig:lattice}.
Annealing experiments with the same protocol were also performed for
systems with different ratios of Li and Mn in Li$_{x}$Mn$_{1-x}$ where
$x \in \{0.25,0.30,0.33,0.36,0.42,0.50,0.58,0.64,0.70,0.75\}$, but
these results are not shown here for brevity. Our goal is to build a
model that will accurately predict the evolution of concentrations of
different particle clusters present in the lattice without having to
solve the entire annealing problem. We note that the elements Mn and
Li have charges, respectively, of $(+4)$ and $(+1)$. In
  simulated annealing these element charges are used to calculate the
  energy changes $\Delta E$ caused by element swaps and to determine
  the evolution of the lattice. However, the cluster approximation
  model, cf.~Section \ref{sec:cluster}, makes no assumptions about the
  charges of the elements and hence for simplicity the symbols $(+)$
  and $(-)$ will hereafter represent the elements Mn and Li,
  respectively.  The concentrations $\widetilde{C}_i$, $i \in \{
(++), (+-), (--)\}$ of $2$-clusters as functions of time (or
temperature), cf.~Figure \ref{fig:SA}c, will be used as data to
construct the optimal closure approximation and to infer the reaction
rates in the model. However, while these concentrations will be
provided for a single stoichiometry only, the resulting model will be
shown to remain accurate for a broad range of stoichiometries. The
lattice evolution in this method does not have a natural time scale
and for concreteness we will assume that the unit of time is set by an
individual iteration of the simulated annealing experiment. Notably,
in this model all concentrations are independent of location on the
lattice due to spatial homogeneity.
\begin{figure}
  \centering
  \begin{subfigure}[b]{0.32\textwidth}
    \centering
    \includegraphics[width=\textwidth]{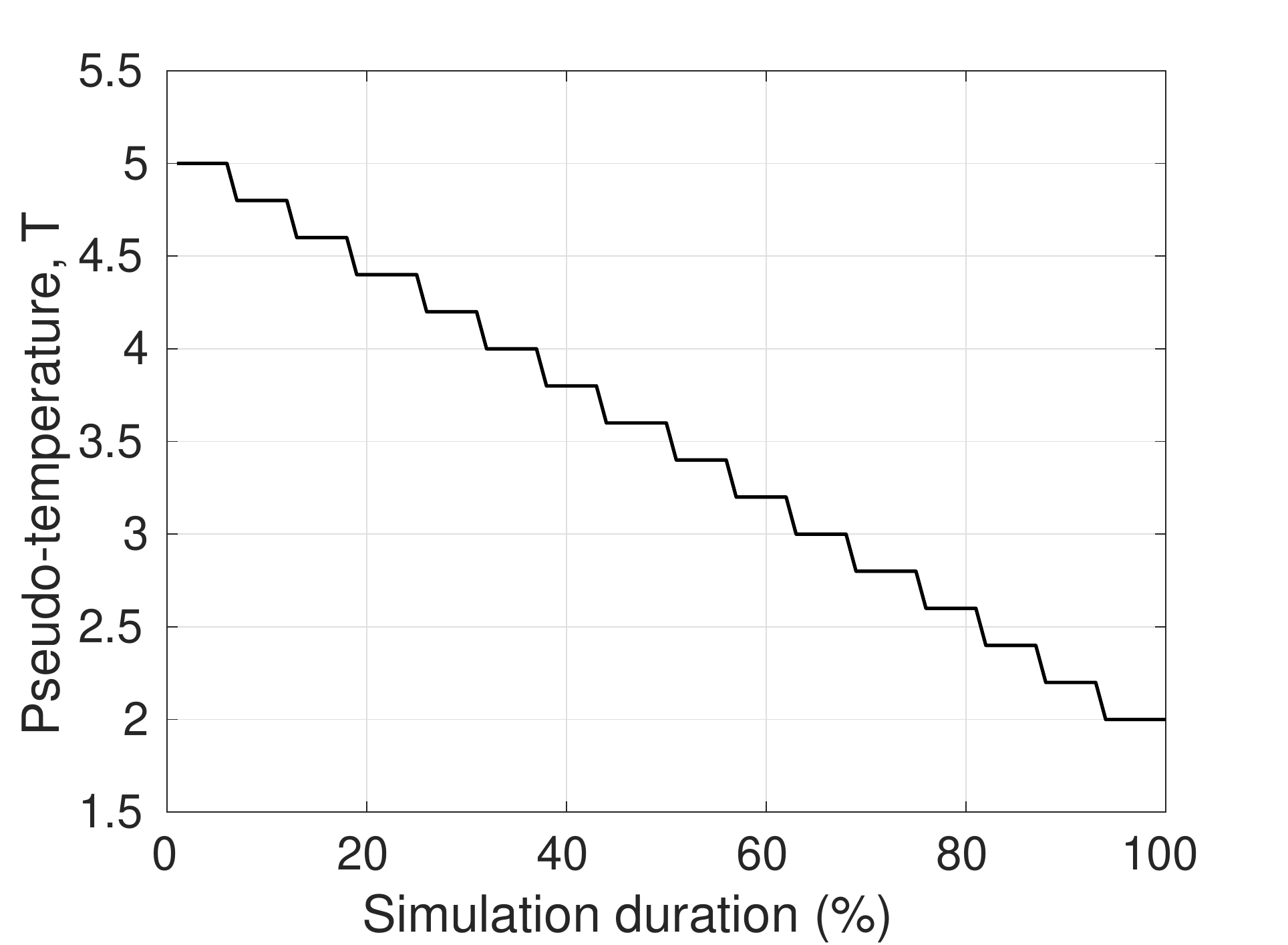}
    \caption{}
  \end{subfigure}
  \hfill
  \begin{subfigure}[b]{0.32\textwidth}
    \centering
    \includegraphics[width=\textwidth]{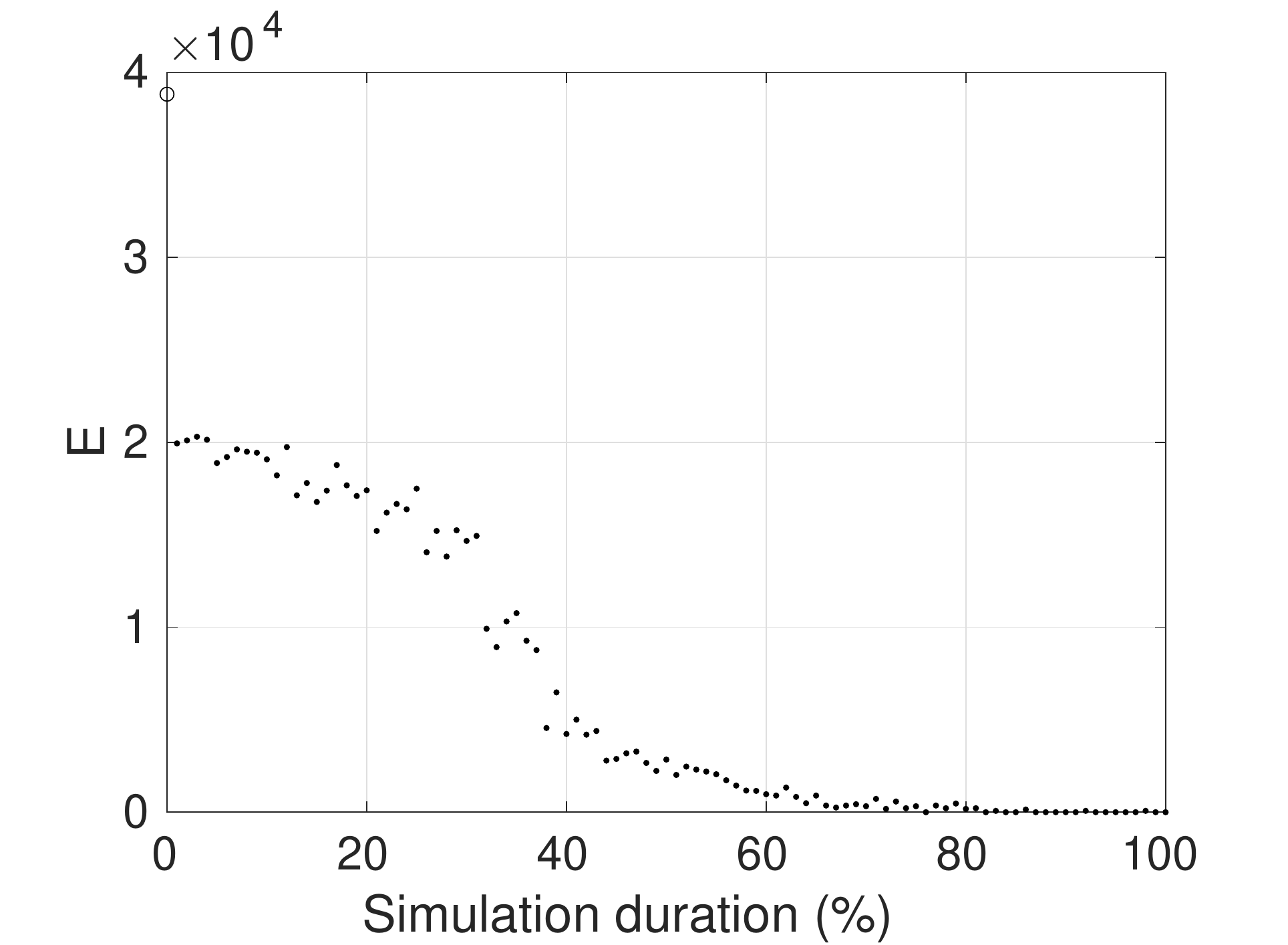}
    \caption{}
  \end{subfigure}
  \begin{subfigure}[b]{0.32\textwidth}
    \centering
    \includegraphics[width=\textwidth]{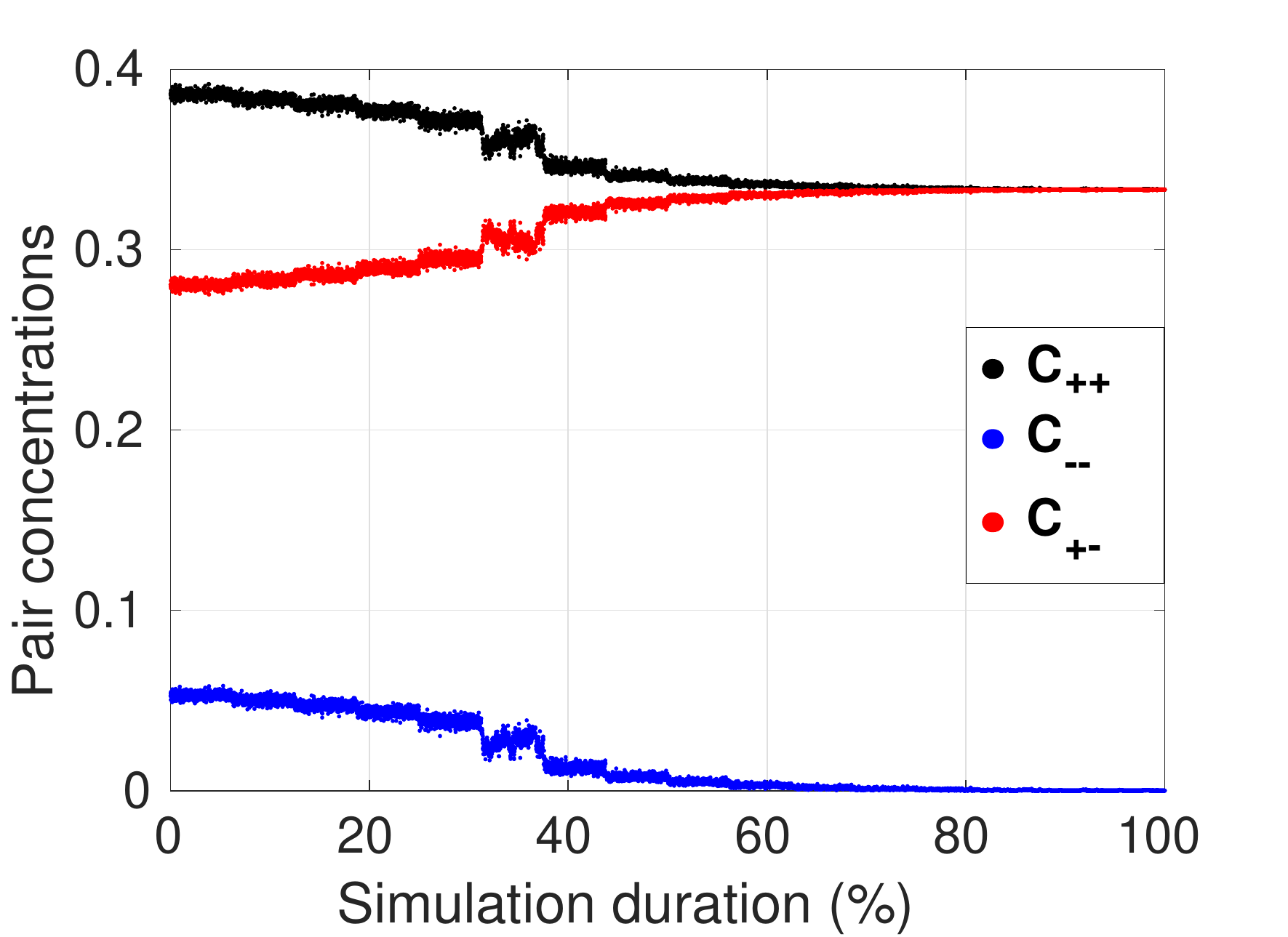}
    \caption{}
  \end{subfigure}
  \hfill
  \caption{Dependence of (a) the ``temperature'' of the system, $T$, (b)
    its lattice energy $E$ and (c) concentrations $\widetilde{C}_i$,
    $i \in \{ (++), (+-), (--)\}$, of different 2-clusters on time
    expressed as a fraction of the entire annealing experiment for the
    Li$_{1/3}$Mn$_{2/3}$ system.}
  \label{fig:SA}
\end{figure}
\begin{figure}
  \centering
  \begin{subfigure}[b]{0.45\textwidth}
    \centering
    \includegraphics[width=\textwidth]{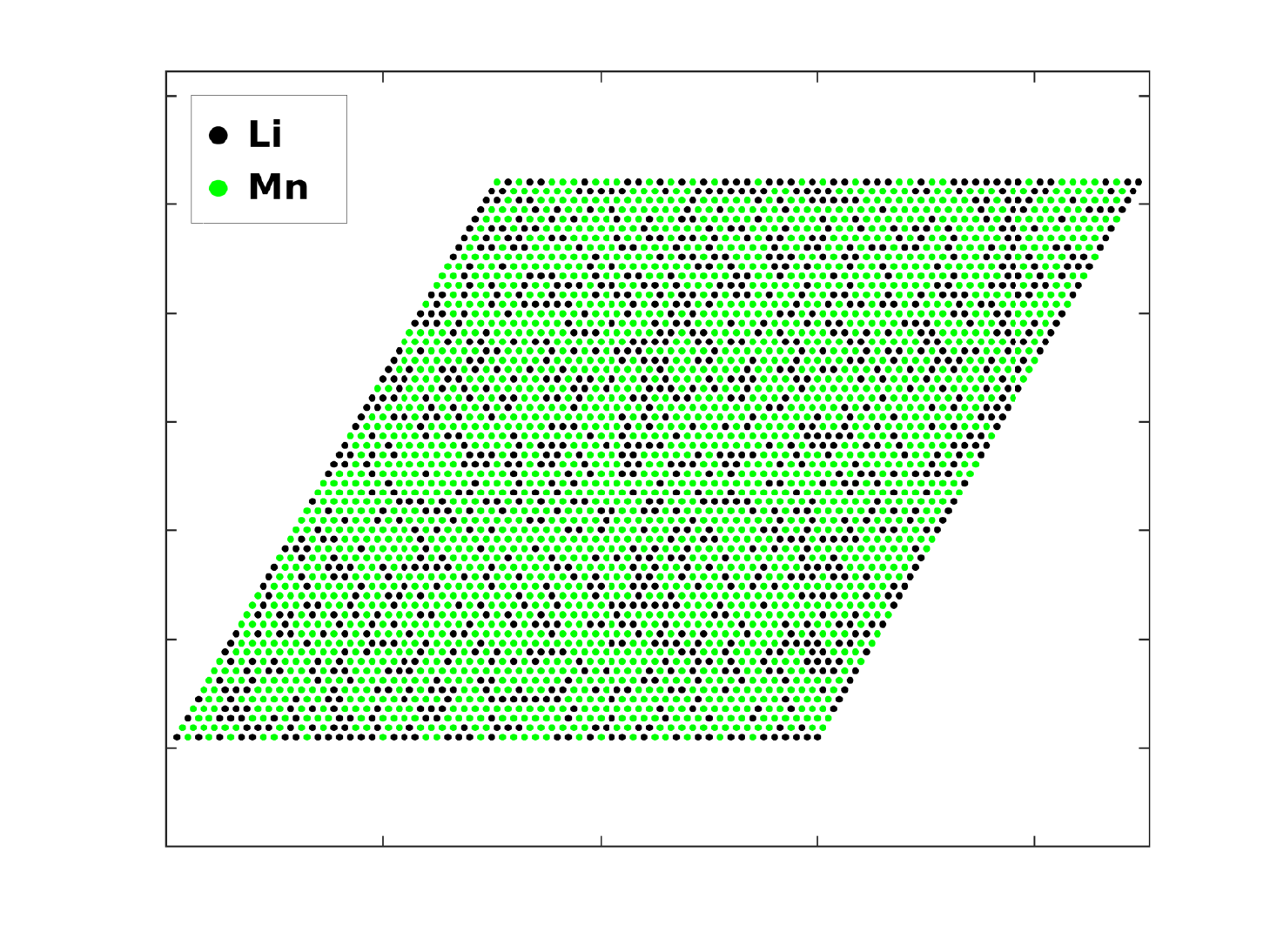}
    \caption{}
  \end{subfigure}
  \begin{subfigure}[b]{0.45\textwidth}
    \centering
    \includegraphics[width=\textwidth]{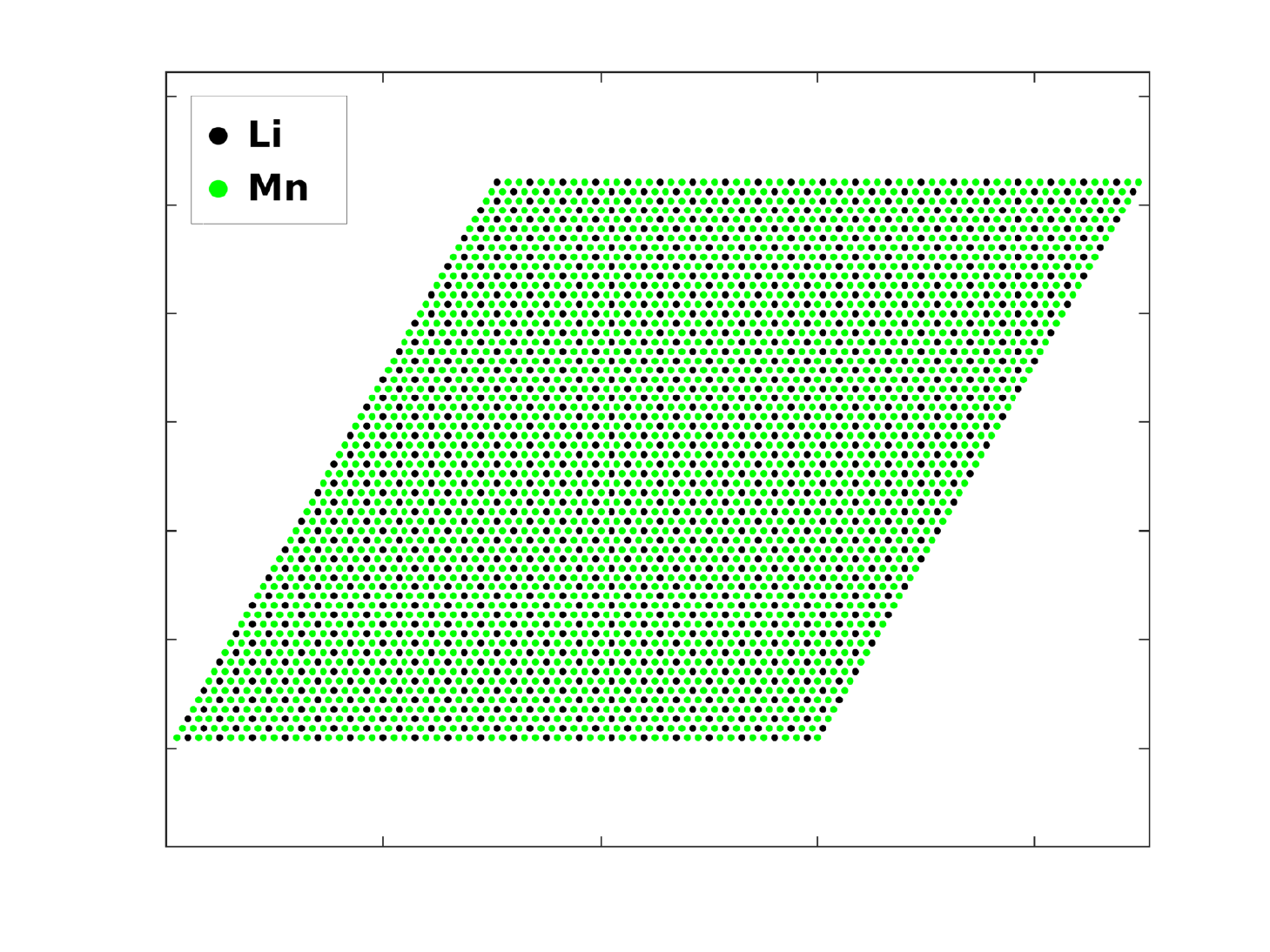}
    \caption{}
  \end{subfigure}
  \hfill
  \caption{(a) Initial random state and (b) the final ordered state of
    the lattice for the Li$_{1/3}$Mn$_{2/3}$ system obtained via
    simulated annealing \cite{Harris2017}. Black and green dots
    represent Li ions (more generally, negative elements) and Mn ions
    (more generally, positive elements), respectively.}
  \label{fig:lattice}
\end{figure}

\section{Cluster Approximation}
\label{sec:cluster}

In this section we develop a system of evolution equations for
concentrations of clusters in a two-element system with elements
denoted $(+)$ (or $\fplus$) and $(-)$ (or $\fminus$). We note that
these notations need not correspond to the charge of the elements. In
this study, a cluster of size $n$ is referred to as a $n$-cluster and
elements inside the cluster form a closed or an open chain. The
concentration of a cluster is defined as the probability of finding
that particular cluster among all clusters of the same shape but with
different compositions. As an example, the concentration of the
$3$-cluster shown in Figure \ref{fig:d2} is denoted $C_{ijk}$, where
$i,j,k \in \{+,-\}$.

\begin{figure}
  \centering
  \begin{tikzpicture} [->,>=stealth',auto,thick,main node/.style={circle,draw,font=\sffamily\Large\bfseries},scale=0.8]
    \coordinate (Origin)   at (0,0);
    \coordinate (XAxisMin) at (-5,0);
    \coordinate (XAxisMax) at (5,0);
    \coordinate (YAxisMin) at (0,-5);
    \coordinate (YAxisMax) at (0,5);
    \clip (-5.1,-5.1) circle (1.9cm); 
    \begin{scope}
      \pgftransformcm{1}{0}{1/2}{sqrt(3)/2}{\pgfpoint{0cm}{0cm}}
      \draw[style=help lines,dashed,draw=black] (-12,-12) grid[ystep=0cm,xstep=1cm] (12,12);
    \end{scope}
    
    \pgftransformcm{1}{0}{-1/2}{sqrt(3)/2}{\pgfpoint{0cm}{0cm}} 
    \draw[style=help lines,dashed] (-12,-12) grid[step=1cm] (12,12);
    \foreach \x in {-6,-5.5,...,1}{
      \foreach \y in {-6,-5.5,...,1}{
        \node[draw = black, draw opacity = 0.6, circle,inner sep=1.3pt,fill=black, fill opacity=0.6] at (2*\x,2*\y) {};
      }
    }
    
    \node(1)[draw = black,circle,inner sep=2pt,fill=white] at (-9,-6) {i};
    \node(2)[draw = black,circle,inner sep=1.3pt,fill=white] at (-8,-6) {j};
    \node(3)[draw = black,circle,inner sep=1.4pt,fill=white] at (-7,-6) {k};
    \draw [thick,-,black] (1) -- (2) node {};
    \draw [thick,-,black] (2) -- (3) node {};
  \end{tikzpicture}
  \caption{An example of a linear chain $3$-cluster on a 2D lattice.}
  \label{fig:d2}
\end{figure}
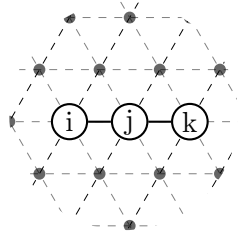

\begin{remark}	
  The normalization condition requires that the sum of the
  concentrations of all possible $n$-clusters with the
  same geometry must be equal to one \cite{Ben-Avraham1992}:
  \begin{equation}
    \begin{aligned}
      \sum_{S_1, S_2,\dots, S_n}^{} C_{S_1S_2\cdots S_n} = 1,
    \end{aligned}
    \label{eq:C=1}
  \end{equation}
  where the indices $1,2,3,\dots,n$ enumerate different sites within a
  cluster with two consecutive ones corresponding to nearest neighbours
  and $S_i \in \{+,-\}$ denotes the state of that specific site.
  Applying this to $1$-clusters and $2$-clusters in our model, the
  following equations are derived from the normalization condition:
  \begin{subequations}
    \begin{align}
      \begin{split}
        \normalfont{C_{+}} + \normalfont{C_{-}} = 1
        \label{eq:CCa}
      \end{split}\\
      \begin{split}
        \normalfont{C_{++}} +
        \normalfont{C_{--}} + 
        \normalfont{C_{+-}} +
        \normalfont{C_{-+}} = 1 \\
        \Rightarrow 
        \normalfont{C_{++}} +
        \normalfont{C_{--}} + 
        2 \normalfont{C_{+-}} = 1.
        \label{eq:CCb}
      \end{split}
    \end{align}
    \label{eq:CC}
  \end{subequations}
  The concentrations of the $(+-)$ and $(-+)$ clusters are the same
  due to the rotational symmetry of the system, as stated in Theorem
  \ref{thm:1} in the Appendix.
\end{remark}
The aim is to deduce a dynamical system describing the evolution of
the probabilities of $2$-clusters. There are three different types of
$2$-clusters found on the lattice, namely, $\fplus\fplus$,
$\fminus\fminus$ and $\fplus\fminus$.

\subsection{Production and Destruction of $2$-Clusters} 
\label{sec:dc}

The rate of change of the concentration of specific clusters is
determined by the rate at which they are produced and destroyed.
Production or destruction of a certain cluster occurs through swaps
among nearest-neighbour elements on the lattice. Each swap of
nearest-neighbour elements is called here a reaction. The rate
equations can then be derived using the window method
\cite{Ben-Avraham1992}. In this approach we consider all possible
reactions that change the composition of a particular 2-cluster in a
certain window containing this cluster, via a swap between one of the
elements inside the window and one of its nearest-neighbour elements
outside the window. For example, in order to derive the rate equation
for the $(\fplus\fplus)$ cluster, in Figure \ref{fig:++} we show all
possible reactions that will produce or destroy this cluster via
nearest-neighbour element swaps.  In each of the reactions, the
neighbour element (highlighted in red) will swap with one of the
elements of the window (highlighted in blue) to produce a
($\fplus\fplus$) cluster in the forward reaction.  Conversely, reverse
reactions destroy the ($\fplus\fplus$) cluster and produce a
{($\fplus\fminus$) cluster}.  The rotational symmetry of the
lattice allows us to reduce the number of possible reactions to those
shown in Figure \ref{fig:++}.  Moreover, reactions taking place inside
a triangular-shaped $3$-cluster do not change the total count of
2-clusters inside the triangle and are therefore disregarded. Each
reaction has a unique rate constant denoted $k_1,k_2,\dots$. The rate
constants have the units of $1/\text{sec}$ and control the evolution
of different clusters participating in a reaction. We note that in
deriving the rate equations each reaction is accounted for in
proportion to the number of its rotational symmetries on the lattice.

\begin{figure}
  \centering
  \begin{adjustbox}{minipage=\linewidth,scale=0.8}
    \begin{subfigure}[b]{0.48\textwidth}
      \scalebox{0.65}{
        \begin{tikzpicture}[->,thick,every node/.style={circle,draw=blue,font=\sffamily\Large\bfseries}]
          \draw[help lines] (0,0);
          \begin{scope}[node distance=0.1]
            \node (a)[inner sep=1pt] at (0,0) {+};
            \node (b)[inner sep=3.3pt] at (1,0) {--};
            \node (c)[draw = red,inner sep=1pt] at (1.5,+0.7) {+};
            \node [inner sep=7pt, fill=gray!20, draw=gray] at (0.5,-0.7) {};
            \node [inner sep=7pt, fill=gray!20, draw=gray] at (0.5,+0.7) {};
            \node [inner sep=7pt, fill=gray!20, draw=gray] at (-0.5,+0.7) {};
            \node [inner sep=7pt, fill=gray!20, draw=gray] at (-1,0) {};
            \node [inner sep=7pt, fill=gray!20, draw=gray] at (2,0) {};
            \node [inner sep=7pt, fill=gray!20, draw=gray] at (1.5,-0.7) {};
            \node [inner sep=7pt, fill=gray!20, draw=gray] at (-0.5,-0.7) {};
          \end{scope}
          \draw [<->] (b) to [ultra thick,out=350,in=330] (c);
          \node (1)[inner sep=2pt,draw=none] at (4.5,0.6) {$k_1$};
          \draw [-{Straight Barb[left]},black,line width=2pt] (3,0.1) to (6,0.1);
          \node (1)[inner sep=2pt,draw=none] at (4.5,-0.6) {$k_3$};
          \draw [-{Straight Barb[left]},black,line width=2pt] (6,-0.1) to (3,-0.1);
          \begin{scope}[node distance=0.1]
            \node (a)[inner sep=1pt] at (8,0) {+};
            \node (b)[inner sep=1pt] at (9,0) {+};
            \node (c)[draw = red,inner sep=3.3pt] at (9.5,+0.7) {--};
            \node [inner sep=7pt, fill=gray!20, draw=gray] at (7,0) {};
            \node [inner sep=7pt, fill=gray!20, draw=gray] at (10,0) {};
            \node [inner sep=7pt, fill=gray!20, draw=gray] at (7.5,+0.7) {};
            \node [inner sep=7pt, fill=gray!20, draw=gray] at (8.5,-0.7) {};
            \node [inner sep=7pt, fill=gray!20, draw=gray] at (8.5,+0.7) {};
            \node [inner sep=7pt, fill=gray!20, draw=gray] at (7.5,-0.7) {};
            \node [inner sep=7pt, fill=gray!20, draw=gray] at (9.5,-0.7) {};
          \end{scope}
        \end{tikzpicture}
      }
    \end{subfigure}\quad
    \begin{subfigure}[b]{0.48\textwidth}
      \scalebox{0.65}{
        \begin{tikzpicture}[->,thick,every node/.style={circle,draw=blue,font=\sffamily\Large\bfseries}]
          \draw[help lines] (0,0);
          \begin{scope}[node distance=0.1]
            \node (a)[inner sep=1pt] at (0,0) {+};
            \node (b)[inner sep=3.3pt] at (1,0) {--};
            \node (c)[draw = red,inner sep=1pt] at (2,0) {+};
            \node [inner sep=7pt, fill=gray!20, draw=gray] at (0.5,-0.7) {};
            \node [inner sep=7pt, fill=gray!20, draw=gray] at (0.5,+0.7) {};
            \node [inner sep=7pt, fill=gray!20, draw=gray] at (-0.5,+0.7) {};
            \node [inner sep=7pt, fill=gray!20, draw=gray] at (-1,0) {};
            \node [inner sep=7pt, fill=gray!20, draw=gray] at (1.5,-0.7) {};
            \node [inner sep=7pt, fill=gray!20, draw=gray] at (1.5,+0.7) {};
            \node [inner sep=7pt, fill=gray!20, draw=gray] at (-0.5,-0.7) {};
          \end{scope}
          \draw [<->] (b) to [ultra thick,out=80,in=100] (c);
          \node (1)[inner sep=2pt,draw=none] at (4.5,0.6) {$k_2$};
          \draw [-{Straight Barb[left]},black,line width=2pt] (3,0.1) to (6,0.1);
          \node (1)[inner sep=2pt,draw=none] at (4.5,-0.6) {$k_4$};
          \draw [-{Straight Barb[left]},black,line width=2pt] (6,-0.1) to (3,-0.1);
          \begin{scope}[node distance=0.1]
            \node (a)[inner sep=1pt] at (8,0) {+};
            \node (b)[inner sep=1pt] at (9,0) {+};
            \node (c)[draw = red,inner sep=3.3pt] at (10,0) {--};
            \node [inner sep=7pt, fill=gray!20, draw=gray] at (7,0) {};
            \node [inner sep=7pt, fill=gray!20, draw=gray] at (9.5,-0.7) {};
            \node [inner sep=7pt, fill=gray!20, draw=gray] at (7.5,+0.7) {};
            \node [inner sep=7pt, fill=gray!20, draw=gray] at (8.5,-0.7) {};
            \node [inner sep=7pt, fill=gray!20, draw=gray] at (8.5,+0.7) {};
            \node [inner sep=7pt, fill=gray!20, draw=gray] at (7.5,-0.7) {};
            \node [inner sep=7pt, fill=gray!20, draw=gray] at (9.5,+0.7) {};
          \end{scope}
        \end{tikzpicture}
      }
    \end{subfigure}
  \end{adjustbox}
  \caption{All unique (up to rotational and translational symmetries)
    reversible reactions to destroy or produce clusters
    (\textcircled{+}\textcircled{+}) and
    (\textcircled{+}\textcircled{--}).}
  \label{fig:++}
  \centering
  \begin{adjustbox}{minipage=\linewidth,scale=0.8}
    \begin{subfigure}[b]{0.48\textwidth}
      \scalebox{0.65}{
        \begin{tikzpicture}[->,thick,every node/.style={circle,draw=blue,font=\sffamily\Large\bfseries}]
          \draw[help lines] (0,0);
          \begin{scope}[node distance=0.1]
            \node (a)[inner sep=3.3pt] at (0,0) {--};
            \node (b)[inner sep=1pt] at (1,0) {+};
            \node (c)[draw = red,inner sep=3.3pt] at (1.5,+0.7) {--};
            \node [inner sep=7pt, fill=gray!20, draw=gray] at (0.5,-0.7) {};
            \node [inner sep=7pt, fill=gray!20, draw=gray] at (0.5,+0.7) {};
            \node [inner sep=7pt, fill=gray!20, draw=gray] at (-0.5,+0.7) {};
            \node [inner sep=7pt, fill=gray!20, draw=gray] at (-1,0) {};
            \node [inner sep=7pt, fill=gray!20, draw=gray] at (2,0) {};
            \node [inner sep=7pt, fill=gray!20, draw=gray] at (1.5,-0.7) {};
            \node [inner sep=7pt, fill=gray!20, draw=gray] at (-0.5,-0.7) {};
          \end{scope}
          \draw [<->] (b) to [ultra thick,out=350,in=330] (c);
          \node (1)[inner sep=2pt,draw=none] at (4.5,0.6) {$k_5$};
          \draw [-{Straight Barb[left]},black,line width=2pt] (3,0.1) to (6,0.1);
          \node (1)[inner sep=2pt,draw=none] at (4.5,-0.6) {$k_7$};
          \draw [-{Straight Barb[left]},black,line width=2pt] (6,-0.1) to (3,-0.1);
          \begin{scope}[node distance=0.1]
            \node (a)[inner sep=3.3pt] at (8,0) {--};
            \node (b)[inner sep=3.3pt] at (9,0) {--};
            \node (c)[draw = red,inner sep=1pt] at (9.5,+0.7) {+};
            \node [inner sep=7pt, fill=gray!20, draw=gray] at (7,0) {};
            \node [inner sep=7pt, fill=gray!20, draw=gray] at (10,0) {};
            \node [inner sep=7pt, fill=gray!20, draw=gray] at (7.5,+0.7) {};
            \node [inner sep=7pt, fill=gray!20, draw=gray] at (8.5,-0.7) {};
            \node [inner sep=7pt, fill=gray!20, draw=gray] at (8.5,+0.7) {};
            \node [inner sep=7pt, fill=gray!20, draw=gray] at (7.5,-0.7) {};
            \node [inner sep=7pt, fill=gray!20, draw=gray] at (9.5,-0.7) {};
          \end{scope}
        \end{tikzpicture}
      }
    \end{subfigure}\quad
    \begin{subfigure}[b]{0.48\textwidth}
      \scalebox{0.65}{
        \begin{tikzpicture}[->,thick,every node/.style={circle,draw=blue,font=\sffamily\Large\bfseries}]
          \draw[help lines] (0,0);
          \begin{scope}[node distance=0.1]
            \node (a)[inner sep=3.3pt] at (0,0) {--};
            \node (b)[inner sep=1pt] at (1,0) {+};
            \node (c)[draw = red,inner sep=3.3pt] at (2,0) {--};
            \node [inner sep=7pt, fill=gray!20, draw=gray] at (0.5,-0.7) {};
            \node [inner sep=7pt, fill=gray!20, draw=gray] at (0.5,+0.7) {};
            \node [inner sep=7pt, fill=gray!20, draw=gray] at (-0.5,+0.7) {};
            \node [inner sep=7pt, fill=gray!20, draw=gray] at (-1,0) {};
            \node [inner sep=7pt, fill=gray!20, draw=gray] at (1.5,-0.7) {};
            \node [inner sep=7pt, fill=gray!20, draw=gray] at (1.5,+0.7) {};
            \node [inner sep=7pt, fill=gray!20, draw=gray] at (-0.5,-0.7) {};
          \end{scope}
          \draw [<->] (b) to [ultra thick,out=80,in=100] (c);
          \node (1)[inner sep=2pt,draw=none] at (4.5,0.6) {$k_6$};
          \draw [-{Straight Barb[left]},black,line width=2pt] (3,0.1) to (6,0.1);
          \node (1)[inner sep=2pt,draw=none] at (4.5,-0.6) {$k_8$};
          \draw [-{Straight Barb[left]},black,line width=2pt] (6,-0.1) to (3,-0.1);
          \begin{scope}[node distance=0.1]
            \node (a)[inner sep=3.3pt] at (8,0) {--};
            \node (b)[inner sep=3.3pt] at (9,0) {--};
            \node (c)[draw = red,inner sep=1pt] at (10,0) {+};
            \node [inner sep=7pt, fill=gray!20, draw=gray] at (7,0) {};
            \node [inner sep=7pt, fill=gray!20, draw=gray] at (9.5,-0.7) {};
            \node [inner sep=7pt, fill=gray!20, draw=gray] at (7.5,+0.7) {};
            \node [inner sep=7pt, fill=gray!20, draw=gray] at (8.5,-0.7) {};
            \node [inner sep=7pt, fill=gray!20, draw=gray] at (8.5,+0.7) {};
            \node [inner sep=7pt, fill=gray!20, draw=gray] at (7.5,-0.7) {};
            \node [inner sep=7pt, fill=gray!20, draw=gray] at (9.5,+0.7) {};
          \end{scope}
        \end{tikzpicture}
      }
    \end{subfigure}
  \end{adjustbox}
  \caption{All unique (up to rotational and translational symmetries)
    reversible reactions to destroy or produce clusters
    (\textcircled{--}\textcircled{--}) and
    (\textcircled{+}\textcircled{--}).}
  \label{fig:--}
\end{figure}

As can be observed in Figure \ref{fig:++}, 3-clusters with three types
of bonds are involved in the derivation of rate equations. The first
type is the linear $3$-cluster in which the two bonds are colinear.
The second type is the cluster in which there is an obtuse angle of
120 degrees between the bonds due to the triangular shape of the
lattice. The third type is the triangular cluster in which the
elements form a triangle with 60 degrees between the bonds.  We will
refer to these as the linear, angled and triangular clusters,
respectively.  For simplicity, linear clusters will be represented as
a combination of elements with a straight line [$\overline{(\bullet
  \bullet \bullet)}$], angled clusters as a combination of elements
with a hat sign [$\widehat{(\bullet \bullet \bullet)}$] and
  triangular clusters as a combination of elements with a triangle
  [$\widetriangle{(\bullet \bullet \bullet)}$], where $\bullet$ is
either $+$ or $-$.  The set of all 3-cluster types will be denoted
\begin{equation}
\begin{aligned}
\Theta = \Big\{ & \overline{+++},\overline{---},\overline{++-},\overline{--+},\overline{+-+},\overline{-+-}, \\
                & \widehat{+++},\widehat{---}, \widehat{+-+},\widehat{-+-},\widehat{++-},\widehat{--+}, \\ 
                & \widetriangle{+++}, \widetriangle{---},\widetriangle{++-},\widetriangle{--+} \Big\}.
\end{aligned}
\label{eq:Theta}
\end{equation}
The rate equations for the ($\fminus\fminus$) and ($\fplus\fminus$)
clusters can be derived in a similar way, by considering all possible
reactions that produce or destroy these two clusters as shown in
Figure \ref{fig:--}. We thus obtain the following system of rate
equations for the concentrations $C_{++}$, $C_{--}$ and $C_{+-}$
\begin{subequations}
	\begin{align}
	\frac{d}{dt} C_{++}
	= 4k_1 C_{\widehat{+-+}}
	+ 2k_2 C_{\overline{+-+}}
	- 4k_3 C_{\widehat{++-}}
	- 2k_4 C_{\overline{++-}},
	\label{eq:dC++}
	\end{align}
	\begin{align}
	\frac{d}{dt} C_{--}
	= 4k_5 C_{\widehat{-+-}}
	+ 2k_6  C_{\overline{-+-}}
	- 4k_7 C_{\widehat{--+}}
	- 2k_8 C_{\overline{--+}},
	\label{eq:dC--}
	\end{align}
	\begin{align}
	\begin{split}
	\frac{d}{dt} C_{+-}
	= 2k_3 C_{\widehat{++-}}
	+ 2k_7 C_{\widehat{--+}}
	+ k_4 C_{\overline{++-}}
	+ k_8 C_{\overline{--+}}\\ 
	- 2k_1 C_{\widehat{+-+}}
	- 2k_5 C_{\widehat{-+-}}
	- k_2 C_{\overline{+-+}}
	- k_6  C_{\overline{-+-}} ~.
	\end{split}
	\label{eq:dC+-}
	\end{align}
	\label{eq:dC}
\end{subequations}
An important aspect of system \eqref{eq:dC} is its hierarchical
structure in the sense that the rates of change of concentrations of
2-clusters are given in terms of the concentrations of 3-clusters and
if one were to write down equations for their rates of change they
would involve concentrations of 4-clusters, etc. Thus, system
\eqref{eq:dC} is not closed and needs to be truncated which we will do so
here at the level of 2-clusters. Two strategies for closing the
truncated system are discussed in Section \ref{sec:closure}.

In addition, the normalization condition \eqref{eq:CCb} can be modified to a dynamic
form by taking the derivative with respect to time 
\begin{equation}
\begin{aligned}
\frac{d}{dt}C_{++} +
\frac{d}{dt} C_{--}
+ 2\frac{d}{dt} C_{+-} = 0.
\end{aligned}
\label{eq:dCCb}
\end{equation}
As can be verified, this equation is satisfied automatically by system
\eqref{eq:dC++}--\eqref{eq:dC+-}.  Moreover, the rate of the forward
reaction will be equal to the rate of corresponding reverse reaction
in the chemical equilibrium. As we are interested in the
equilibrium state of reactions, the following relations can be written
for each pair of forward and reverse reactions in equilibrium
\begin{subequations}
	\begin{align} 
	k_1 C_{\widehat{+-+}}
	= k_3 C_{\widehat{++-}} 
	\qquad \Longrightarrow \qquad 
	Q_1 = \frac{k_1}{k_3} = \frac{C_{\widehat{++-}}
          }{C_{\widehat{+-+}}}, 	\label{eq:Qa}
	\\[5mm]
	k_2 C_{\overline{+-+}}
	= k_4 C_{\overline{++-}}
	\qquad \Longrightarrow \qquad
	Q_2 = \frac{k_2}{k_4} =
          \frac{C_{\overline{++-}}}{C_{\overline{+-+}}}, 	\label{eq:Qb}
	\\[5mm]
	k_5 C_{\widehat{-+-}}
	= k_7 C_{\widehat{--+}}
	\qquad \Longrightarrow \qquad
	Q_3 = \frac{k_5}{k_7} = \frac{
          C_{\widehat{--+}}}{C_{\widehat{-+-}}}, 	\label{eq:Qc}
	\\[5mm]
	k_6  C_{\overline{-+-}}
	= k_8 C_{\overline{--+}}
	\qquad \Longrightarrow \qquad
	Q_4 = \frac{k_6}{k_8} =
          \frac{C_{\overline{--+}}}{C_{\overline{-+-}}}, 	\label{eq:Qd}
	\end{align}
	\label{eq:Q}
\end{subequations}
where $Q_i$, $i=1,\dots,4$, denote the equilibrium constants for each
reversible reaction.

\section{Closure Approximations}
\label{sec:closure}

In this section we discuss two strategies for closing system
\eqref{eq:dC}, by which we mean expressing the concentration of
3-clusters on the right-hand side (RHS) of this system in terms of a
suitable function of the concentrations of 2-clusters. In other words,
the goal is to replace each of the triplet concentrations $C_i$,
$i \in \Theta$, in \eqref{eq:dC} with suitably chosen functions
$g_i(C_+,C_-,C_{++},C_{--},C_{+-})$, such that the closed system will
have the form
\begin{subequations}
	\begin{align}
	\begin{split}
	\frac{d}{dt} C_{++}
	= 4k_1 g_{\widehat{+-+}}(C_+,C_-,C_{++},C_{--},C_{+-})
	+ 2k_2 g_{\overline{+-+}}(C_+,C_-,C_{++},C_{--},C_{+-})\\
	- 4k_3 g_{\widehat{++-}}(C_+,C_-,C_{++},C_{--},C_{+-})
	- 2k_4 g_{\overline{++-}}(C_+,C_-,C_{++},C_{--},C_{+-}),
	\label{eq:dC++c}
	\end{split}
	\end{align}
	\begin{align}
	\begin{split}
	\frac{d}{dt} C_{--}
	= 4k_5 g_{\widehat{-+-}}(C_+,C_-,C_{++},C_{--},C_{+-})
	+ 2k_6  g_{\overline{-+-}}(C_+,C_-,C_{++},C_{--},C_{+-})\\
	- 4k_7 g_{\widehat{--+}}(C_+,C_-,C_{++},C_{--},C_{+-})
	- 2k_8 g_{\overline{--+}}(C_+,C_-,C_{++},C_{--},C_{+-}),
	\label{eq:dC--c}
	\end{split}
	\end{align}
	\begin{align}
	\begin{split}
	\frac{d}{dt} C_{+-}
	= 2k_3 g_{\widehat{++-}}(C_+,C_-,C_{++},C_{--},C_{+-})
	+ 2k_7 g_{\widehat{--+}}(C_+,C_-,C_{++},C_{--},C_{+-})\\
	+ k_4 g_{\overline{++-}}(C_+,C_-,C_{++},C_{--},C_{+-})
	+ k_8 g_{\overline{--+}}(C_+,C_-,C_{++},C_{--},C_{+-})\\
	- 2k_1 g_{\widehat{+-+}}(C_+,C_-,C_{++},C_{--},C_{+-})
	- 2k_5 g_{\widehat{-+-}}(C_+,C_-,C_{++},C_{--},C_{+-})\\
	- k_2 g_{\overline{+-+}}(C_+,C_-,C_{++},C_{--},C_{+-})
	- k_6  g_{\overline{-+-}}(C_+,C_-,C_{++},C_{--},C_{+-}) ~.
	\end{split}
	\label{eq:dC+-c}
	\end{align}
	\label{eq:dCc}
\end{subequations}
The first approach to finding these functions is the pair
approximation based on the classical method introduced in
\cite{Ben-Avraham1992} and the second is a new optimal closure
approximation. The problem of finding the rate constants
$k_1,\dots,k_8$ in \eqref{eq:dC} will be addressed in Section
\ref{sec:Bayes}.

\subsection{Pair Approximation}
\label{sec:pair}

The pair approximation is a classical approach to closing truncated
hierarchical dynamical systems. It was first used by Dickman
\cite{Dickman1986} in a surface-reaction model and later by Matsuda et
al.~\cite{Matsuda1992} for a structured lattice appearing in a
population dynamics problem. In our model, we use the pair
approximation approach in order to close the dynamical system
\eqref{eq:dC} at the level of $2$-clusters. The state of a site is
denoted $i,j,k \in \{+,-\}$ for a two-element system. Global
concentrations are denoted $C_i$ giving the probability that a
randomly chosen site in the lattice is in state $i \in
\{+,-\}$. Similarly, {$C_{ij}$} is the global concentration of
2-clusters in state {$ij$}. In addition, {\em local}
concentrations are denoted {$P_{j|i}$} and give the conditional
probability that a randomly chosen nearest neighbour of a site in
state {$i$} is in state {$j$}.  These local concentrations
can be expressed in terms of global concentrations using the rules
governing conditional probabilities as \cite{Matsuda1992,Sato2000}
\begin{subequations}
  \begin{align}
    C_{ij} = C_{ji} = C_i P_{j|i} = C_{j} P_{i|j},\label{eq:S/S'a}\\
    \sum_{i\in\{+,-\}}^{} C_i = 1, \\
    \sum_{i\in\{+,-\}}^{} P_{i|j} = 1 ~~ for ~ any ~ j \in \{+,-\}.
  \end{align}
  \label{eq:S/S'}
\end{subequations}
Equation \eqref{eq:S/S'a} is invariant with respect to the rotational
symmetries of the lattice, cf.~Appendix \ref{sec:sym}. Also, the
global concentration of a triplet in state $(ijk)$ can be derived in a
similar approach as Equation \eqref{eq:S/S'a},
\begin{equation}
	\begin{aligned}
	C_{ijk} = C_i P_{j|i} P_{k|ij} = C_{ij} P_{k|ij}.
	\end{aligned}
	\label{eq:ijk}
\end{equation}
The $P_{k|ij}$ term in this equation involves $3$ elements in a
triplet.  In order to break down the triplet concentration in terms of
pair and singlet concentrations, one is required to find an equivalent
expression for the $P_{k|ij}$ term. The underlying assumption of the
pair approximation method is to neglect the interaction between the
non-nearest neighbour elements, $i$ and $k$ in this case, according to
Figure \ref{fig:d2} \cite{Matsuda1992,Harada1994,Sato2000}.  This
results in an approximation at the level of 3-clusters expressed in
terms of quantities defined at the level of 2-clusters as
\begin{equation}
	P_{k|ij} \approx P_{k|j}.
	\label{eq:PA_assumption}
\end{equation}
A different approach could also be adopted to derive the pair
approximation formulation resulting in the same closure model. In this
approach, assuming a triplet in state $(ijk)$ on a random lattice (in
which all non-nearest-neighbour elements are decoupled), the global
concentration of this triplet can be written as
\begin{subequations}
\label{eq:Cijk}
\begin{align}
C_{ijk} & = C_i C_j C_k Q_{ij} Q_{jk} T_{ijk}, \label{eq:CijkA} \\
Q_{ij} & = \frac{C_{ij}}{C_i C_j},
\end{align}
\end{subequations}
where $C_i$, $C_j$ and $C_k$ denote the global concentrations of
singlets, $Q_{ij}$ and $Q_{ik}$ are the pair correlations of nearest
neighbours and $T_{ijk}$ is the triple correlation of the chain. Note
that element $i$ and element $k$ on a random lattice are considered
not to be nearest-neighbours. Also, there is no factor $Q_{ik}$ in
equation \eqref{eq:CijkA} as the correlation of non-nearest-neighbours
is represented by $T_{ijk}$. According to the underlying assumption of
pair approximation, the non-nearest-neighbour elements are decoupled.
There is no deterministic way of calculating correlations of
non-nearest neighbour elements \cite{Baalen2000} and some additional
assumptions have to be made in order to close
\eqref{eq:dC}. The standard pair approximation method neglects all
triple correlations such that $T_{ijk} = 1$. {This is an
  equivalent approximation to Equation \eqref{eq:PA_assumption}.}

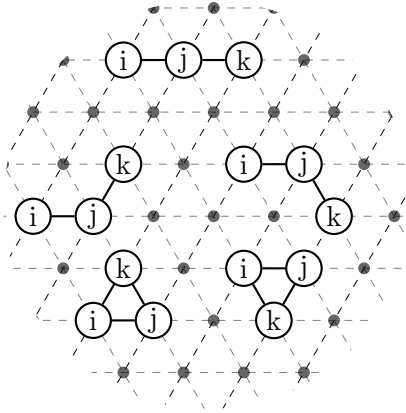
\begin{figure}
  \centering
  \begin{tikzpicture} [->,>=stealth',auto,thick,main node/.style={circle,draw,font=\sffamily\Large\bfseries},scale=0.8]
    \coordinate (Origin)   at (0,0);
    \coordinate (XAxisMin) at (-5,0);
    \coordinate (XAxisMax) at (5,0);
    \coordinate (YAxisMin) at (0,-5);
    \coordinate (YAxisMax) at (0,5);
    \clip (-3.6,-2.4) circle (3.4cm); 
    \begin{scope}
      \pgftransformcm{1}{0}{1/2}{sqrt(3)/2}{\pgfpoint{0cm}{0cm}}
      \draw[style=help lines,dashed,draw=black] (-12,-12) grid[ystep=0cm,xstep=1cm] (12,12);
    \end{scope}
    
    \pgftransformcm{1}{0}{-1/2}{sqrt(3)/2}{\pgfpoint{0cm}{0cm}} 
    \draw[style=help lines,dashed] (-12,-12) grid[step=1cm] (12,12);
    \foreach \x in {-6,-5.5,...,1}{
      \foreach \y in {-6,-5.5,...,1}{
        \node[draw = black, draw opacity = 0.6, circle,inner sep=1.3pt,fill=black, fill opacity=0.6] at (2*\x,2*\y) {};
      }
    }
    \node(1)[draw = black,circle,inner sep=2pt,fill=white] at (-5,0) {i};
    \node(2)[draw = black,circle,inner sep=1.3pt,fill=white] at (-4,0) {j};
    \node(3)[draw = black,circle,inner sep=1.4pt,fill=white] at (-3,0) {k};
    \draw [thick,-,black] (1) -- (2) node {};
    \draw [thick,-,black] (2) -- (3) node {};
    
    \node(1)[draw = black,circle,inner sep=2pt,fill=white] at (-8,-3) {i};
    \node(2)[draw = black,circle,inner sep=1.3pt,fill=white] at (-7,-3) {j};
    \node(3)[draw = black,circle,inner sep=1.4pt,fill=white] at (-6,-2) {k};
    \draw [thick,-,black] (1) -- (2) node {};
    \draw [thick,-,black] (2) -- (3) node {};
    
    \node(1)[draw = black,circle,inner sep=2pt,fill=white] at (-4,-2) {i};
    \node(2)[draw = black,circle,inner sep=1.3pt,fill=white] at (-3,-2) {j};
    \node(3)[draw = black,circle,inner sep=1.4pt,fill=white] at (-3,-3) {k};
    \draw [thick,-,black] (1) -- (2) node {};
    \draw [thick,-,black] (2) -- (3) node {};
    
    \node(1)[draw = black,circle,inner sep=2pt,fill=white] at (-8,-5) {i};
    \node(2)[draw = black,circle,inner sep=1.3pt,fill=white] at (-7,-5) {j};
    \node(3)[draw = black,circle,inner sep=1.4pt,fill=white] at (-7,-4) {k};
    \draw [thick,-,black] (1) -- (2) node {};
    \draw [thick,-,black] (2) -- (3) node {};
    \draw [thick,-,black] (1) -- (3) node {};
    
    \node(1)[draw = black,circle,inner sep=2pt,fill=white] at (-5,-4) {i};
    \node(2)[draw = black,circle,inner sep=1.3pt,fill=white] at (-4,-4) {j};
    \node(3)[draw = black,circle,inner sep=1.4pt,fill=white] at (-5,-5) {k};
    \draw [thick,-,black] (1) -- (2) node {};
    \draw [thick,-,black] (2) -- (3) node {};
    \draw [thick,-,black] (1) -- (3) node {};
    
  \end{tikzpicture}
  \caption{Schematic of a 2D triangular lattice with chains of
    $3$-clusters with $180$-degree bonds, $120$-degree bonds, and
    $60$-degree bonds.  These cluster types are referred to,
    respectively, as linear, angled and triangular throughout this
    text. The clumping intensity of this lattice is equal to the
    proportion of the triangles over all triplets types, which is
    equal to $\frac{2}{5}$.}
  \label{fig:3clusters}
\end{figure}

Each regular lattice can be described by two parameters: the number of
neighbours per site ($m$) and the proportion of triangles to triplets
($\theta$), which determines the clumping intensity of the lattice. A
triangular lattice has $m=6$ neighbours per site and $\theta =
\frac{2}{5}$, as shown in Figure \ref{fig:3clusters}. Similarly,
chain-like triplets in a triangular lattice can be categorized into
two groups: linear triplets with 180-degree bonds, and angled triplets
with 120-degree bonds. As is evident from Figure \ref{fig:3clusters},
the probability of finding a triplet in a closed form, angled form and
linear form is equal to $\frac{2}{5}$, $\frac{2}{5}$, and
$\frac{1}{5}$, respectively. As the shape of triplets is important in
our model, these probabilities have to be taken into account as
coefficients when calculating the corresponding concentrations. Morris
\cite{MorrisAndrew1997} and Keeling \cite{Keeling1997} have proposed
formulas for approximating the fraction of closed and open chains in a
certain state $(ijk)$ on a regular lattice by taking into account the
clumping effect of triangles in the lattice. Following these studies,
the concentrations of each type of triplet are approximated as
\begin{subequations}
  \begin{align}
    C_{\overline{ijk}} & \approx g_{\overline{ijk}} = (1-\theta) \frac{1}{3} \frac{C_{ij} C_{jk}}{C_{j}},\label{eq:Cijk2a}\\
    C_{\widehat{ijk}} &  \approx g_{\widehat{ijk}}  = (1-\theta) \frac{2}{3} \frac{C_{ij} C_{jk}}{C_{j}},\label{eq:Cijk2b}\\
    C_{\widetriangle{ijk}} & \approx g_{\widetriangle{ijk}} = \theta \frac{C_{ij} C_{jk} C_{ki}}{C_{i} C_{j} C_{k}},\label{eq:Cijk2c}
  \end{align}
  \label{eq:Cijk2}%
\end{subequations}
where $\overline{ijk}$ denotes a linear cluster, $\widehat{ijk}$
denotes an angled triplet with 120-degree bonds and $\widetriangle{ijk}$
denotes a triangular cluster with 60-degree bonds. The specific forms
taken by expressions \eqref{eq:Cijk2a}--\eqref{eq:Cijk2b} for
different $i,j,k \in \{+,-\}$ are collected in Table
\ref{tab:closures}. Applying this pair approximation to close system
\eqref{eq:dC} gives
\begin{subequations}
\begin{align} 
\label{eq:dCclosed}
  \begin{aligned}
    \frac{d}{dt}  C_{++}
    = & \phantom{-} 4k_1 \frac{2}{3} (1-\theta) \frac{C_{+-}^2}{C_{-}}
    + 2k_2 \frac{1}{3} (1-\theta) \frac{C_{+-}^2}{C_{-}}\\
    & - 4k_3 \frac{2}{3} (1-\theta) \frac{C_{++} C_{+-}}{C_{+}}
    - 2k_4 \frac{1}{3} (1-\theta) \frac{C_{++} C_{+-}}{C_{+}},
  \end{aligned} \\
  \begin{aligned}
    \frac{d}{dt}  C_{--}
    =  & \phantom{-} 4k_5 \frac{2}{3} (1-\theta) \frac{C_{+-}^2}{C_{+}}
    + 2k_6 \frac{1}{3} (1-\theta) \frac{C_{+-}^2}{C_{+}}\\
    & - 4k_7 \frac{2}{3} (1-\theta) \frac{C_{--} C_{+-}}{C_{-}}
    - 2k_8 \frac{1}{3} (1-\theta) \frac{C_{--} C_{+-}}{C_{-}}.
  \end{aligned} 
\end{align}
\end{subequations}
We note that equation \eqref{eq:dC+-} is eliminated from the system of
equations as we use the normalization condition
$\frac{d}{dt}C_{++} + \frac{d}{dt} C_{--} + 2\frac{d}{dt} C_{+-}
  = 0$, cf.~\eqref{eq:dCCb}, to close the system.

\subsection{Optimal Approximation}
\label{sec:optimal}

As will be shown in Section \ref{sec:results_opt}, the closure based
on the pair approximation introduced above is not very accurate.  In
order to improve the accuracy of the closure, here we propose a new
approach based on nonlinear regression analysis of simulated annealing
data. This is a data-driven strategy where an {\em optimal} form of
the closure is obtained by fitting an expression in an assumed
well-justified form to the data. The pair approximation scheme
attempts to predict the concentrations of the higher-order clusters in
terms of concentrations of lower-order ones using expressions with the
functional forms given in \eqref{eq:Cijk2}. In the new approach, we
close system \eqref{eq:dC} using relations generalizing the
expressions in \eqref{eq:Cijk2} which depend on a number of adjustable
parameters.  These parameters, representing the exponents of different
concentrations, are then calibrated against the simulated annealing
data by solving a suitable constrained optimization problem.
Information about the new more general closure relations and how they
compare to the pair approximation for different 3-clusters is
collected in Table \ref{tab:closures} where we also group the
parameters to be determined in the vector $\mathbf{V}_i$, with
$i \in \Theta$ representing different cluster types.

\begin{table}[h!]
  \begin{ruledtabular}
    \begin{tabular}{ c  c  c  c }
      Triplet Type & Pair Approximation & Optimal Approximation  & Parameters (exponents)  \\
          $i$      & $g_i(C_+,\dots,C_{+-})$ & $g_i(C_+,\dots,C_{+-};\mathbf{V}_i)$ & $\mathbf{V}_i$ \\
      $\overline{+++}$      & $\frac{1}{5} \frac{C_{++}^{2} }{C_{+}}$ & $\frac{1}{5} \frac{C_{++}^{\gamma_1} }{C_{+}^{\xi_1}}$ & $\mathbf{V}_{\overline{+++}} = \left[\gamma_1~\xi_1\right]$ \\[3mm]
      $\overline{---}$      &  $\frac{1}{5} \frac{C_{--}^{2} }{C_{-}}$ & $\frac{1}{5} \frac{C_{--}^{\gamma_1} }{C_{-}^{\xi_1}}$ & $\mathbf{V}_{\overline{---}} = \left[\gamma_1~\xi_1\right]$\\[3mm]
      $\overline{+-+}$      & $\frac{1}{5} \frac{C_{+-}^{2}}{ C_{-}}$& $\frac{1}{5} \frac{C_{+-}^{\gamma_1}}{C_{+}^{\xi_1} C_{-}^{\xi_2}}$ & $\mathbf{V}_{\overline{+-+}} = \left[\gamma_1~\xi_1~\xi_2\right]$\\[3mm]
      $\overline{-+-}$      & $\frac{1}{5} \frac{C_{+-}^{2}}{C_{+}}$ & $\frac{1}{5} \frac{C_{+-}^{\gamma_1}}{C_{+}^{\xi_1} C_{-}^{\xi_2}}$ & $\mathbf{V}_{\overline{-+-}} = \left[\gamma_1~\xi_1~\xi_2\right]$\\[3mm]
      $\overline{++-}$      & $\frac{1}{5} \frac{C_{++} C_{+-}}{C_{+}}$ & $\frac{1}{5} \frac{C_{++}^{\gamma_1} C_{+-}^{\gamma_2}}{C_{+}^{\xi_1} C_{-}^{\xi_2}}$ & $\mathbf{V}_{\overline{++-}} = \left[\gamma_1~\gamma_2~\xi_1~\xi_2\right]$\\[3mm]
      $\overline{--+}$      & $\frac{1}{5} \frac{C_{--} C_{+-}}{C_{-}}$ & $\frac{1}{5} \frac{C_{--}^{\gamma_1} C_{+-}^{\gamma_2}}{C_{+}^{\xi_1} C_{-}^{\xi_2}}$ & $\mathbf{V}_{\overline{--+}} = \left[\gamma_1~\gamma_2~\xi_1~\xi_2\right]$ \\[3mm]
      $\widehat{+++}$      & $\frac{2}{5} \frac{C_{++}^2}{C_{+}}$ & $\frac{2}{5} \frac{C_{++}^{\gamma_1}}{C_{+}^{\xi_1}}$ & $\mathbf{V}_{\widehat{+++}} = \left[\gamma_1~\xi_1\right]$\\[3mm]
      $\widehat{---}$       &$\frac{2}{5} \frac{C_{--}^{2}}{C_{-}}$& $\frac{2}{5} \frac{C_{--}^{\gamma_1}}{C_{-}^{\xi_1}}$ & $\mathbf{V}_{\widehat{---}} = \left[\gamma_1~\xi_1\right]$\\[3mm]
      $\widehat{+-+}$       &$\frac{2}{5} \frac{C_{+-}^{2}}{C_{-}}$& $\frac{2}{5} \frac{C_{+-}^{\gamma_1}}{C_{+}^{\xi_1} C_{-}^{\xi_2}}$ & $\mathbf{V}_{\widehat{+-+}} = \left[\gamma_1~\xi_1~\xi_2\right]$\\[3mm]
      $\widehat{-+-}$      &$\frac{2}{5} \frac{C_{+-}^{2}}{C_{+}}$ & $\frac{2}{5} \frac{C_{+-}^{\gamma_1}}{C_{+}^{\xi_1} C_{-}^{\xi_2}}$  & $\mathbf{V}_{\widehat{-+-}} = \left[\gamma_1~\xi_1~\xi_2\right]$\\[3mm]
      $\widehat{++-}$      &$\frac{2}{5} \frac{C_{++} C_{+-}}{C_{+}}$ & $\frac{2}{5} \frac{C_{++}^{\gamma_1} C_{+-}^{\gamma_2}}{C_{+}^{\xi_1} C_{-}^{\xi_2}}$ & $\mathbf{V}_{\widehat{++-}} = \left[\gamma_1~\gamma_2~\xi_1~\xi_2\right]$\\[3mm]
      $\widehat{--+}$      & $\frac{2}{5} \frac{C_{--} C_{+-}}{C_{-}}$ & $\frac{2}{5} \frac{C_{--}^{\gamma_1} C_{+-}^{\gamma_2}}{C_{+}^{\xi_1} C_{-}^{\xi_2}}$ &  $\mathbf{V}_{\widehat{--+}} = \left[\gamma_1~\gamma_2~\xi_1~\xi_2\right]$\\[3mm]
      $\widetriangle{+++}$  &$\frac{2}{5} \frac{C_{++}^{3}}{C_{+}^{3}}$ & $\frac{2}{5} \frac{C_{++}^{\gamma_1}}{C_{+}^{\xi_1}}$ & $\mathbf{V}_{\widetriangle{+++}} = \left[\gamma_1~\xi_1\right]$\\[3mm]
      $\widetriangle{---}$  & $\frac{2}{5} \frac{C_{--}^{3}}{C_{-}^{3}}$& $\frac{2}{5} \frac{C_{--}^{\gamma_1}}{C_{-}^{\xi_1}}$ & $\mathbf{V}_{\widetriangle{---}} = \left[\gamma_1~\xi_1\right]$\\[3mm]
      $\widetriangle{++-}$  & $\frac{2}{5} \frac{C_{++} C_{+-}^{2}}{C_{+}^{2} C_{-}}$ & $\frac{2}{5} \frac{C_{++}^{\gamma_1} C_{+-}^{\gamma_2}}{C_{+}^{\xi_1} C_{-}^{\xi_2}}$ & $\mathbf{V}_{\widetriangle{++-}} = \left[\gamma_1~\gamma_2~\xi_1~\xi_2\right]$\\[3mm]
      $\widetriangle{--+}$  & $\frac{2}{5} \frac{C_{--} C_{+-}^{2}}{C_{+}C_{-}^{2}}$ & $\frac{2}{5} \frac{C_{--}^{\gamma_1} C_{+-}^{\gamma_2}}{C_{+}^{\xi_1} C_{-}^{\xi_2}}$ & $\mathbf{V}_{\widetriangle{--+}} = \left[\gamma_1~\gamma_2~\xi_1~\xi_2\right]$\\[3mm]
    \end{tabular}
    \caption{The functional forms of the closures based on the pair
      approximation and on the proposed optimal closures for each
      triplet type.  Unknown parameters (exponents) are indicated in
      the last column.}
    \label{tab:closures}
  \end{ruledtabular}
\end{table}

Notably, the new functional forms are generalizations of the
expressions used in the pair approximation obtained by allowing for
more freedom in how the new expressions for closures depend on the
cluster concentrations. The numerators of the new expressions involve
concentrations of {\em all} nearest-neighbour 2-clusters such that the
effect of non-nearest-neighbour clusters is still neglected. The
denominators, on the other hand, involve the concentrations of
singlets present in the triplet which makes the functional form of the
new closure different from the pair approximation in some cases.  The
parameters (exponents) defining the proposed optimal closures in Table
\ref{tab:closures} are subject to the following constraints ensuring
well-posedness of the resulting system \eqref{eq:dCc}
\begin{enumerate}
\item the difference of the sums of the exponents in the numerators
  and in the denominators is equal to one, i.e., $\sum_{j}^{}
    \gamma_j - \sum_{j}^{} \xi_j = 1$, ensuring that the terms
  representing the closure have the units of concentration,

\item the exponents in the numerators need to be non-negative,
    i.e., $\gamma_j \ge 0$, since otherwise the corresponding terms
  representing the closure model may become unbounded as the
  concentration approaches zero, causing solutions of the ODE system
  \eqref{eq:dCc} to blow up,

\item the exponents in the numerators need to be bounded
  $\gamma_1, \gamma_2 \le \delta$, where $\delta$ is the upper bound
  on the exponent which needs to be specified, as otherwise the
  corresponding terms representing the closure model may also become
  large causing solutions of the ODE system \eqref{eq:dCc} to blow up,

\item while denominators involve concentrations of singlets
    only, which are time independent, in some cases it is necessary to
    restrict the corresponding exponents as otherwise the terms
    representing the closure model will have large prefactors which
    may also cause the solutions of the ODE system \eqref{eq:dCc} to
    blow up; hence, we impose $\beta_1 \le \xi_1, \xi_2 \le \beta_2$,
    where $\beta_1$ and $\beta_2$ are the lower and upper bounds on
    the exponents to be specified; 
\end{enumerate}

Optimal parameters $\mathbf{V}_i $ of the closure model are obtained
separately for each cluster type $i$ by minimizing the mean-square
error between the experimental concentration data $\widetilde{C}_i(t)$
obtained from simulated annealing experiments, and the predictions of
the corresponding ansatz function
$g_i(\widetilde{C}_+,\widetilde{C}_-,\widetilde{C}_{++}(t),\widetilde{C}_{--}(t),\widetilde{C}_{+-}(t);\mathbf{V}_i)$,
cf.~Table \ref{tab:closures}, obtained with the parameter vector
$\mathbf{V}_i$ over the time window $[0,T]$, where $T$ corresponds to
the end of the simulated annealing process. Then, for each
$i\in \Theta$, error functional is defined as
\begin{equation}
\allowdisplaybreaks[1]
\begin{aligned}
J_i(\mathbf{V}_i) = \frac{1}{2} \int_{0}^{T} \left[g_i(\widetilde{C}_+,\widetilde{C}_-,\widetilde{C}_{++}(t),\widetilde{C}_{--}(t),\widetilde{C}_{+-}(t);\mathbf{V}_i) - \widetilde{C}_{i}(t) \right]^2 dt
\end{aligned}
\label{eq:Ji}
\end{equation}
which leads to the following family of constrained optimization
problems
\begin{equation}
\allowdisplaybreaks[1]
\begin{aligned}
& \underset{\mathbf{V}_i}{\min} ~~ J_i(\mathbf{V}_i),\\
& \text{subject to:} ~~~
\begin{cases}
0 \le \gamma_j \le \delta, & 1 \le j \le \Gamma_i \\
\beta_1 \le \xi_j \le \beta_2, &  1 \le j \le \Xi_i  \\
\sum_{j}^{} \gamma_j - \sum_{j}^{} \xi_j = 1
\end{cases}, 
\end{aligned}
\label{eq:min}
\end{equation}
for each $i\in \Theta$, where $\Gamma_i, \Xi_i \in \{ 1,2 \}$ are the
numbers of the exponents appearing in the numerator and the
denominator for a given cluster type, cf.~Table \ref{tab:closures}.

We note that choosing different values of the adjustable parameters
$\delta$, $\beta_1$ and $\beta_2$, which determine how stringent
the constraints in the optimization problem \eqref{eq:min} are, has the
effect of regularizing the solutions of this problem. We will
consider the following two cases (when the lower/upper bound is equal
to $-\infty / \infty$, this means that effectively there is no bound)
\begin{itemize}
\item ``soft'' regularization with $\beta_1 = -\infty$, $\beta_2 =
  \infty$, $\delta = 6$, and

\item ``hard'' regularization with $\beta_1=0$ and $\beta_2 = \delta = 2$.
\end{itemize}
In each case optimization problem \eqref{eq:min} is solved numerically
in MATLAB using the nonlinear programming routine {\tt fmincon}.  The
optimal closures determined in these two ways are compared to the pair
approximation in Section \ref{sec:results_opt}.

\section{Determining Reaction Rates via Bayesian Inference}
\label{sec:Bayes}

In order for the truncated model \eqref{eq:dCc} closed with either the
pair or optimal approximation to predict the time evolution of
2-cluster concentrations, it must be equipped with correct values of
the rate constants $k_1,\dots,k_8$, cf.~Figures \ref{fig:++} and
\ref{fig:--}. Here we show how these constants can be determined by
solving an appropriate inverse problem. It will be demonstrated that
this problem is in fact ill-posed and a suitable solution will be
obtained using Bayesian inference which also provides information
about the uncertainty of this solution.

We define the error functional as
\begin{equation}
\mathscr{J}(\mathbf{K}) = \frac{1}{2} \int_{0}^{T} \big\|\mathbf{C}(t,\mathbf{K}) - \widetilde{\mathbf{C}}(t) \big\|_2^2 \, dt + \alpha \big\| \mathbf{Q}(\mathbf{K}) - \widetilde{\mathbf{Q}} \big\|_2^2,
\label{eq:J}
\end{equation}
where $\widetilde{\mathbf{C}}(t) = \left[
\widetilde{C}_{++}(t),\widetilde{C}_{--}(t),\widetilde{C}_{+-}(t)
\right]$ is the vector of pair concentrations obtained from the
simulated annealing experiment, cf.~Figure \ref{fig:SA}c, 
$\mathbf{K} = \left[ k_1,k_2,\cdots,k_8 \right]$ is the vector
of unknown rate constants, and
$\mathbf{C}(t,\mathbf{K})$ is the vector of pair concentrations
predicted by model \eqref{eq:dCc} equipped with the rate constants
$\bK$. The second term in \eqref{eq:J} is the mean-square error
between the equilibrium constants $\bQ(\bK) = [Q_1, Q_2, Q_3, Q_4]$,
cf.~relation \eqref{eq:Q}, predicted by model \eqref{eq:dCc} equipped
with parameters $\bK$ and the equilibrium constant $\widetilde{\bQ} =
[\widetilde{Q}_1, \widetilde{Q}_2, \widetilde{Q}_3, \widetilde{Q}_4 ]$
obtained experimentally via simulated annealing. We note that the
equilibrium constants in \eqref{eq:Q} are written in terms of
3-cluster concentrations and one of the closure models (i.e., the pair
or the optimal approximation) is used to express the equilibrium
constants in terms of 2-cluster concentrations. The parameter $\alpha$
weights the relative importance of matching the equilibrium constants
versus matching the time-dependent concentrations in \eqref{eq:J}.

The optimal reaction rates are then obtained by solving the problem
\begin{equation}
\allowdisplaybreaks[1]
\begin{aligned}
& \underset{\bK \in \RR^8}{\min \mathscr{J}(\bK)} \\
& \text{subject to system \eqref{eq:dCc}}
\end{aligned}
\label{eq:minJ}
\end{equation}
separately for the case of the pair and the optimal approximations. We
note that the minimization problems \eqref{eq:min} and \eqref{eq:minJ}
are in fact quite different: in the former the mismatch between the
evolution of 3-cluster concentrations is minimized with respect to a
suitably-parameterized structure of the closure model, whereas in the
latter one seeks to minimize the mismatch between the evolution of
$2$-clusters in order to find the optimal reaction rates in the closed
system \eqref{eq:dCc}.

Inverse problems such as \eqref{eq:minJ} are often ill-posed, in the
sense that they usually do not admit a unique exact solution, but
rather many, typically infinitely many, approximate solutions. This is
a result of the presence of multiple local minima, which is a
consequence of the non-convexity of the error functional \eqref{eq:J},
and the fact that these minima are often ``shallow'' reflecting weak
dependence of the model predictions $\bC$ on the parameters $\bK$. As
will be evident from the results presented in Section
\ref{sec:results}, it is thus not very useful to solve problem
\eqref{eq:minJ} directly using standard methods of numerical
optimization \cite{nw00}. Instead, we will adopt a probabilistic
approach based on Bayesian inference where the unknown parameters in
the vector $\bK$ and the corresponding model predictions $\bC$ will be
represented in terms of suitable conditional probability
densities. This will allow us to systematically assess the relative
uncertainty of the many approximate solutions admitted by problem
\eqref{eq:minJ}. The mathematical foundations of Bayesian inference
are reviewed in the monographs
\cite{t05,KaipioSomersalo2005,Smith2013}.

In the Bayesian framework the distribution of the model parameters is
given by the {\em posterior} probability distribution
$\mathbb{P}\left(\textbf{K}|\widetilde{\textbf{C}}\right)$ defined as
the probability of obtaining parameters $\bK$ given the observed
experimental data $\tbC$. According to Bayes' rule, we then have
\begin{equation}
\mathbb{P}\left(\textbf{K}|\widetilde{\textbf{C}}\right) = \frac{\mathbb{P}\left(\widetilde{\textbf{C}}|\textbf{K}\right) \mathbb{P}\left(\textbf{K}\right)}{\mathbb{P}\left(\widetilde{\textbf{C}}\right)},
\label{eq:Bayes}
\end{equation}
where $\PP\left(\widetilde{\textbf{C}}|\textbf{K}\right)$ is the
\textit{likelihood} function describing the likelihood of obtaining
observations $\tbC$ given the model parameters $\bK$,
$\PP\left(\bK\right)$ is the {\em prior} probability distribution
reflecting some a priori assumptions on the parameters $\bK$ (based,
e.g., on direct measurements or literature data), whereas
$\PP\left(\bC\right)$ can be viewed as a normalizing factor.

A common approach to choosing the prior distribution
$\PP\left(\bK\right)$ is to use an uniform distribution, leading to
the so-called uninformative prior, and this is the approach we adopt
here. As regards the likelihood function, it is usually defined as
\begin{equation}
\PP\left(\widetilde{\textbf{C}}|\textbf{K}\right) \propto e^{-\mathscr{J}\left(\textbf{K}\right)}.
\label{eq:lhood}
\end{equation}
This definition of the likelihood function arises from the fact that
parameter values are considered more likely if they produce model
predictions $\bC$ closer to the data $\tbC$. Moreover, if the error
functional is a quadratic function of the model parameters $\bK$, then
the distribution in \eqref{eq:lhood} is Gaussian.

The main challenge is efficient sampling of the likelihood function
$\PP\left(\widetilde{\textbf{C}}|\textbf{K}\right)$ and this can be
performed using a Markov-Chain Monte-Carlo (MCMC) approach. It is a
form of a random walk in the parameter space designed to 
preference the sampling of
high-likelihood regions of the space while also exploring other
regions. MCMC methods are commonly used to sample arbitrary
distributions known up to a normalizing factor. In particular, these
methods are used to sample distributions in high dimensions where
exploration of the entire space with classical methods is
computationally intractable.  MCMC techniques have found applications
in many different fields such as electrochemistry
\cite{Sethurajan2019}, medical imaging \cite{Zhou2019,Huo2020},
environmental and geophysical sciences \cite{Laine2008,Lucas2017} and
ecology \cite{Camli2020}.

In the MCMC algorithm, a kernel $\mathcal{Q}
\left(\textbf{K}^*|\textbf{K}\right)$ is used to generate a proposal
for a move in the parameter space from the current point $\bK$ to a
new point $\textbf{K}^*$. This new point is accepted with a
probability given by the Hastings ratio; otherwise, it is rejected
(the ``Metropolis rejection''). In order to preserve the reversibility
of the Markov chain, the Hastings ratio for the acceptance probability
is defined as
\begin{equation}
\alpha \left(\textbf{K}^*,\textbf{K}\right) = \min \left\{1,\frac{\PP\left(\textbf{K}^*|\widetilde{\textbf{C}}\right) \mathcal{Q}\left(\textbf{K}|\textbf{K}^*\right) }{\PP \left(\textbf{K}|\widetilde{\textbf{C}} \right) \mathcal{Q}\left(\textbf{K}^*|\textbf{K}\right)} \right\}.
\label{eq:alpha}
\end{equation}
Thus, the Markov chain is reversible with respect to the
posterior distribution, meaning that a transition in space is equally
probable during forward and backward evolution. This
property makes the posterior distribution invariant on the Markov
chain. In other words, if given enough iterations, the distribution
converges to its equilibrium distribution. The most common choice of
the random walk is in the form
\begin{equation}
\textbf{K}^* = \textbf{K} + \boldsymbol{\xi}
\end{equation}
such that
$\mathcal{Q}\left(\textbf{K}^*|\textbf{K}\right) = \mathcal{Q}\left(
  \textbf{K}^* - \textbf{K} \right) = \mathcal{Q}\left(
  \boldsymbol{\xi} \right)$, where $\boldsymbol{\xi}$ is an
8-dimensional random variable drawn from a uniform distribution with
scale $\boldsymbol{\sigma} \in \RR^8$, i.e.,
$\boldsymbol{\xi} \sim \mathcal{U}\left[
  -\boldsymbol{\sigma},\boldsymbol{\sigma} \right]$. Note that the
components of the scale $\boldsymbol{\sigma}$ represent intervals
defining the uniform distribution. It has been suggested that uniform
kernels outperform Gaussian ones in terms of convergence of the MCMC
algorithm \cite{Thawornwattana2018}, hence, we adopt the uniform
kernel in our study.  The choice of symmetric kernels simplifies
relation \eqref{eq:alpha} as the factors representing the density in
the numerator and denominator cancel. However, the choice of scale for
the proposal kernel is nontrivial. Small scales will result in slow
convergence to the posterior distribution, whereas large scales will
prevent sampling of desirable regions in the parameter space.
Moreover, in our model there is no prior information about an
appropriate scale for the proposal kernel. In order to tackle this
issue, a two-step Delayed-Rejection Metropolis-Hastings (DR-MH)
algorithm is used \cite{Bedard2014,Laine2008,Trust2016}. In this
algorithm, the rejection of the first proposed point at a given
iteration of the Markov chain is delayed by proposing a new step in
the space based on a different scale. Normally, the scale of the first
kernel is chosen to be large in order to explore a wider region of the
high-dimensional parameter space and the scale of the second kernel is
small to gather more samples from higher-likelihood regions. This
approach combines exploration of large regions in a high dimensional
space with focus on high-likelihood neighbourhoods.
The DR-MH algorithm also ensures the reversibility of the Markov
chain, meaning that the direction of time in which the random walk is
taking place does not affect the dynamics of the Markov chain. In
other words, a random walk in the forward direction of the chain from
state $n$ to state $n+1$ is equally probable as the reverse walk from
state $n+1$ to state $n$. This ensures that the chain remains in an
equilibrium state as it evolves. This is an important property as the
Markov chain is essentially a random walk in the posterior space and
reversibility is required to ensure it remains in the same posterior
space. The acceptance probability of the delayed proposed point is
calculated using relation \eqref{eq:20}. To initialize the DR-MH
algorithm, we require an initial set of model parameters. They could
be random, without any prior information about the parameters. As an
alternative, we determine this initial point
$\boldsymbol{\mathcal{Y}}_1$ by solving problem \eqref{eq:minJ} using
a standard numerical optimization method \cite{nw00} several times
with random initial guesses and then taking the best parameter set
corresponding to the lowest value of the cost functional
$\mathscr{J}(\bK)$.  Algorithm \ref{alg:MCMC} outlines the entire
procedure needed to approximate the posterior probability distribution
$\PP\left(\textbf{K}|\widetilde{\textbf{C}}\right)$.  Additional
details concerning MCMC approaches can be found in monographs
\cite{Kaipo2006,Robert2004}.
\begin{widetext}
  \begin{algorithm}[H]
    \SetAlgoLined
    \KwIn{~~$M$ \textbf{---} Number of samples to be drawn from the posterior distribution\\
~~~~~~~~~~~~~$\mathcal{Q}_1\left(\textbf{K}^*|\textbf{K}\right)$ \textbf{---} Proposal density of the first trial\\
~~~~~~~~~~~~~$\boldsymbol{\mathcal{Y}}_1$ \textbf{---} initial point for the random walk in the space $\mathbb{R}^8$\\
~~~~~~~~~~~~~$\boldsymbol{\sigma}_1, \boldsymbol{\sigma}_2$ \textbf{---} scales defining the random walk}
    \KwOut{$\PP\left(\textbf{K}|\widetilde{\textbf{C}}\right)$ \textbf{---} Posterior probability distribution}
    \SetAlgoVlined
    \vspace{0.2cm}\hrule\vspace{0.2cm}
    $n \leftarrow 1$\\
    \Repeat(){$M$ samples are drawn}{
      Propose a step: $\boldsymbol{\xi} \sim \mathcal{U}\left[ -\boldsymbol{\sigma}_1,\boldsymbol{\sigma}_1 \right] $\\
      Propose a candidate: $\boldsymbol{\mathcal{Y}}_1 = \textbf{K}^{n-1} + \boldsymbol{\xi} $\\
      Accept the proposed step with probability $\alpha_1$:
      \begin{equation}
        \begin{split}
          \alpha_1 \left(\boldsymbol{\mathcal{Y}}_1,\textbf{K}^{n-1}\right) \propto \min \Bigg\{ 1,\frac{\exp(-\mathscr{J}\left(\boldsymbol{\mathcal{Y}}_1\right))}{\exp(-\mathscr{J}\left(\textbf{K}^{n-1}\right))} \Bigg\}
        \end{split}
      \end{equation}\\
      Draw a random number: $r \sim \mathcal{U}\left[0,1\right] $\\
      \eIf(){$\alpha_1\left(\boldsymbol{\mathcal{Y}}_1,\textbf{K}^{n-1}\right)  < r $}{
        Propose a new step with scale $\boldsymbol{\sigma}_2$: $\boldsymbol{\xi} \sim \mathcal{U}\left[ -\boldsymbol{\sigma}_2,\boldsymbol{\sigma}_2 \right] $\\
        Propose a new candidate:  $\boldsymbol{\mathcal{Y}}_2 = \textbf{K}^{n-1} + \boldsymbol{\xi} $\\
        Accept the new proposed point with probability $\alpha_2$:
        \begin{equation}
          \alpha_2 \left(\textbf{K}^{n-1},\boldsymbol{\mathcal{Y}}_1,\boldsymbol{\mathcal{Y}}_2\right) = \min\Bigg\{1, \frac{\PP \left(\boldsymbol{\mathcal{Y}}_2|\widetilde{\textbf{C}}\right) \mathcal{Q}_1\left(\boldsymbol{\mathcal{Y}}_2|\boldsymbol{\mathcal{Y}}_1\right)\left[\max\left(0, 1-\frac{\PP\left(\boldsymbol{\mathcal{Y}}_1|\widetilde{\textbf{C}}\right)}{\PP\left(\boldsymbol{\mathcal{Y}}_2|\widetilde{\textbf{C}}\right)}\right)\right]}{\PP\left(\ \textbf{K}^{n-1} |\widetilde{\textbf{C}}\right) \mathcal{Q}_1\left(\boldsymbol{\mathcal{Y}}_1|\textbf{K}^{n-1}\right)\left[\max\left(0, 1-\frac{\PP\left(\boldsymbol{\mathcal{Y}}_1|\widetilde{\textbf{C}}\right)}{\PP\left(\textbf{K}^{n-1}|\widetilde{\textbf{C}}\right)}\right)\right]}\Bigg\}
          \label{eq:20}
        \end{equation}\\
        Draw a random number: $r \sim \mathcal{U}\left[0,1\right] $\\
        \eIf(){$\alpha_2\left(\textbf{K}^{n-1},\boldsymbol{\mathcal{Y}}_1,\boldsymbol{\mathcal{Y}}_2\right) > r $}{$\textbf{K}^n \leftarrow \boldsymbol{\mathcal{Y}}_2$\\ $n \leftarrow n+1$}(){discard $\boldsymbol{\mathcal{Y}}_2$}
      }(){$\textbf{K}^n \leftarrow \boldsymbol{\mathcal{Y}}_1$\\ $n \leftarrow n+1$}
      
    }()
    Construct the posterior probability distribution
    \caption{ \textsc{Random walk delayed rejection algorithm} }
    \label{alg:MCMC}
  \end{algorithm}
\end{widetext}

In our computations reported in Section \ref{sec:results_rates} we
employ an uniform positive prior $\PP\left(\textbf{K}\right)$, whereas
the scales of the first and second trial of the two-step delayed
rejection algorithm are defined as the initial guess
$\boldsymbol{\mathcal{Y}}_1$ multiplied by a factor of 0.1 and 0.01,
respectively. The total number of
samples in the Markov chain is $M = 10^5$.

\section{Results}
\label{sec:results}

In this section we first determine the optimal structure of the
closure models given in Table \ref{tab:closures} by solving
optimization problem \eqref{eq:min} for each type of 3-cluster in the
set $\Theta$, cf.~\eqref{eq:Theta}, as described in Section
\ref{sec:optimal}. Then, based on these results, in Section
\ref{sec:interpret} we propose a new closure model which we refer to
as \textit{Sparse Approximation} (SA) and in Section
\ref{sec:predict} we assess the predictive capability of the
considered models by analyzing how accurately they predict the time
evolution of 3-cluster concentrations for a range of different
stoichiometries. Finally, we determine the reaction rates in the
truncated model \eqref{eq:dCc} closed with the pair approximation,
optimal approximation and sparse approximation using
Bayesian inference to solve problem \eqref{eq:minJ}, as described in
Section \ref{sec:Bayes}.

\subsection{Optimal Closures}
\label{sec:results_opt}

Parameters of the closure relations given in Table \ref{tab:closures}
are determined separately for each cluster type by solving problem
\eqref{eq:min} and the obtained results are collected in the form of
the values of the exponents in Table \ref{tab:exp}, where, for
comparison we also show the exponents corresponding to the pair
approximation, cf.~Section \ref{sec:pair}. We recall that for each
3-cluster type problem \eqref{eq:min} is solved with both soft and
hard regularization. In Table \ref{tab:exp}, the optimal results are
presented for solving problem \eqref{eq:min} subject to hard
regularization ($\beta_1 = 0$, $\delta=\beta_2=2$) by separately
fitting the closure models to the data obtained for two systems with
Li$_{1/2}$Mn$_{1/2}$ and Li$_{1/3}$Mn$_{2/3}$.  The first system is
interesting since, as we shall see below, due to the symmetry in the
concentrations of Li and Mn, closure models calibrated based on the
data from this system are particularly robust with respect to
different stoichiometries. The second system is considered in our
analysis due to its interesting behaviour at low temperatures where
physically relevant crystalline microstructure are obtained, as
discussed in Section \ref{sec:problem}. This system is also used as a
benchmark in \cite{Harris2017}. In Table \ref{tab:exp} we note that
most of the exponents in the optimal closure approximation tend to be
different from the corresponding exponents in the pair
approximation. Interestingly, we observe that many exponents obtained
for the optimal closure by fitting to the data for the system
Li$_{1/2}$Mn$_{1/2}$ are equal to zero or one, opening the possibility
of finding a simpler closure model to be investigated in Section
\ref{sec:interpret}.

The accuracy of representing the concentrations of 3-clusters based on
2-cluster concentrations is investigated for different closure
approximations in Figure \ref{fig:OA3OA2} for 4 representative triplet
types, namely, $\overline{+++}$, $\overline{-+-}$, $\widehat{---}$ and
$\widetriangle{---}$.  A significant improvement is evident for most
cluster types when the optimal closure is used. This is confirmed in
quantitative terms in Figure \ref{fig:OAPA_error} showing the
mean-square error \eqref{eq:Ji} for each $3$-cluster type for the pair
approximation and the optimal closure fitted to Li$_{1/2}$Mn$_{1/2}$
and Li$_{1/3}$Mn$_{2/3}$ systems.  For both systems and for almost all
3-cluster types the optimal closure leads to a more accurate
description with errors \eqref{eq:Ji} smaller by a few orders of
magnitude than when the pair approximation is used. In the next
section we will simplify the obtained optimal closure and will propose
an interpretation of the resulting structure.

\begin{table}[h!]
  \begin{ruledtabular}
    \begin{tabular}{ c c c c c c c c c c c c c}
      \multirow{2}{*}{Triplet Type}& PA & OA-1/3 & OA-1/2 & PA & OA-1/3 & OA-1/2 & PA & OA-1/3 & OA-1/2 & PA & OA-1/3 & OA-1/2\\
      \cline{2-2}
      \cline{3-3}
      \cline{4-4}
      \cline{5-5}
      \cline{6-6}
      \cline{7-7}
      \cline{8-8}
      \cline{9-9}
      \cline{10-10}
      \cline{11-11}
      \cline{12-12}
      \cline{13-13}
      & \multicolumn{3}{c}{$\gamma_1$}&\multicolumn{3}{c}{$\xi_1$}&& \\
      \cline{2-4}
      \cline{5-7}
      $\overline{+++}$      & 2 & 1.12 & 1.00  & 1 & 0.12 & 0.00    & - & - & - & - & - & -        \\[1mm]
      $\overline{---}$      & 2 & 1.19 & 1.00    & 1 & 0.19 & 0.00    & - & - & - & - & - & -        \\[1mm]
      $\widehat{+++}$       & 2 & 1.00 & 1.00 & 1 & 0.00 & 0.00  & - & - & - & - & - & -                   \\[1mm]
      $\widehat{---}$       & 2 & 1.39 & 1.00 & 1 & 0.39 & 0.00       & - & - & - & - & - & -        \\[1mm]
      $\widetriangle{+++}$  & 3 & 2.00 & 2.00   & 3 & 0.99 & 0.99     & - & - & - & - & - & -              \\[1mm]
      $\widetriangle{---}$  & 3 & 1.76 & 1.18 & 3 & 0.76 & 0.18 & - & - & - & - & - & -        \\[1mm]
      & \multicolumn{3}{c}{$\gamma_1$}&\multicolumn{3}{c}{$\xi_1$}&\multicolumn{3}{c}{$\xi_2$}& \\
      \cline{2-4}
      \cline{5-7}
      \cline{8-10}
      $\overline{+-+}$      & 2 & 1.00 & 0.99 & 1 & 0.00 & 0.00 & 0 & 0.00 & 0.00        & - & - & -             \\[1mm]
      $\overline{-+-}$      & 2 & 2.00 & 0.99 & 1 & 0.99 & 0.00 & 0 & 0.00 & 0.00 & - & - & -                    \\[1mm]
      $\widehat{+-+}$       & 2 & 1.00 & 1.00 & 1 & 0.00 & 0.00 & 0 & 0.00 & 0.00 & - & - & -                    \\[1mm]
      $\widehat{-+-}$       & 2 & 0.99 & 1.00 & 1 & 0.00 & 0.00 & 0 & 0.00 & 0.00 & - & - & -                    \\[1mm]
      & \multicolumn{3}{c}{$\gamma_1$}&\multicolumn{3}{c}{$\gamma_2$}&\multicolumn{3}{c}{$\xi_1$}&\multicolumn{3}{c}{$\xi_2$} \\
      \cline{2-4}
      \cline{5-7}
      \cline{8-10}
      \cline{11-13}
      $\overline{++-}$      & 1 & 2.00 & 0.00 & 1 & 0.00 & 1.00 & 1 & 0.00 & 0.00 & 0 & 1.00 & 0.00                    \\[1mm]
      $\overline{--+}$      & 1 & 0.38 & 0.00 & 1 & 0.62 & 1.00 & 1 & 0.00 & 0.00 & 0 & 0.00 & 0.00              \\[1mm]
      $\widehat{++-}$       & 1 & 2.00 & 0.72 & 1 & 0.52 & 0.28 & 1 & 1.52 & 0.00 & 0 & 0.00 & 0.00        \\[1mm]
      $\widehat{--+}$       & 1 & 0.66 & 0.15 & 1 & 0.34 & 0.85 & 1 & 0.00 & 0.00 & 0 & 0.00 & 0.00 \\[1mm]
      $\widetriangle{++-}$  & 1 & 2.00 & 0.00 & 2 & 0.00 & 1.00 & 2 & 0.00 & 0.00 & 1 & 0.99 & 0.00                    \\[1mm]
      $\widetriangle{--+}$  & 1 & 0.60 & 0.00 & 2 & 0.40 & 1.00 & 1 & 0.00 & 0.00 & 2 & 0.00 & 0.00             \\[1mm]
    \end{tabular}
    \caption{Exponents defining the optimal closure models, cf.~Table
      \ref{tab:closures}, found by solving problem \eqref{eq:min} with
      hard regularization ($\beta_1 = 0$, $\beta_2=\delta=2$) based on
      the data for the system Li$_{1/3}$Mn$_{2/3}$ (OA-1/3) and the
      system Li$_{1/2}$Mn$_{1/2}$ (OA-1/2) for each 3-cluster type
      indicated in the first column.  For comparison, the exponents
      characterizing the pair approximation (PA) are also shown. The
      results are rounded to two decimal places.}
    \label{tab:exp}
  \end{ruledtabular}
\end{table}

\begin{figure}
  \centering
    \begin{subfigure}[b]{0.45\textwidth}
  	\centering
  	\includegraphics[width=\textwidth]{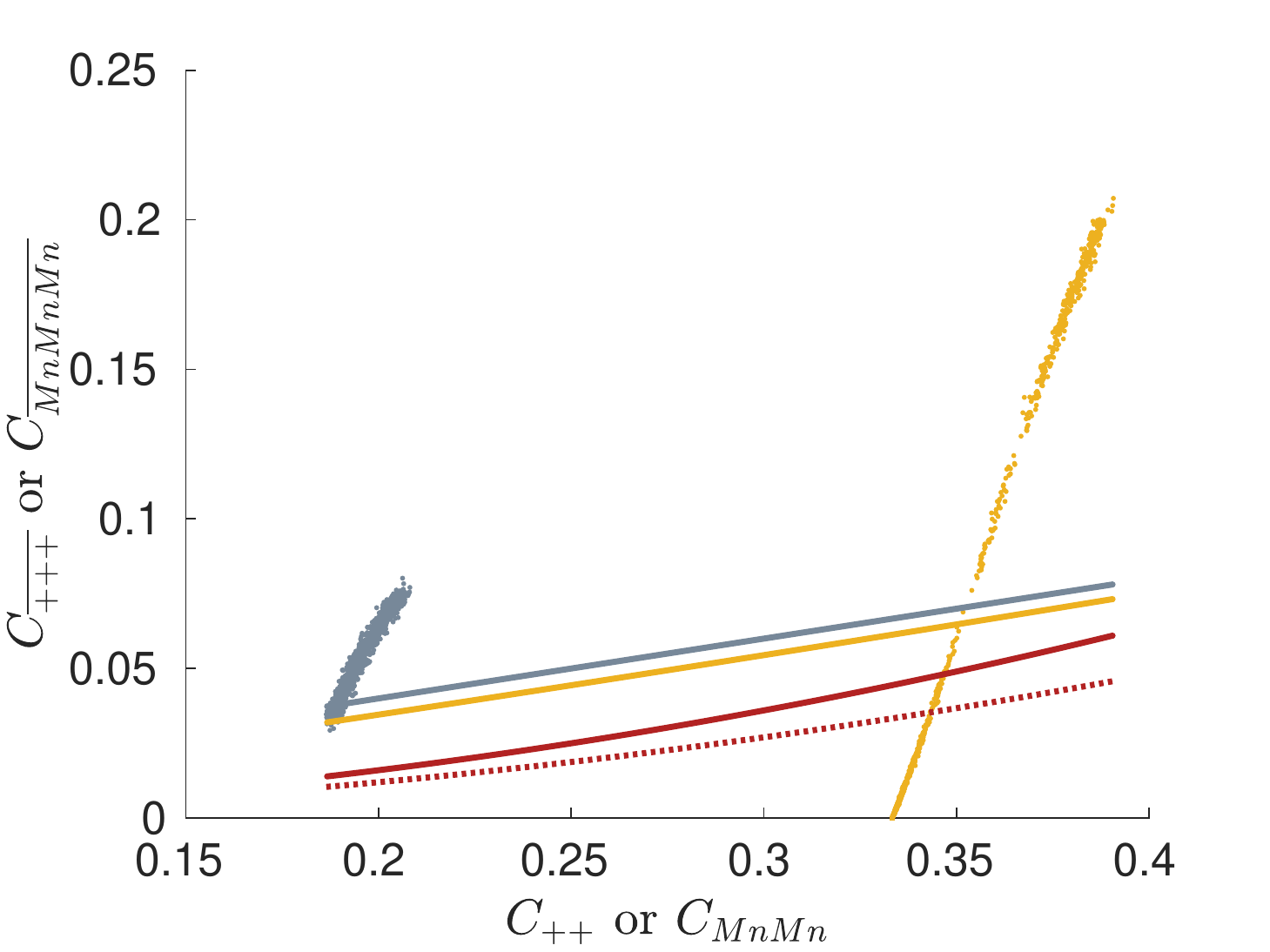}
  	\caption{}
  	\end{subfigure}
    \begin{subfigure}[b]{0.45\textwidth}
      \centering
      \includegraphics[width=\textwidth]{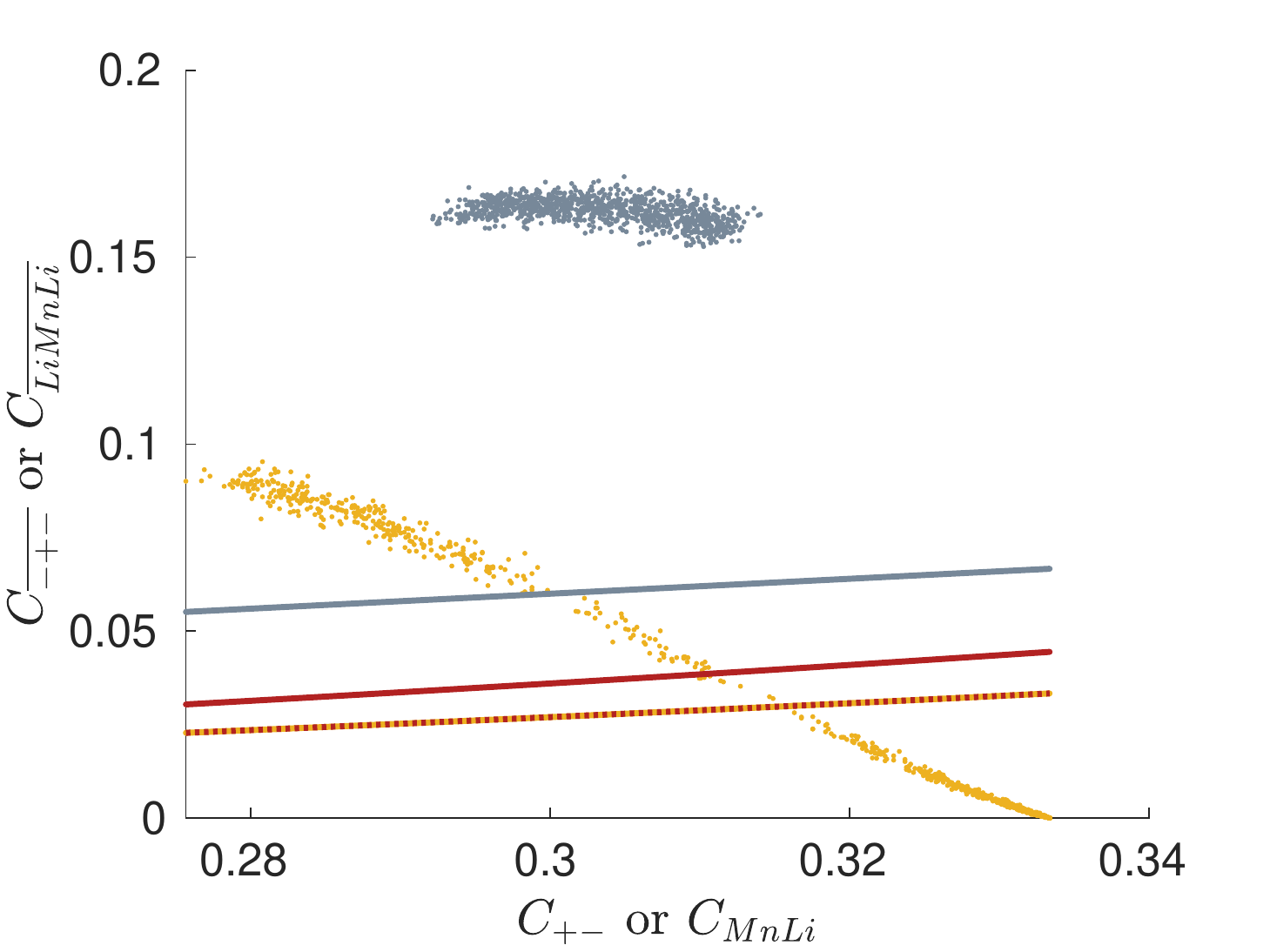}
      \caption{}
    \end{subfigure}
    \begin{subfigure}[b]{0.45\textwidth}
      \centering
      \includegraphics[width=\textwidth]{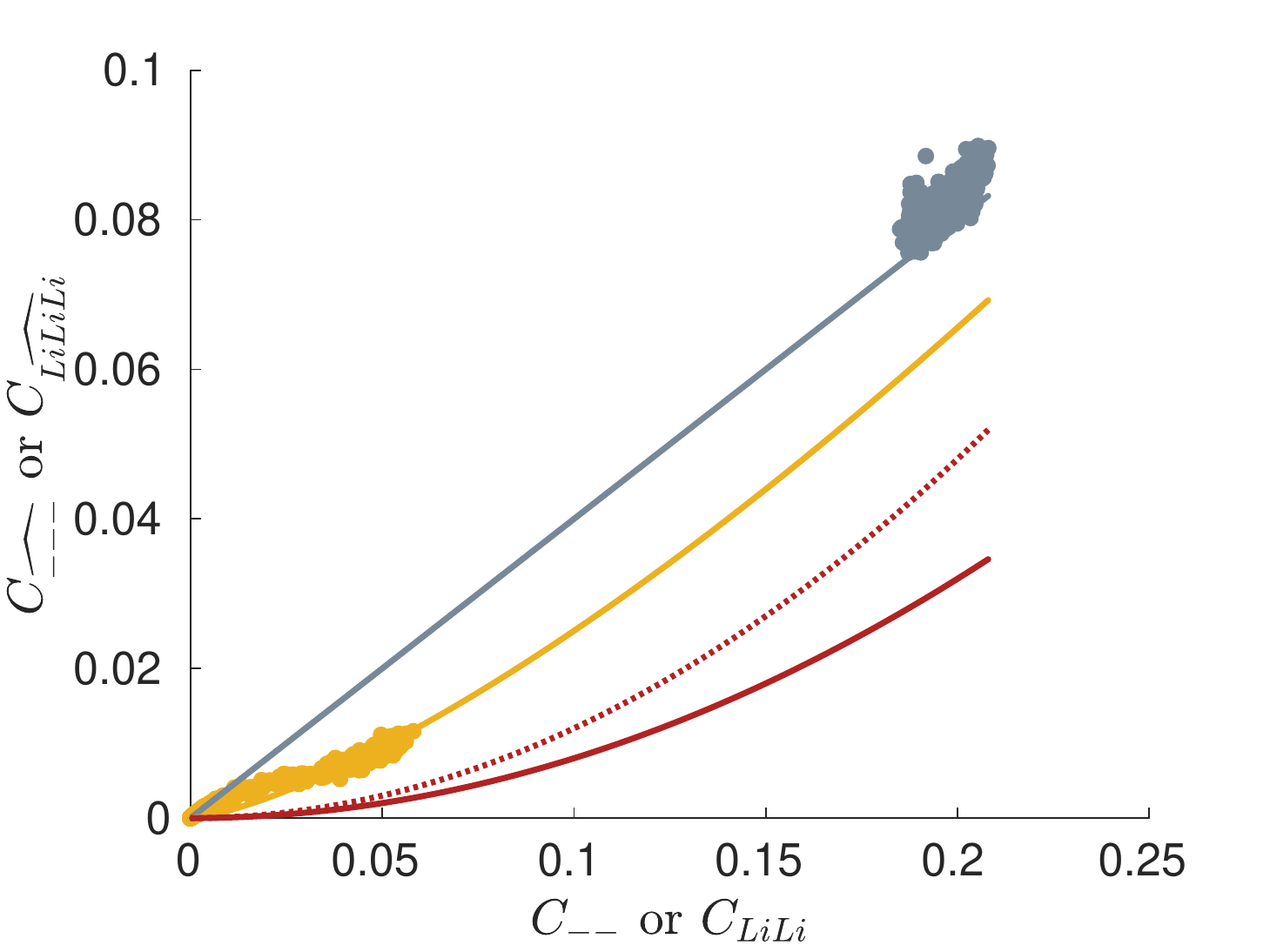}
      \caption{}
    \end{subfigure}
    \begin{subfigure}[b]{0.45\textwidth}
      \centering
      \includegraphics[width=\textwidth]{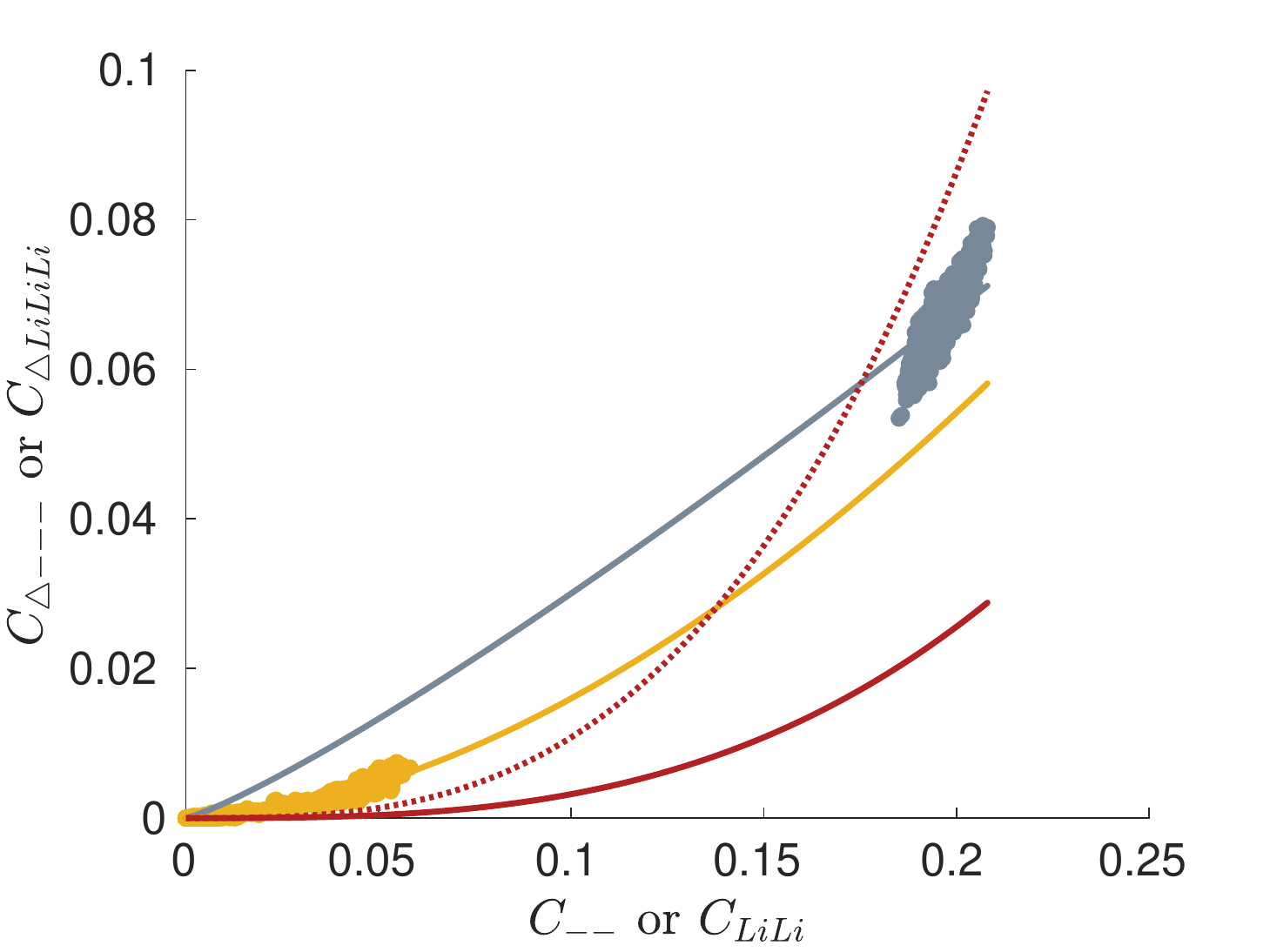}
      \caption{}
    \end{subfigure}
    \caption{Experimental triple concentrations for (gray symbols) the
      system Li$_{1/2}$Mn$_{1/2}$ and (yellow symbols) the system
      Li$_{1/3}$Mn$_{2/3}$ as functions of the corresponding pair
      concentrations for the 3-cluster types: $\overline{+++}$ (a),
      $\overline{-+-}$ (b), $\widehat{---}$ (c) and
      $\widetriangle{---}$ (d).  The corresponding reconstructions of
      triple concentrations from lower-order concentrations obtained
      via the optimal approximation and the pair approximation are
      shown with the grey lines and red solid lines for the system
      Li$_{1/2}$Mn$_{1/2}$, and with the yellow and red dotted lines
      for the system Li$_{1/3}$Mn$_{2/3}$. Note that the yellow and
      red dotted lines overlap in (b).}
  \label{fig:OA3OA2}
\end{figure}

\begin{figure}
	\centering
	\mbox{
		\begin{subfigure}[b]{0.5\textwidth}
			\centering
			\includegraphics[width=\textwidth]{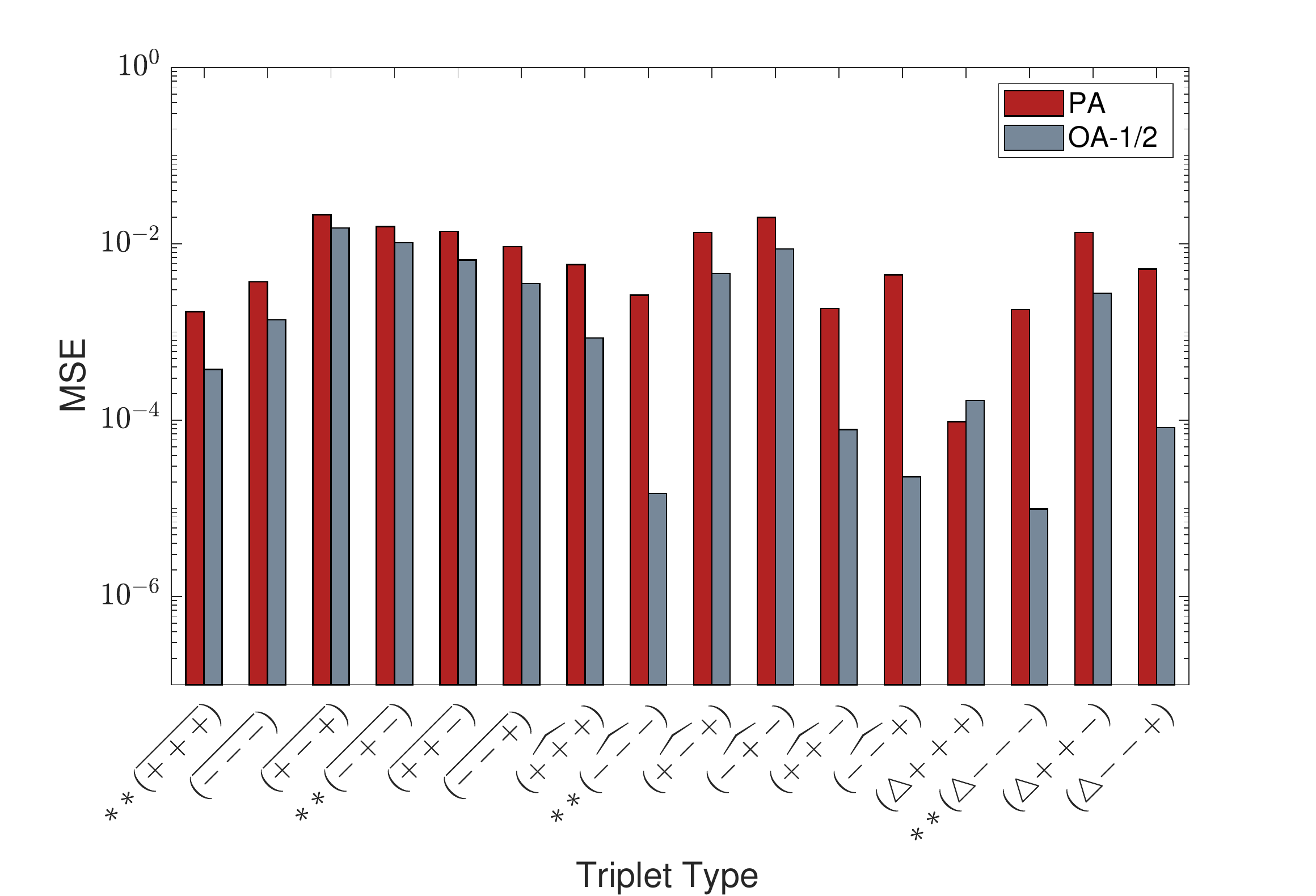}
			\caption{}
		\end{subfigure}
		\begin{subfigure}[b]{0.5\textwidth}
			\centering
			\includegraphics[width=\textwidth]{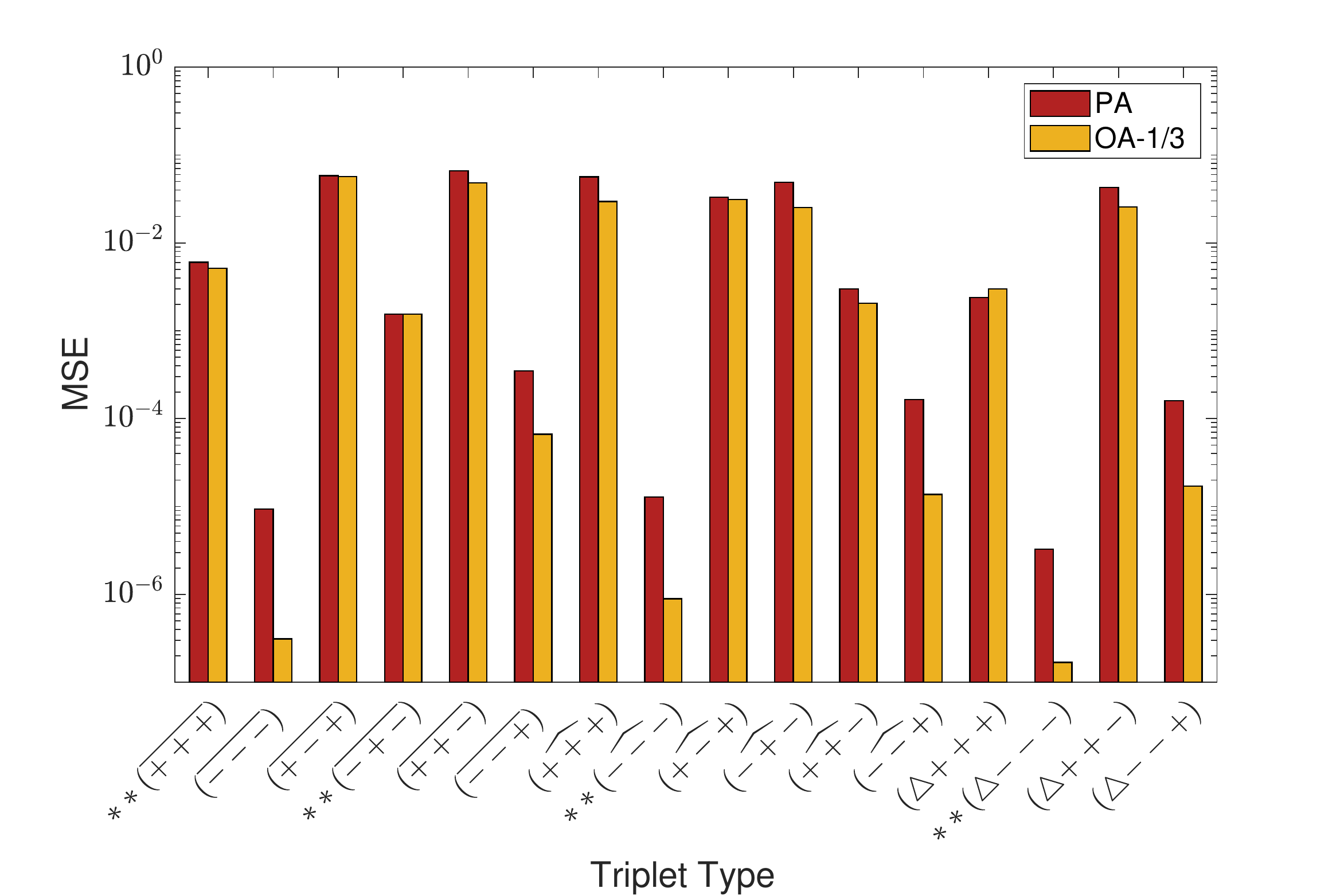}
			\caption{}
	\end{subfigure}}
      \caption{The mean-square errors \eqref{eq:Ji} for the pair
        approximation (PA) and optimal approximation subject to hard
        regularization for (a) the system Li$_{1/2}$Mn$_{1/2}$
        (OA-1/2) and (b) for the system Li$_{1/3}$Mn$_{2/3}$
        (OA-1/3). Predictions of the closure models for the triplet
        types marked with $(**)$ are analyzed in Figure
        \ref{fig:OA3OA2}.}
	\label{fig:OAPA_error}
\end{figure}

     \FloatBarrier

\subsection{Sparse Approximation and its Interpretation}
\label{sec:interpret}

In this section we investigate the exponents characterizing the
optimal closure presented in Table \ref{tab:exp}. As can be observed,
many exponents in the optimal closure relations are equal or close to
zero and this trend is more pronounced in the optimal closure obtained
by fitting the data for the symmetric system Li$_{1/2}$Mn$_{1/2}$
(when an exponent is zero, then the closure relation does not depend
on the corresponding 2-cluster concentration).  Thus, as is evident
from Table \ref{tab:newexp}, the resulting structure of the closure is
much simpler (``sparser'') for the optimal approximation than for the
closure obtained based on the pair approximation. More specifically,
note that for all triplet types, except for $(\widehat{++-})$,
$(\widehat{--+})$, $(\widetriangle{+++})$ and $(\widetriangle{---})$,
the optimal closure depends on the concentration of one 2-cluster
only. In order to make the structure of the closure model more uniform
which will facilitate its interpretation, we adjust the expressions
which do not follow the pattern. More specifically, in the optimal
closure relations for the clusters $(\widehat{--+})$ and
$(\widetriangle{---})$ the exponents are rounded up and down to the
nearest integer, whereas for $(\widehat{++-})$ and
  $(\widetriangle{+++})$ the change is more significant and involves
  adjusting the structure of the closure relation. We refer to this
simplified closure model as the {\em Sparse Approximation}
(SA) and its functional form is presented in Table \ref{tab:newexp}.

\begin{table}[h!]
	\begin{ruledtabular}
		\begin{tabular}{ c c c c }
			Triplet Type & Pair Approximation & Optimal Approximation & Sparse Approximation \\
			\cline{1-1}
			\cline{2-2}
			\cline{3-3}
			\cline{4-4}
			$\overline{+++}$      & $\frac{1}{5} \frac{C_{++}^{2} }{C_{+}}$                   & $\frac{1}{5} C_{++}$ & $\frac{1}{5} C_{++}$    \\[1mm]
			$\overline{---}$      & $\frac{1}{5} \frac{C_{--}^{2} }{C_{-}}$                   & $\frac{1}{5} C_{--}$ & $\frac{1}{5} C_{--}$    \\[1mm]
			$\overline{+-+}$      & $\frac{1}{5} \frac{C_{+-}^{2}}{ C_{-}}$                   & $\frac{1}{5} C_{+-}$ & $\frac{1}{5} C_{+-}$    \\[1mm]
			$\overline{-+-}$      & $\frac{1}{5} \frac{C_{+-}^{2}}{C_{+}}$                    & $\frac{1}{5} C_{+-}$ & $\frac{1}{5} C_{+-}$    \\[1mm]
			$\overline{++-}$      & $\frac{1}{5} \frac{C_{++} C_{+-}}{C_{+}}$                 & $\frac{1}{5} C_{+-}$ & $\frac{1}{5} C_{+-}$    \\[1mm]
			$\overline{--+}$      & $\frac{1}{5} \frac{C_{--} C_{+-}}{C_{-}}$                 & $\frac{1}{5} C_{+-}$ & $\frac{1}{5} C_{+-}$ \\[1mm]
			$\widehat{+++}$       & $\frac{2}{5} \frac{C_{++}^2}{C_{+}}$                      & $\frac{2}{5} C_{++}$ & $\frac{2}{5} C_{++}$    \\[1mm]
			$\widehat{---}$       & $\frac{2}{5} \frac{C_{--}^{2}}{C_{-}}$                    & $\frac{2}{5} C_{--}$ & $\frac{2}{5} C_{--}$    \\[1mm]
			$\widehat{+-+}$       & $\frac{2}{5} \frac{C_{+-}^{2}}{C_{-}}$                    & $\frac{2}{5} C_{+-}$ & $\frac{2}{5} C_{+-}$    \\[1mm]
			$\widehat{-+-}$       & $\frac{2}{5} \frac{C_{+-}^{2}}{C_{+}}$                    & $\frac{2}{5} C_{+-}$ & $\frac{2}{5} C_{+-}$    \\[1mm]
			$\widehat{++-}$       & $\frac{2}{5} \frac{C_{++} C_{+-}}{C_{+}}$                 & $\left[\frac{2}{5} C_{++}^{0.72} C_{+-}^{0.28}\right]$     & $\frac{2}{5} C_{+-}$        \\[1mm]
			$\widehat{--+}$       & $\frac{2}{5} \frac{C_{--} C_{+-}}{C_{-}}$                 & $\left[\frac{2}{5} C_{--}^{0.15} C_{+-}^{0.85}\right]$     &  $\frac{2}{5} C_{+-}$    \\[1mm]
			$\widetriangle{+++}$  & $\frac{2}{5} \frac{C_{++}^{3}}{C_{+}^{3}}$                & $\left[\frac{2}{5} \frac{C_{++}^{2}}{C_{+}}\right]$                  & $\frac{2}{5} C_{++}$     \\[1mm]
			$\widetriangle{---}$  & $\frac{2}{5} \frac{C_{--}^{3}}{C_{-}^{3}}$                & $\left[\frac{2}{5} \frac{C_{--}^{1.18}} {C_{-}^{0.18}}\right]$   & $\frac{2}{5} C_{--}$  \\[1mm]
			$\widetriangle{++-}$  & $\frac{2}{5} \frac{C_{++} C_{+-}^{2}}{C_{+}^{2} C_{-}}$   & $\frac{2}{5} C_{+-}$ & $\frac{2}{5} C_{+-}$     \\[1mm]
			$\widetriangle{--+}$  & $\frac{2}{5} \frac{C_{--} C_{+-}^{2}}{C_{+}C_{-}^{2}}$    & $\frac{2}{5} C_{+-}$ & $\frac{2}{5} C_{+-}$    \\[1mm]
		\end{tabular}
		\caption{Closure relations for 3-clusters of different
                  types derived based on the pair approximation, the
                  optimal approximation using the data for the system
                  Li$_{1/2}$Mn$_{1/2}$, cf.~Table \ref{tab:exp}, and
                  the sparse approximation discussed in
                  Section \ref{sec:interpret}.}
		\label{tab:newexp}
	\end{ruledtabular}
\end{table}

\begin{figure}
	\centering
	\includegraphics[width=\textwidth]{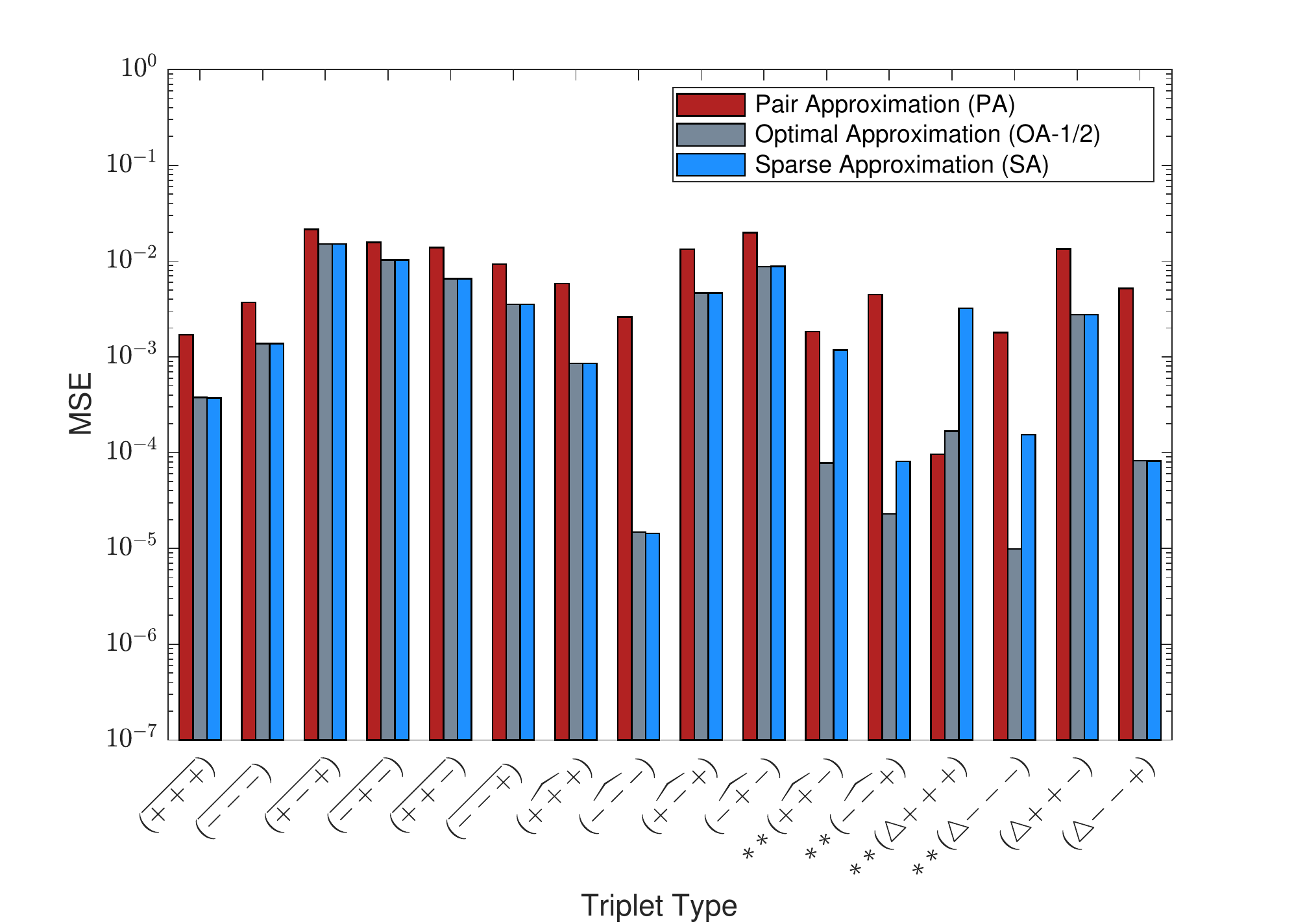}
	\caption{The mean-square reconstructions errors \eqref{eq:Ji}
          for the pair approximation, the optimal approximation
          constructed subject to hard regularization based on the data
          for the system Li$_{1/2}$Mn$_{1/2}$ and for the
          corresponding sparse approximation for different
          cluster types, cf.~Table \ref{tab:newexp}. Note that the
          results for the last two closures differ only for the
          clusters marked with $(**)$.}
	\label{fig:PA_OA_AOA}
\end{figure}

We now comment on how to interpret the structure of the sparse
approximation.  As discussed in Section \ref{sec:pair}, the pair
approximation model neglects the correlation between non-nearest
neighbour elements. This is due to the lack of information about the
triple correlation term $T_{ijk}$ in \eqref{eq:CijkA}. Considering
relations \eqref{eq:Cijk} for the sparse approximation, the
triplet correlation term is $T_{ijk} = \frac{C_j}{C_{ij}}$ for the
linear and angled triplets, and $T_{ijk} =
\frac{C_{i}C_{j}C_{k}}{C_{ij}C_{jk}}$ for the triangular triplets.
This is contrary to the assumption that $T_{ijk}=1$ which is central
to the pair approximation. With the data in Table \ref{tab:newexp} we
are now in the position to refine the assumptions underlying these
approximations.  Referring to relations \eqref{eq:Cijk}, the
concentration of the triplet ($C_{ijk}$) can be written as the global
pair concentration ($C_{ij}$) times the conditional probability of
finding a nearest-neighbour element to the pair in a certain state
($P_{k|ij}$). Considering the linear and angled triplets in the
sparse approximation formulation, we obtain
\begin{equation}
	\label{eq:Cijk_lin}
	\begin{aligned}
	C_{ijk} = C_i C_j C_k Q_{ij} Q_{jk} T_{ijk} = C_i C_j C_k \frac{C_{ij}}{C_i C_j} \frac{C_{jk}}{C_j C_k}  \frac{C_j}{C_{ij}} = C_{ij} \frac{C_{jk}/C_j}{C_{ij}/C_j} = C_{ij} \frac{P_{k|j}}{P_{i|j}}.
	\end{aligned}
\end{equation}
In a similar way one can consider the triangular triplets where
\begin{equation}
\label{eq:Cijk_tri}
\begin{aligned}
C_{ijk} = C_i C_j C_k Q_{ij} Q_{jk} Q_{ik} T_{ijk} = C_{ij} \frac{C_{ik}/C_i}{C_{ij}/C_i} = C_{ij} \frac{P_{k|i}}{P_{j|i}}.
\end{aligned}
\end{equation}
We thus deduce 
\begin{subequations}
\label{eq:Pijk}
\begin{align}
P_{k|ij} = \frac{P_{k|j}}{P_{i|j}}, ~~~~~ \text{for linear and angled clusters},\label{eq:Pijka}\\
P_{k|ij} = \frac{P_{k|i}}{P_{j|i}}, ~~~~~~~~~~~~~~~ \text{for triangular clusters}.\label{eq:Pijkb}
\end{align}
\end{subequations}
These relations break down the probability of a $3$-cluster in terms
of probabilities of two $2$-clusters. They can be regarded as
generalizations of the pair approximation model, cf. relation
\eqref{eq:PA_assumption}, with the inclusion of a term in the
denominator. To understand the meaning of this extension of the pair
approximation, we refer to relation \eqref{eq:PA_assumption}. It is
clear that closure is achieved using the pair approximation by
assuming that the conditional probability of an element $k$ being a
nearest-neighbour of $j$ is equal to the conditional probability of
$k$ being a nearest-neighbour of an $ij$ pair. In other words, the
pair approximation model assumes that an element $j$ is always a
nearest-neighbour of $i$, and we cannot find an element $j$ which is
not a nearest-neighbour of $i$. However, we know that this simplifying
assumption is not correct in general and there is always a possibility
of finding an element $j$ which is not a nearest-neighbour of $i$. By
re-arranging relation \eqref{eq:Pijka} in the form
$P_{k|j} = P_{i|j} P_{k|ij}$, it is evident that the SA model assumes
that $j$ might not always be a nearest-neighbour of $i$ and accounts
for this possibility through the term $P_{i|j}$. A similar
interpretation can be adopted for triangular clusters.

The accuracy of the optimal approximation is certainly affected when
the exponents in the closure relations for the four triplet types are
adjusted as discussed above, cf.~Table \ref{tab:newexp}. Figure
\ref{fig:PA_OA_AOA} shows the reconstruction errors for triplet
concentrations obtained using different closure models for the system
Li$_{1/2}$Mn$_{1/2}$. As can be expected, the SA model is less
accurate in comparison to the OA model for the triplets
$(\widehat{++-})$, $(\widehat{--+})$, $(\widetriangle{+++})$ and
$(\widetriangle{---})$. However, the performance of SA model is still
better than that of the pair approximation model for the triplets
$(\widehat{++-})$, $(\widehat{--+})$ and $(\widetriangle{---})$. To
conclude, the adjustments to the OA model sacrifice a degree of the
accuracy in reconstructing the triplet concentration for
$(\widetriangle{+++})$ while achieving a simpler and interpretable
model.

As a result of the simple structure of the SA closure, cf.~Table
\ref{tab:newexp}, system \eqref{eq:dCc} closed with this model becomes
linear and hence analytically solvable. It takes the form
\begin{subequations}
	\begin{align}
	\frac{d}{dt} C_{++} = 2\alpha_1 C_{+-},
	\label{eq:dC++AOA}\\[2mm]
	\frac{d}{dt} C_{--} = 2\alpha_2 C_{+-},
	\label{eq:dC--AOA}\\[2mm]
	\frac{d}{dt} C_{+-} = \left( -\alpha_1 - \alpha_2 \right) C_{+-},
	\label{eq:dC+-AOA}
	\end{align}
	\label{eq:dCcAOA}
\end{subequations}
where the parameters
$\alpha_1 = \frac{4}{5} k_1 + \frac{1}{5} k_2 - \frac{4}{5} k_3 -
\frac{1}{5} k_4$ and
$\alpha_2 = \frac{4}{5} k_5 + \frac{1}{5} k_6 - \frac{4}{5} k_7 -
\frac{1}{5} k_8$ are linear combinations of the reaction rates. The
solution then is
\begin{subequations}
	\begin{alignat}{2}
	C_{+-}(t) &= \mu_1 e^{(-\alpha_1-\alpha_2)t}, & \qquad\quad \mu_1 &= C_{+-_{0}},
	\label{eq:analytic+-}\\
	C_{++}(t) &= \frac{2\alpha_1 \mu_1}{-\alpha_1-\alpha_2}
        e^{(-\alpha_1-\alpha_2)t} + \mu_2, &  \mu_2 &= C_{++_{0}} - \frac{2\alpha_1\mu_1}{-\alpha_1-\alpha_2},
	\label{eq:analytic++}\\
	C_{--}(t) &= \frac{2\alpha_2 \mu_1}{-\alpha_1-\alpha_2}
        e^{(-\alpha_1-\alpha_2)t} + \mu_3, &   \mu_3 &= C_{--_{0}} - \frac{2\alpha_2\mu_1}{-\alpha_1-\alpha_2},
	\label{eq:analytic--}
	\end{alignat}
	\label{eq:analytic}
\end{subequations}
where $C_{+-_{0}}$, $C_{++_{0}}$ and $C_{--_{0}}$ are the initial
concentrations of the corresponding $2$-clusters. As is evident in
\eqref{eq:analytic}, the concentrations $C_{++}$ and $C_{--}$ decrease
exponentially in time with the decay rate
$-(\alpha_1+\alpha_2)$. These two parameters instead of eight reaction
rates $k_1$ to $k_8$ are sufficient to describe the evolution of
concentrations of different clusters in time. In addition to producing
an analytically solvable model, an advantage of the SA closure is
that the inverse problem \eqref{eq:minJ} also simplifies and needs to
be solved with respect to $\alpha_1$ and $\alpha_2$ only which does
not require Bayesian inference. Results will be presented in Section
\ref{sec:results_rates}.  

\FloatBarrier

\subsection{Prediction Capability of the Closure Models}
\label{sec:predict}

In order to assess the predictive capability of the truncated model
closed with the optimal approximation or the sparse
approximation, the $3$-cluster concentrations are reconstructed as
functions of time from $2$-cluster concentrations. We are interested
in evaluating the prediction accuracy of these models in comparison to
the model equipped with the pair approximation. In order to assess the
robustness of these predictions, we will do this for stoichiometries
other than the one for which the models were calibrated, cf.~Sections
\ref{sec:results_opt} and \ref{sec:interpret}. More specifically,
while the simulated annealing data for the system with the composition
Li$_{1/3}$Mn$_{2/3}$ was used for calibration, cf.~Figure
\ref{fig:SA}, accuracy of the models will be analyzed here for 10
different stoichiometries Li$_{x}$Mn$_{1-x}$,
$x \in \{0.25,0.30,0.33,0.36,0.42,0.50,0.58,0.64,0.70,0.75\}$. In
particular, we are interested in the effect of regularization --- soft
versus hard with different parameters $\delta$, $\beta_1$ and
$\beta_2$ --- in the solution of problem \eqref{eq:min}.

Robustness of the model performance will be assessed in terms of the
mean-square error \eqref{eq:Ji} averaged over all types of 3-clusters,
i.e.,
\begin{equation}
\mathcal{E} = \frac{1}{|\Theta|}\sum_{i \in \Theta}^{} J_i,
\label{eq:E}
\end{equation}
where $|\Theta| = 16$ is the total number of $3$-clusters,
cf.~\eqref{eq:Theta}, and the true $3$-cluster concentrations
$\widetilde{C}_i(t)$ are obtained from simulated annealing experiments
performed for each considered stoichiometry.  The corresponding
2-cluster concentrations are used to reconstruct the $3$-cluster
concentrations as a function of time for each triplet type via the
optimal and sparse closure approximations. Thus, this
diagnostic is designed to asses only the accuracy of the closure
relations given in Table \ref{tab:newexp}, rather than of the entire
truncated model \eqref{eq:dCc}.

Error \eqref{eq:E} is shown as function of the stoichiometry for the
optimal closure obtained for the system Li$_{1/3}$Mn$_{2/3}$ subject
to hard and soft regularization in Figures \ref{fig:E}a and
\ref{fig:E}b, respectively. In addition, in these figures we also show
the errors obtained with the model based on the pair approximation. As
can be observed, harder regularization results in larger prediction
errors for stoichiometries close to Li$_{1/3}$Mn$_{2/3}$ in comparison
to softer regularization strategies. On the other hand, harder
regularization reveals better predictive performance for
stoichiometries different from Li$_{1/3}$Mn$_{2/3}$. In other words,
less aggressive regularization performs better on stoichiometries
close to the stoichiometry for which the calibration of the closure
relations from Table \ref{tab:closures} was performed in Section
\ref{sec:results_opt}, and the performance gradually degrades as the
stoichiometries become more different from Li$_{1/3}$Mn$_{2/3}$. We
thus conclude that there is a trade-off between robustness and
accuracy of the closure models, in the sense that models optimized for
a particular stoichiometry tend to be less robust when used to
describe other stoichiometries.

Finally, robustness of the closures based on the pair approximation,
the optimal approximation subject to hard regularization for the
system Li$_{1/3}$Mn$_{2/3}$ and the corresponding sparse
approximation is compared for a range of stoichiometries in Figure
\ref{fig:E}. Note that solving the minimization problem \eqref{eq:min}
subject to hard regularization produces more versatile closure models
that can be applied to a range of stoichiometries without significant
loss of accuracy. Hence, the optimal approximation models of interest
are achieved by hard regularization in \eqref{eq:min}. Figure
\ref{fig:MSE} shows the mean error \eqref{eq:E} for a range of
stoichiometries for the three aforementioned closure models. A
significant improvement with respect to the performance of the pair
approximation model is achieved by the optimal closure models for all
stoichiometries. As can be observed, the SA model performs better
than the OA-1/3 model for most of the stoichiometries, except the ones
that are close to the system Li$_{1/3}$Mn$_{2/3}$.  This is due to the
fact that in the OA-1/3 model the minimization problem \eqref{eq:min}
is solved for the system Li$_{1/3}$Mn$_{2/3}$, and hence fits are more
accurate in the neighbourhood of this stoichiometry. We conclude by
noting that when averaged over all stoichiometries, the performance of
the sparse approximation model is improved by $36.13\%$ over
the performance of the pair approximation model.

\begin{figure}
	\centering
	\begin{subfigure}[b]{0.49\textwidth}
		\centering
		\includegraphics[width=\textwidth]{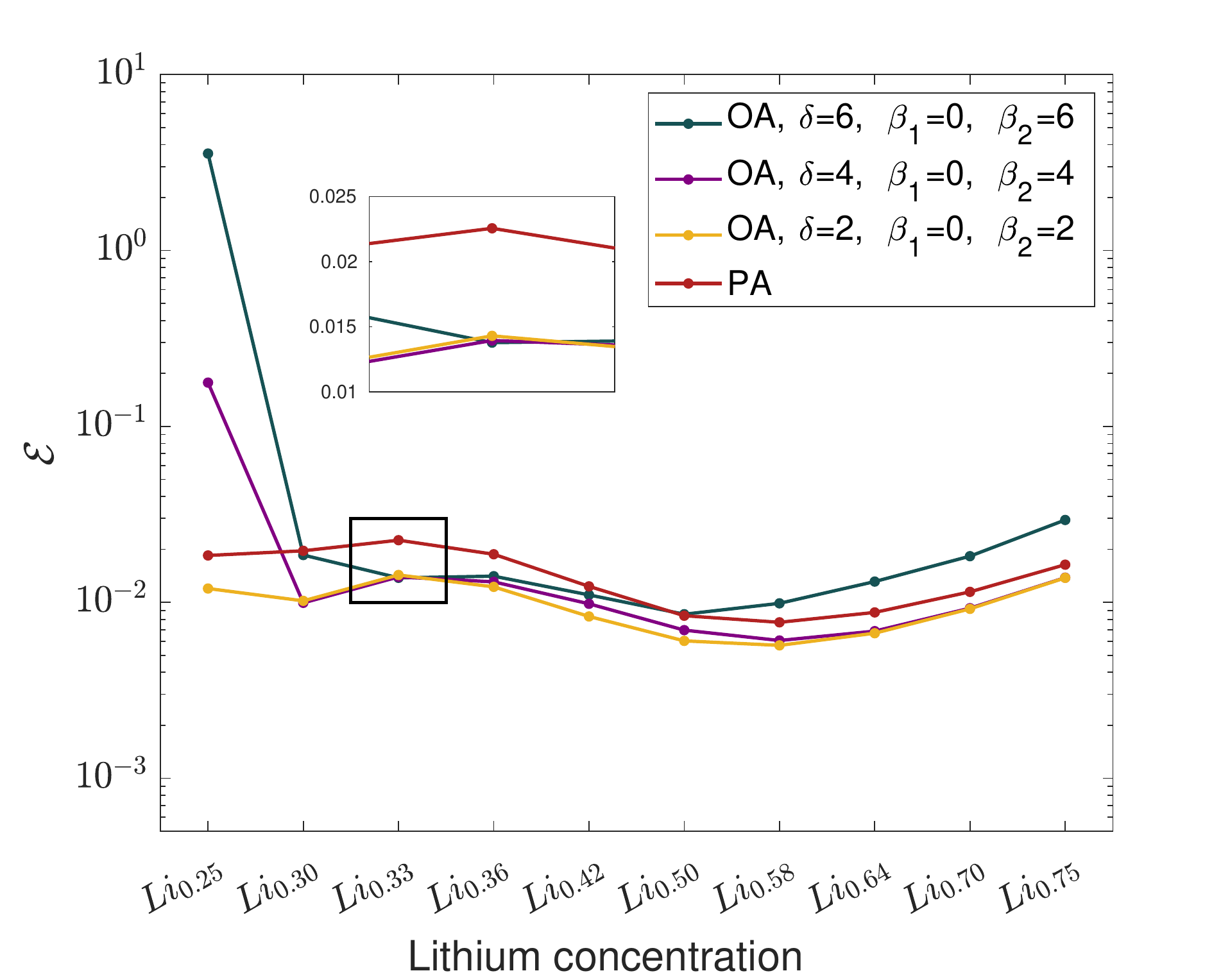}
		\caption{}
	\end{subfigure}
	\begin{subfigure}[b]{0.49\textwidth}
		\centering
		\includegraphics[width=\textwidth]{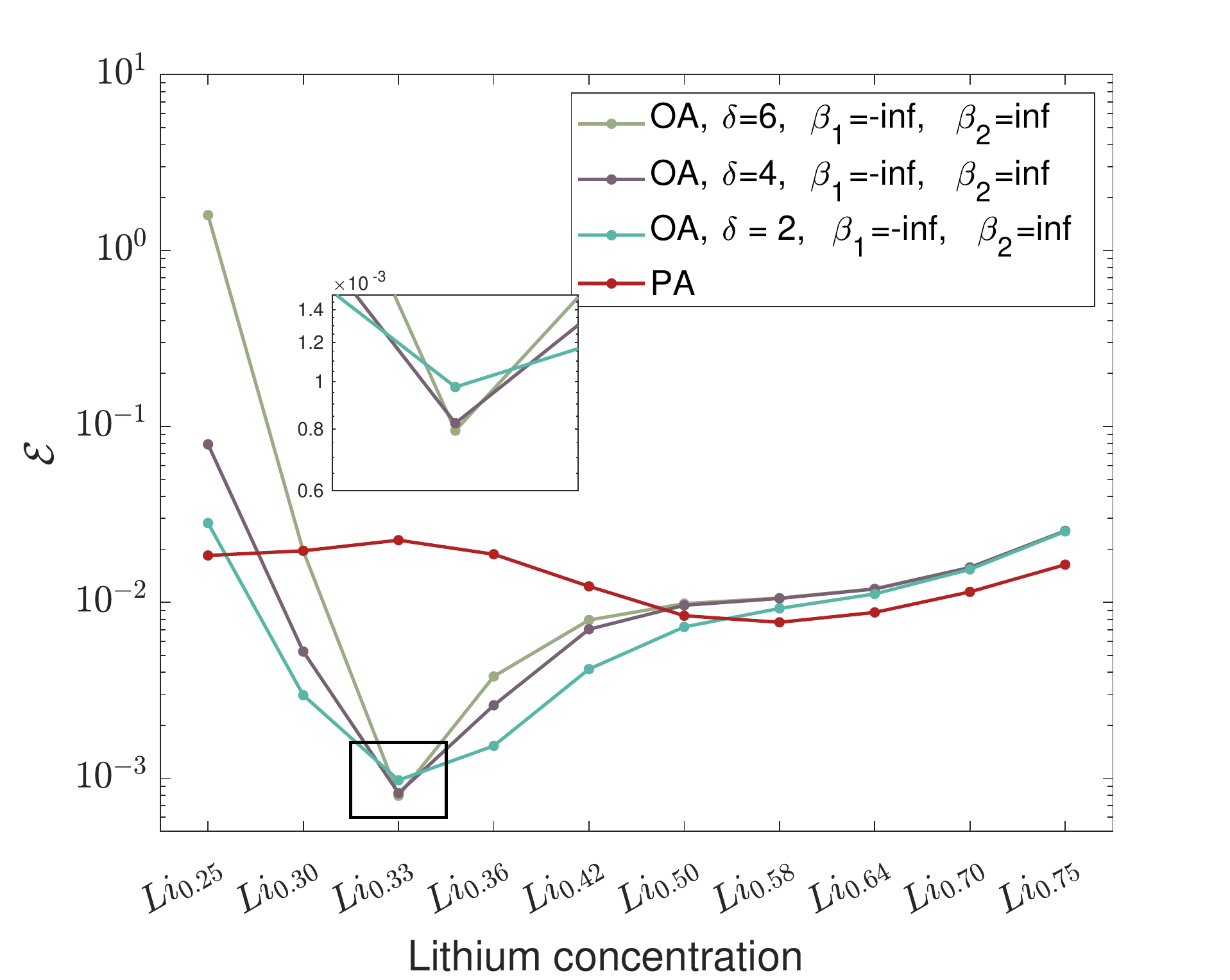}
		\caption{}
	\end{subfigure}
	\caption{Dependence of the mean error \eqref{eq:E}
          characterizing the accuracy of the different closure
          relations on the stoichiometry for (a) hard regularization
          and (b) soft regularization employed in the solution of
          optimization problem \eqref{eq:min} with parameters
          indicated in the legend for Li$_{1/3}$Mn$_{2/3}$ system.
          ``PA'' and ``OA'' refer to, respectively, the pair and the
          optimal approximation.}
	\label{fig:E}
\end{figure}
\begin{figure}
	\centering
	\begin{subfigure}[b]{0.6\textwidth}
		\centering
		\includegraphics[width=\textwidth]{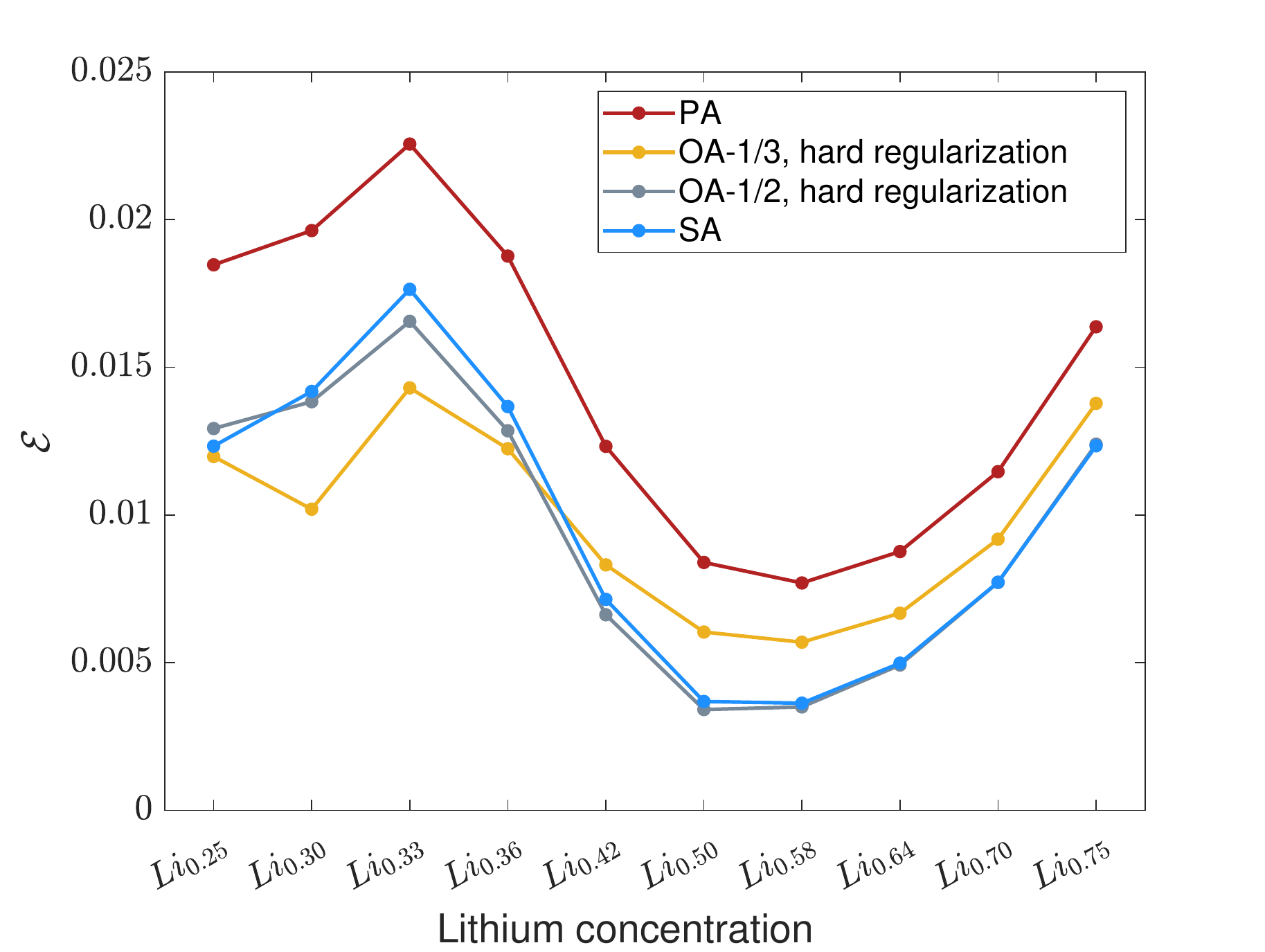}
	\end{subfigure}
	\caption{The mean error \eqref{eq:E} characterizing the
          accuracy of the different closure relations indicated in the
          legend for a range of different stoichiometries. }
	\label{fig:MSE}
\end{figure}

\FloatBarrier

\subsection{Inferring Reaction Rates}
\label{sec:results_rates}

The reaction rates $k_1,\dots,k_8$ in system \eqref{eq:dCc} are
determined in probabilistic terms using Bayesian inference for the
pair approximation and the optimal closure models.  On the other hand,
for the sparse approximation there are only two unknown
parameters ($\alpha_1$ and $\alpha_2$) so they can be inferred by
solving the problem
$\min_{(\alpha_1,\alpha_2) \in \RR^2} \J(\alpha_1,\alpha_2)$ where the
concentrations in the error functional are evaluated using the
closed-form relations \eqref{eq:analytic}. Although this minimization
problem is not convex, a global minimum can be found using standard
optimization methods.

In the problems involving the pair approximation and the optimal
closure models some of the reaction rates were found to be essentially
equal to zero (or vanishingly small), so here the results are
presented for the remaining rates only. In Figures
\ref{fig:posterior}a and \ref{fig:posterior}c we visualize the Markov
chains obtained with Algorithm \ref{alg:MCMC} for system
\eqref{eq:dCc} closed with, respectively, the pair approximation, the
optimal approximation with exponents determined subject to hard
regularization (OA-1/2), cf.~Table \ref{tab:exp}. The Cartesian
coordinates of each point in Figures \ref{fig:posterior}a,c represent
{three of the parameters} characterizing an individual Monte-Carlo
sample, whereas information about the remaining parameters is encoded
in the color of the symbol via the red-green-blue (RGB) mapping, as
shown in the color maps in Figures \ref{fig:posterior}b,d. The size of
the symbols is proportional to $\J(\bK)^{-1}$ such that parameter
values producing better fits stand out as they are represented with
larger symbols. Note that, for clarity, the entire Markov chains are
not presented in Figure \ref{fig:posterior} as the data is filtered
based on the value of the cost function (i.e., data points are shown
only if $\J(\bK)$ is smaller than some threshold).

It is evident from Figures \ref{fig:posterior}a,c that in each case
parameter values producing good fits form a number of
clusters, which reflects the fact that problem \eqref{eq:minJ} indeed
admits multiple local minima. It is also interesting to see that good
fits are obtained with some of the reaction rates varying by 200\% or
more which is a manifestation of the ill-posedness of problem
\eqref{eq:minJ} when the outputs $\bC(\bK)$ reveal weak dependence on
some of the parameters in $\bK$. In order to compare the quality of
fits obtained with the pair and optimal approximations, in Figures
\ref{fig:hist}a,b we show the histograms of the values of the error
functional $\J(\bK)$ obtained along the Markov chains. Overall, the
quality of the fits is comparable in both cases and exhibits
significant uncertainty, although poor fits appear more likely when
the closure based on the the pair approximation is used.  The optimal
parameter values for the closure based on the SA model are
$(\alpha_1^*,\alpha_2^*) = (-0.083,-0.166)$ and, as we can see in Figures
\ref{fig:hist}a,b, while the accuracy of the fit is lower than in the
previous two cases, there is effectively no uncertainty in the
determination of the parameters.

\begin{figure}
	\centering
	\begin{subfigure}[b]{0.42\textwidth}
		\centering
		\includegraphics[width=\textwidth]{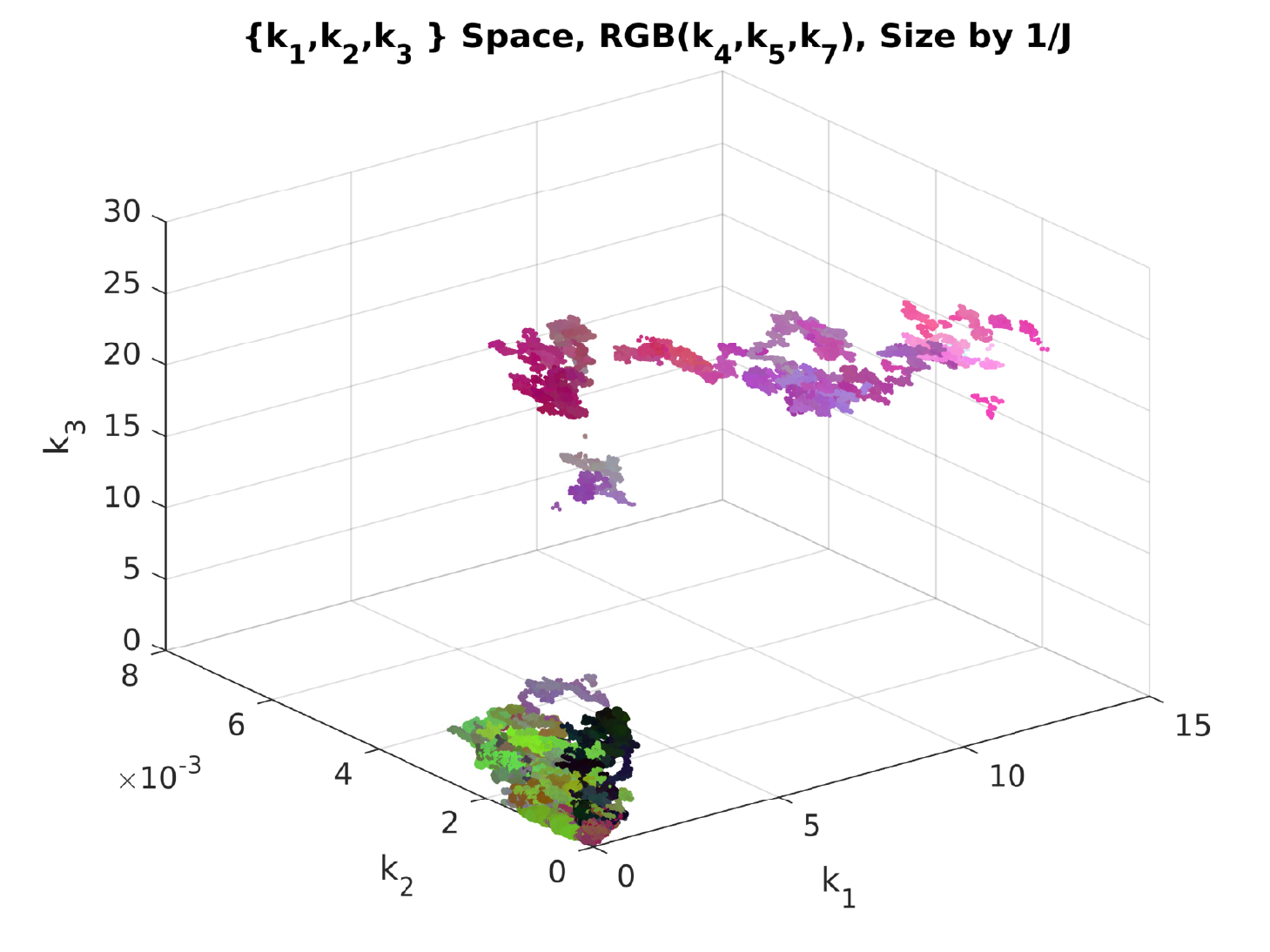}
		\caption{}
	\end{subfigure}
	\begin{subfigure}[b]{0.42\textwidth}
		\centering
		\includegraphics[width=\textwidth]{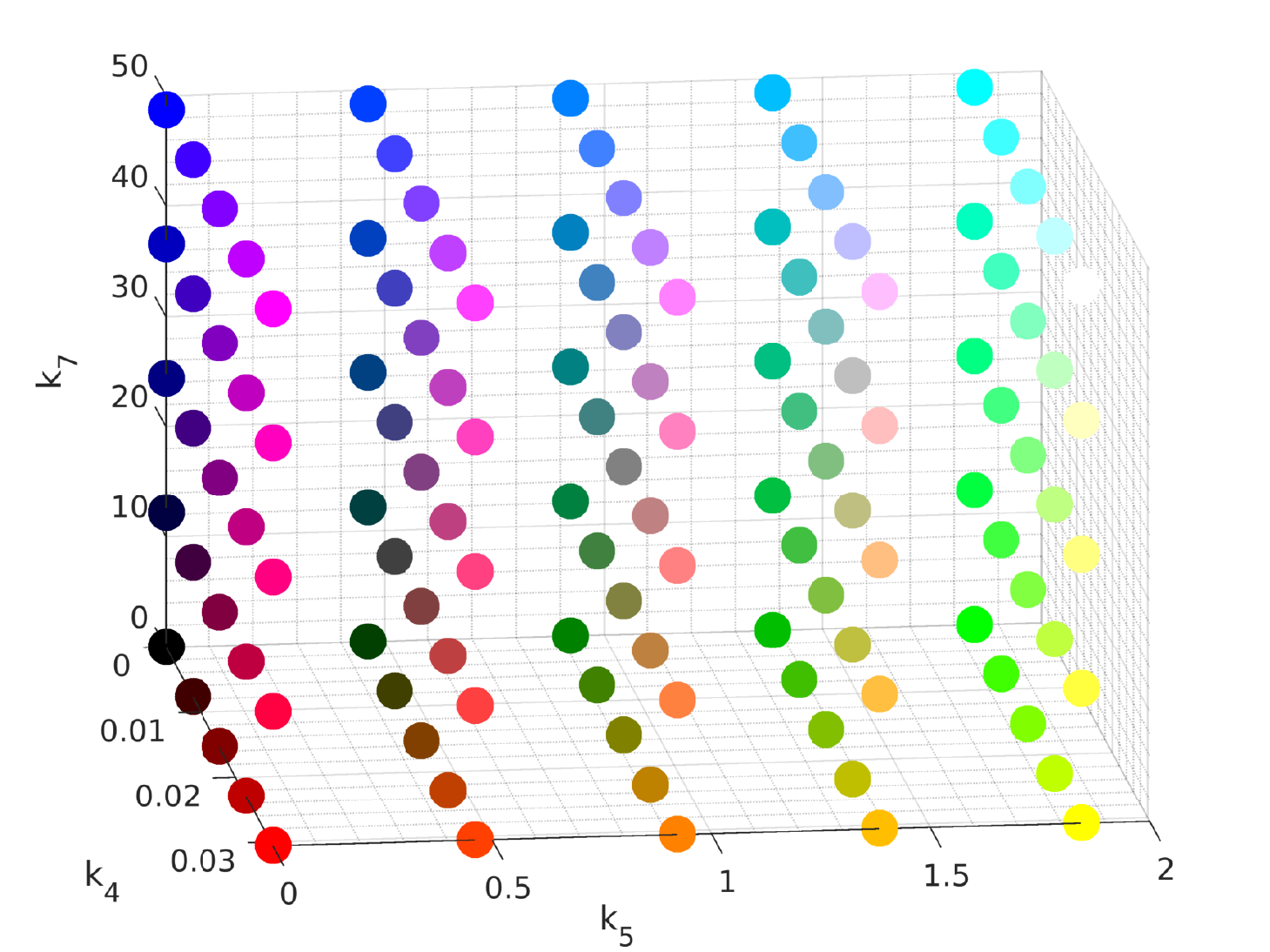}
		\caption{}
	\end{subfigure}
	\begin{subfigure}[b]{0.42\textwidth}
		\centering
		\includegraphics[width=\textwidth]{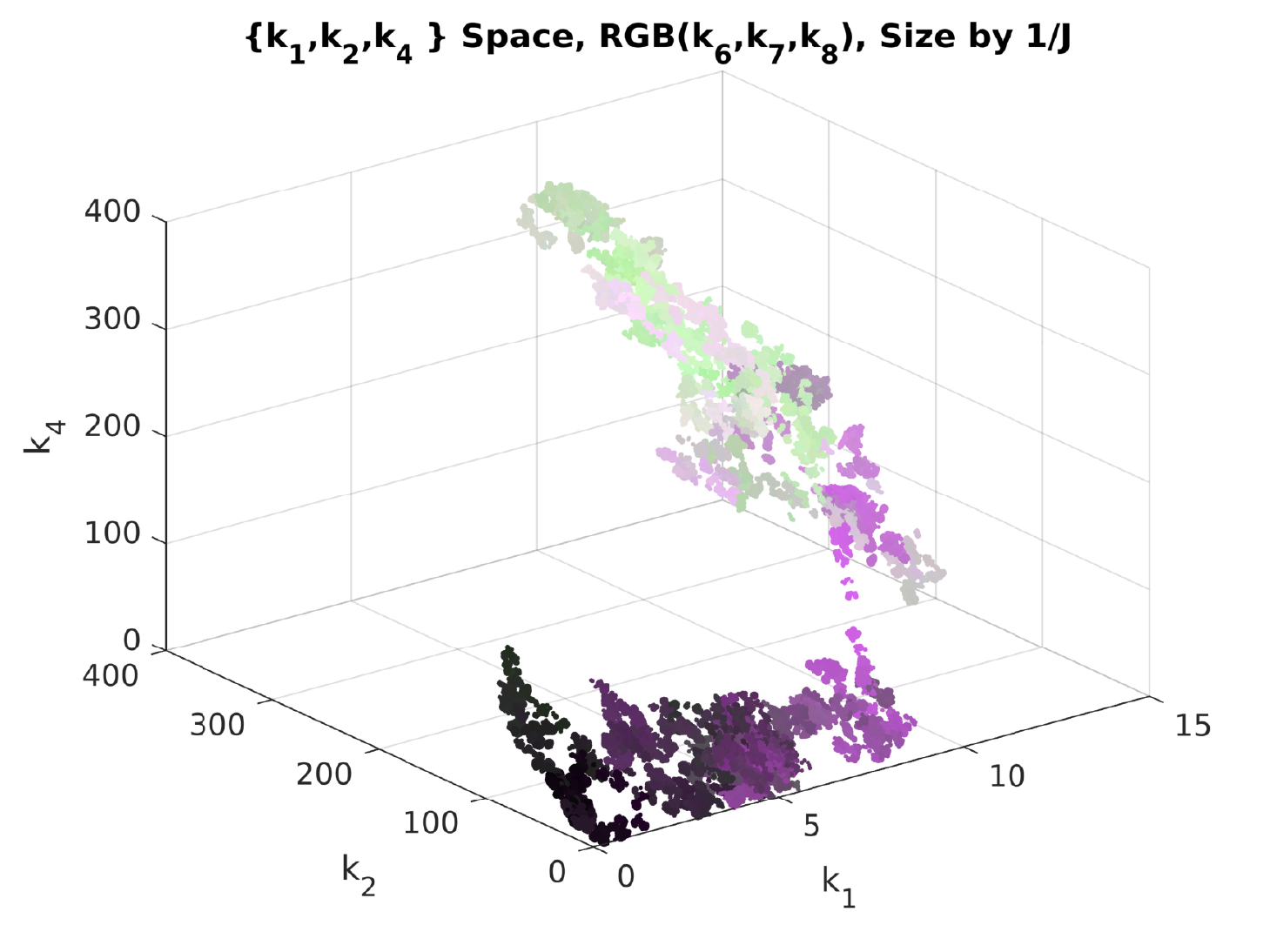}
		\caption{}
	\end{subfigure}
	\begin{subfigure}[b]{0.42\textwidth}
		\centering
		\includegraphics[width=\textwidth]{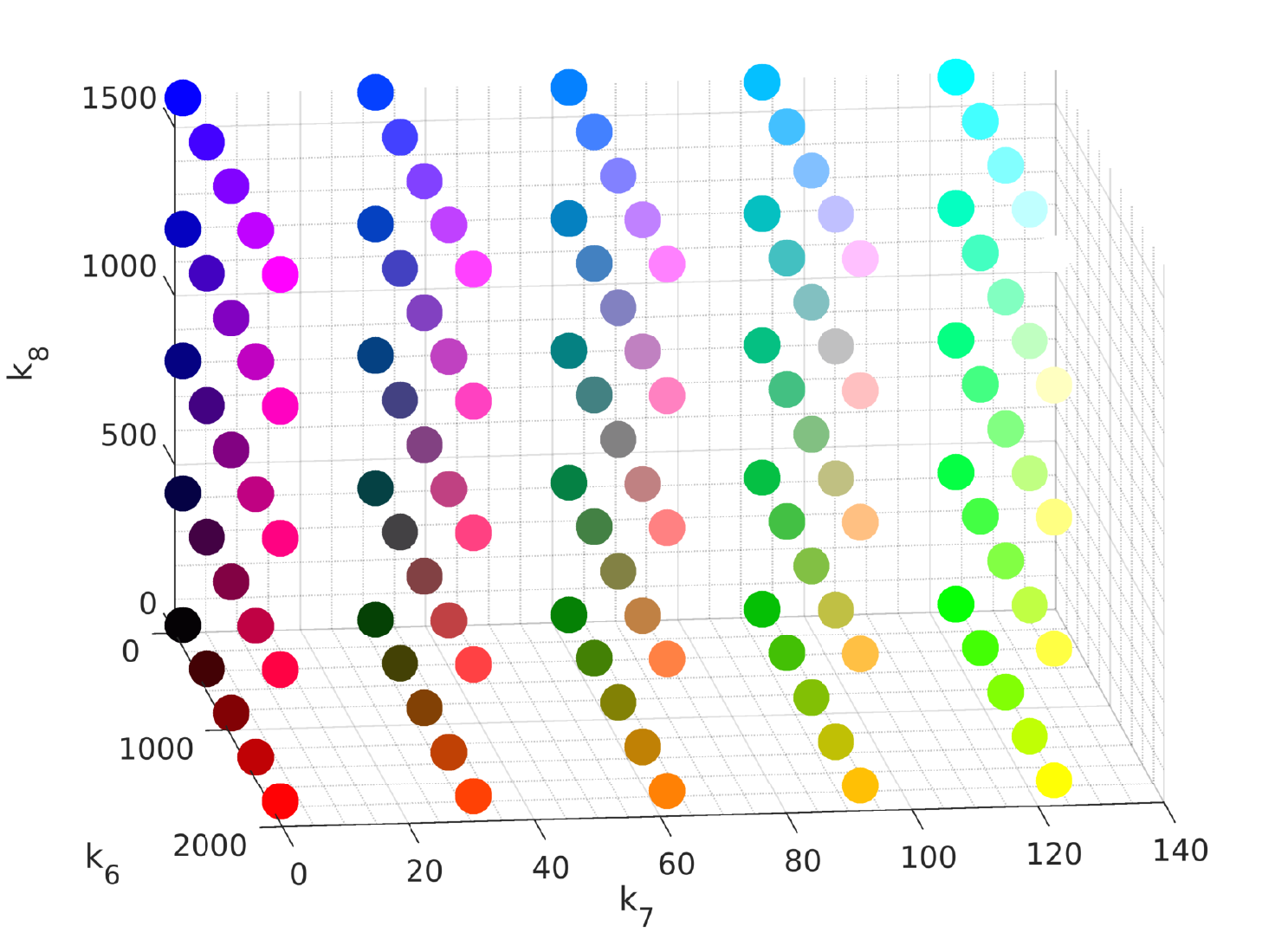}
		\caption{}
	\end{subfigure}
	\caption{Posterior probability densities
          $\mathbb{P}\left(\textbf{K}|\widetilde{\textbf{C}}\right)$
          obtained using Algorithm \ref{alg:MCMC} for problem
          \eqref{eq:minJ} with system \eqref{eq:dCc} closed using (a)
          the pair approximation and (c) the optimal approximation
          with exponents determined subject to hard regularization
          (OA-1/2).  The parameters $k_1$, $k_2$ and $k_3$ are
          represented in term of the Cartesian coordinates whereas the
          remaining three nonzero rate constants are encoded in terms
          of the color of the symbols via the color maps shown in
          panel (b) and (d).  The size of the symbols in panels (a)
          and (c) is proportional to $J(\bK)^{-1}$.}
	\label{fig:posterior}
\end{figure}

\begin{figure}
  \centering
  \mbox{
  \begin{subfigure}[b]{0.45\textwidth}
    \centering
    \includegraphics[width=\textwidth]{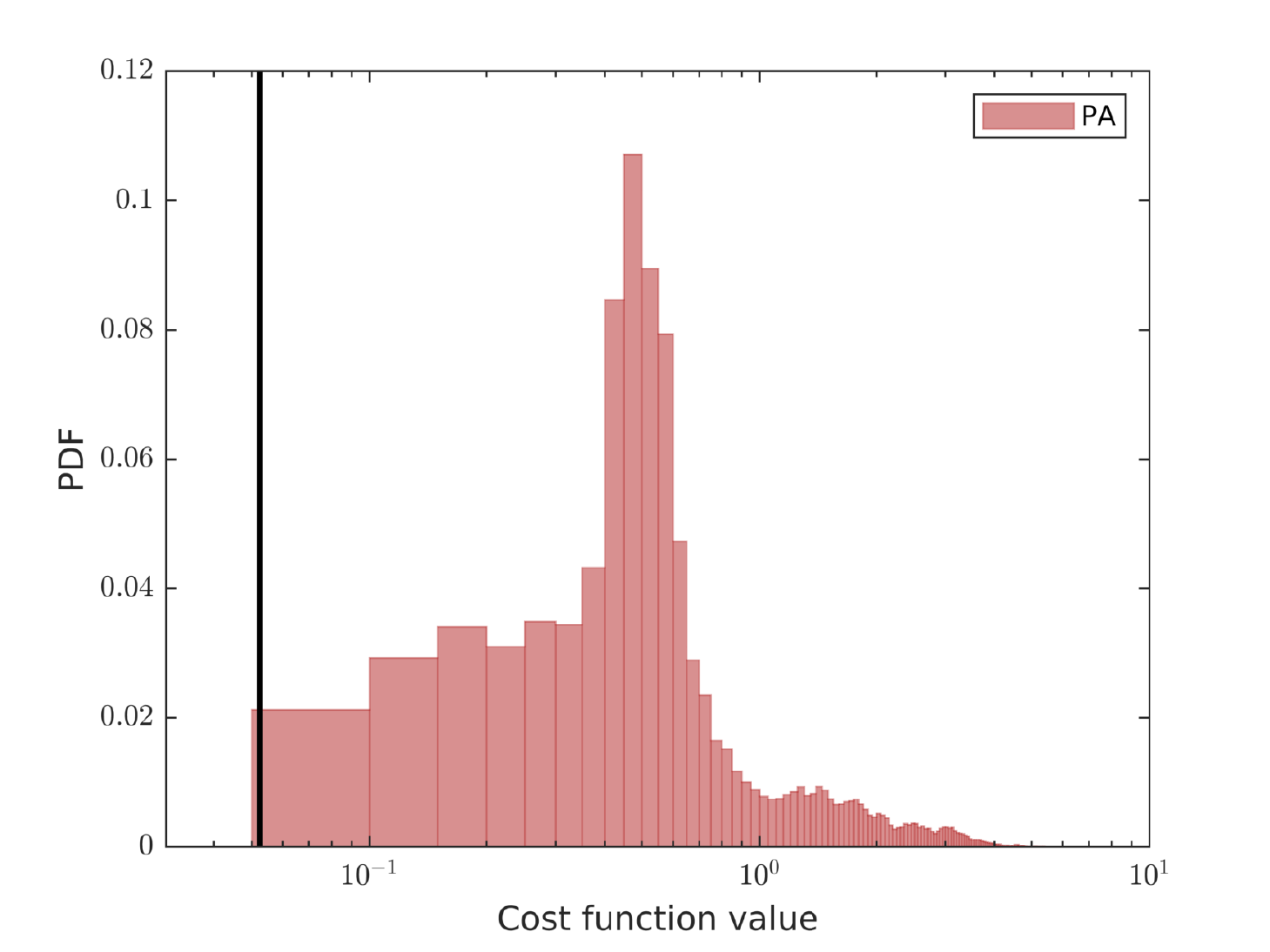} 
    \caption{}
  \end{subfigure}
  \begin{subfigure}[b]{0.45\textwidth}
    \centering
    \includegraphics[width=\textwidth]{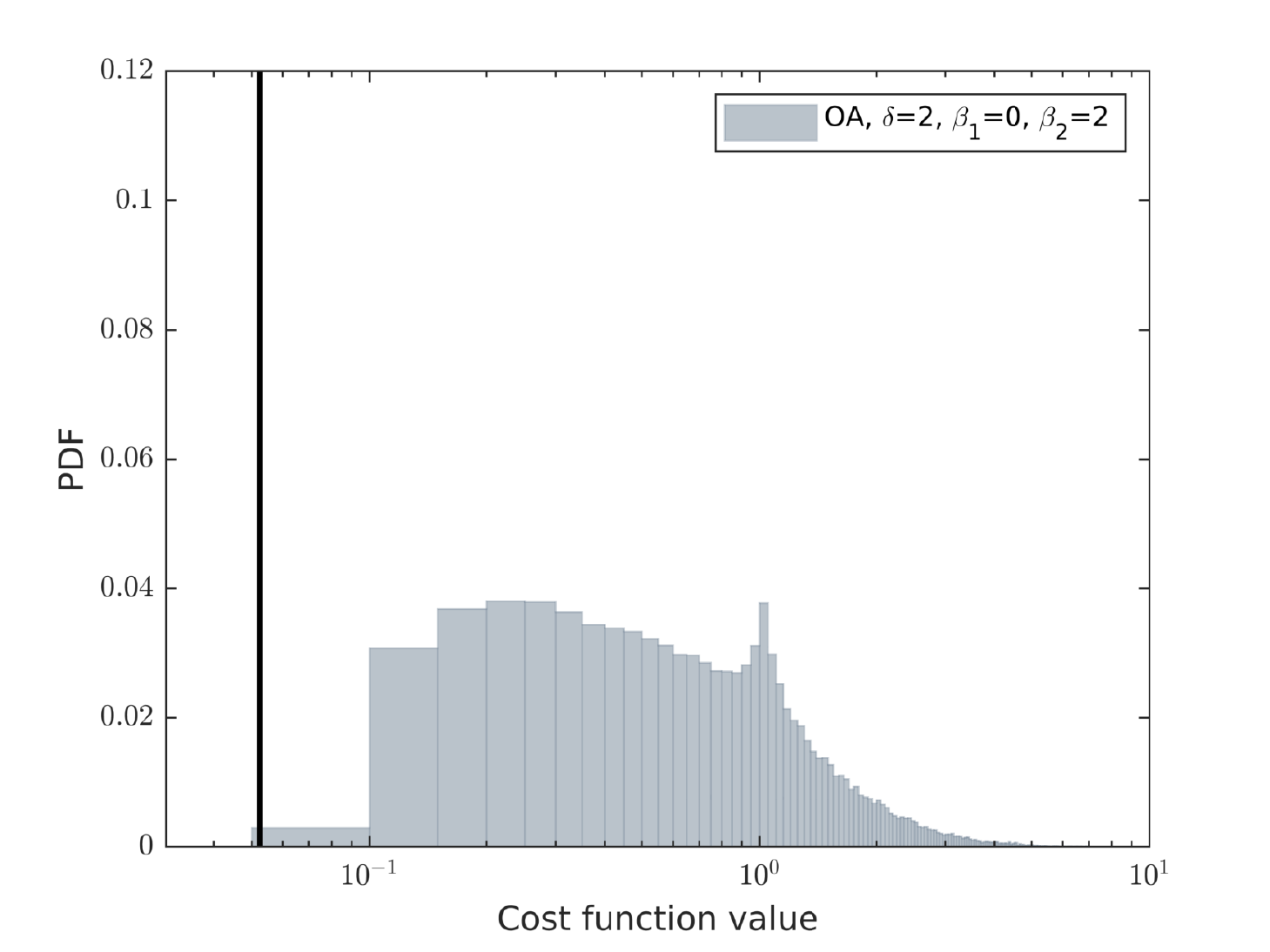} 
    \caption{}
  \end{subfigure}
}
\caption{Histograms of the error functional $\J(\bK)$ obtained along
  the Markov chains for problem \eqref{eq:minJ} with system
  \eqref{eq:dCc} closed using (a) the pair approximation and (b) the
  optimal approximation with exponents determined subject to hard
  regularization (OA-1/2).  The black vertical lines represent the values of
  the error functional $\J(\alpha_1^*,\alpha_2^*)$ obtained when the
  model based on the SA closure is used.}
  \label{fig:hist}
\end{figure}


Finally, some additional comments are in place as regards the results
shown in Figure \ref{fig:posterior}.  The parameters $k_6$ and $k_8$
are close to zero in the model with the closure based on the pair
approximation. These two parameters along with $k_5$ and $k_7$
contribute to the production and destruction of the $(--)$ cluster,
cf.~Figure \ref{fig:--}.  In equilibrium, the concentration of the
$(--)$ (or Li-Li) cluster is zero as is evident in Figure
\ref{fig:lattice}. Hence, the reaction rates have values needed to
annihilate the $(--)$ cluster. The linear triplets $(\overline{--+})$
and $(\overline{-+-})$ do not exist in the equilibrium state and
therefore both $k_6$ and $k_8$ are very close to zero, cf.~Figure
\ref{fig:--}. On the other hand, the $(\widehat{-+-})$ cluster does
exist in the equilibrium state, and hence the reaction rates $k_5$ and
$k_7$ are not zero. However, $k_7$ is bigger than $k_5$, highlighting
the fact that the $(--)$ cluster needs to be destroyed at equilibrium.

The analysis presented above is also true from the point of view of
the equilibrium constants in \eqref{eq:Q}. Indeed, the equilibrium
constants $Q_2$ and $Q_4$ reduce to unity when we use the model
closure based on the optimal approximation subject to hard
regularization. This forces the parameters $k_2$ and $k_6$ to be
highly correlated with $k_4$ and $k_8$, respectively. Additionally,
the equilibrium constants $Q_3$ and $Q_4$ control the production and
destruction of the $(--)$ cluster, cf.~\eqref{eq:Q}.  The constant
$Q_4$ reduces to unity and hence does not contribute to destruction of
the $(--)$ cluster at equilibrium, which is controlled by the
parameters $k_5$ and $k_7$. The wide range of obtained values of $k_6$
and $k_8$, cf.~Figures \ref{fig:posterior}a,c, indicates the low
sensitivity of the model to these parameters, which is in agreement
with our analysis.

\FloatBarrier

\section{Summary \& Conclusions}
\label{sec:final}
We have considered a mathematical model for the evolution of 
different cluster types in a structured lattice. We focused 
our attention on the structured lattice of a nickel-based oxide 
similar to those used in Li-ion batteries. That being said, 
the approach used here is much more broadly applicable. As 
is usual, the mean-clustering approach gives rise to an 
infinite hierarchy of ordinary differential equations, where 
concentrations of clusters of a certain size are described 
in terms of concentrations of clusters of higher order. This 
infinite hierarchy must be truncated at an arbitrary level 
and closed with a suitable closure model (or closure condition) 
in order to be solvable. This closure requires an approximation 
of the concentrations of the higher-order clusters in terms of 
the concentrations of lower-order ones. As a point of departure, 
we consider the pair approximation which is a classical closure 
model, and then introduce its generalization referred to as the 
optimal approximation which is calibrated using a novel data-driven 
approach.

The optimal approximation can be tuned for different levels of 
accuracy and robustness by adjusting the degree of regularization 
employed in the solution of the optimization problem. Our 
analysis shows that the model subject to soft regularization 
results in highly accurate approximations for the local stoichiometry 
but the accuracy deteriorates for other stoichiometries. On the 
other hand, the model subject to hard regularization has a lower 
accuracy at the local stoichiometry but is more robust with respect 
to changes of stoichiometry. The model subject to hard regularization 
produces more accurate results than the pair approximation for a 
broad range of stoichiometries. More importantly, the closure model 
found in this way turns out to have a simple structure with many 
exponents having nearly integer values. Exploiting this structure, 
we arrive at the sparse approximation model which is linear and 
therefore analytically solvable.

In addition to being simpler, the sparse approximation model is also
more accurate and robust than the pair approximation, in that it can
be applied to a wide range of stoichiometries without a significant
loss of accuracy. This model is interpretable as it makes it possible
to refine some of the simplifying assumptions at the heart of the pair
approximation. One of these assumptions states that the conditional
probability of $k$ being a nearest neighbour of $ij$ in a triplet
($ijk$) is equal to that of $k$ being a nearest neighbour of $j$. In
other words, it is assumed that every $j$ element in the lattice has a
nearest neighbour in state $i$. The sparse approximation refines this
assumption by adding a term that takes into account the conditional
probability of $j$ being a nearest neighbour of $i$. This correction
makes the model both simpler and more accurate.

The reaction rates in system (8) closed using one of the closure 
models are determined by formulating a suitable inverse problem. 
We solve these problems using a state-of-the-art Bayesian inference 
approach which also allows us to estimate the uncertainties of 
the reconstructed parameters. The results obtained show that 
the inverse problem is in fact ill-posed in the case of the closures 
based on the pair and optimal approximation, in the sense that the 
corresponding optimization problems admit multiple local minima. 
Moreover, these minima tend to be “shallow” reflecting the low 
sensitivity of the models closed with the pair and optimal 
approximations to the reaction rates. As a result, the inferred 
values of these parameters suffer from uncertainties on the order 
of 200\%. In contrast, the model closed using the sparse 
approximation is well-posed with respect to $\alpha_1$ and $\alpha_2$ which 
are linear combinations of reaction rates. This model is 
analytically solvable which completely eliminates the uncertainty 
in the reconstruction of its parameters. Based on these observations, 
we conclude that the sparse approximation is superior to the pair 
approximation.

Notably, the mean-cluster modelling approach considered in the present 
work can be used to describe the evolution of clusters of arbitrary size 
and type defined on structured lattices various types. The size and shape 
of the cluster and the structure of the lattice determine the reactions 
between elements. More complicated lattices and bigger cluster sizes 
involve more possible nearest-neighbour element swaps, resulting in a 
larger number of parameters in the model. The sparse approximation 
methodology could be utilized in a similar way to close the corresponding 
hierarchical models.

\pagebreak
\appendix
\section{Rotational Symmetry}
\label{sec:sym}

\begin{theorem}
  In a 2D triangular lattice (where each element is surrounded
  by $6$ nearest-neighbours), different spatial orientations of a
  particular $2$-cluster retain the same concentration, i.e., the
  probability of finding a particular $2$-cluster in the lattice is
  independent of its spatial orientation.
	\label{thm:1}
\end{theorem}
\begin{proof}
  Assuming one site in a $\fplus$ state, the concentration of this
  element can be obtained by summing over concentrations of all
  $2$-clusters in which the second element iterates over the possible
  elements in the system. Using ($\bullet$) to denote an
    unspecified state in the lattice, we then have
	\begin{equation}
	\begin{aligned}
	C \left( \stackengine{\Sstackgap}{\fplus}{}{O}{c}{\quietstack}{\useanchorwidth}{S} \right) = 
	C \left( \stackengine{\Sstackgap}{\fplus\hspace{0.1em}$\bullet$}{}{O}{c}{\quietstack}{\useanchorwidth}{S} \right)   =
	C \left( \stackengine{\Sstackgap}{\fplus\hspace{0.1em}\fplus}{}{O}{c}{\quietstack}{\useanchorwidth}{S} \right) +
	C \left( \stackengine{\Sstackgap}{\fplus\hspace{0.1em}\fminus}{}{O}{c}{\quietstack}{\useanchorwidth}{S} \right) \\ 
	= 
	C \left( \stackengine{\Sstackgap}{$\bullet$\hspace{0.1em}\fplus}{}{O}{c}{\quietstack}{\useanchorwidth}{S} \right)   =
	C \left( \stackengine{\Sstackgap}{\fplus\hspace{0.1em}\fplus}{}{O}{c}{\quietstack}{\useanchorwidth}{S} \right) +
	C \left( \stackengine{\Sstackgap}{\fminus\hspace{0.1em}\fplus}{}{O}{c}{\quietstack}{\useanchorwidth}{S} \right) \\ 
	=
	C \bigg( \stackengine{\Sstackgap}{\fplus}{~~~~$\bullet$}{O}{c}{\quietstack}{\useanchorwidth}{S} \bigg)   =
	C \bigg( \stackengine{\Sstackgap}{\fplus}{~~~~\fplus}{O}{c}{\quietstack}{\useanchorwidth}{S} \bigg) +
	C \bigg( \stackengine{\Sstackgap}{\fplus}{~~~~\fminus}{O}{c}{\quietstack}{\useanchorwidth}{S} \bigg) \\
	=
	C \bigg( \stackengine{\Sstackgap}{\fplus}{$\bullet$~~~~}{O}{c}{\quietstack}{\useanchorwidth}{S} \bigg)   =
	C \bigg( \stackengine{\Sstackgap}{\fplus}{\fplus~~~~}{O}{c}{\quietstack}{\useanchorwidth}{S} \bigg) +
	C \bigg( \stackengine{\Sstackgap}{\fplus}{\fminus~~~~}{O}{c}{\quietstack}{\useanchorwidth}{S} \bigg) \\
	=
	C \bigg( \stackengine{\Sstackgap}{\fplus}{~~~~$\bullet$}{U}{c}{\quietstack}{\useanchorwidth}{S} \bigg)   =
	C \bigg( \stackengine{\Sstackgap}{\fplus}{~~~~\fplus}{U}{c}{\quietstack}{\useanchorwidth}{S} \bigg) +
	C \bigg( \stackengine{\Sstackgap}{\fplus}{~~~~\fminus}{U}{c}{\quietstack}{\useanchorwidth}{S} \bigg) \\
	=
	C \bigg( \stackengine{\Sstackgap}{\fplus}{$\bullet$~~~~}{U}{c}{\quietstack}{\useanchorwidth}{S} \bigg)   =
	C \bigg( \stackengine{\Sstackgap}{\fplus}{\fplus~~~~}{U}{c}{\quietstack}{\useanchorwidth}{S} \bigg) +
	C \bigg( \stackengine{\Sstackgap}{\fplus}{\fminus~~~~}{U}{c}{\quietstack}{\useanchorwidth}{S} \bigg) \\
	\Longrightarrow \quad
	C \bigg( \stackengine{\Sstackgap}{\fplus\hspace{0.1em}\fminus}{}{O}{c}{\quietstack}{\useanchorwidth}{S} \bigg) =
	C \bigg( \stackengine{\Sstackgap}{\fminus\hspace{0.1em}\fplus}{}{O}{c}{\quietstack}{\useanchorwidth}{S} \bigg) =
	C \bigg( \stackengine{\Sstackgap}{\fplus}{~~~~\fminus}{O}{c}{\quietstack}{\useanchorwidth}{S} \bigg) =
	C \bigg( \stackengine{\Sstackgap}{\fplus}{\fminus~~~~}{O}{c}{\quietstack}{\useanchorwidth}{S} \bigg) =
	C \bigg( \stackengine{\Sstackgap}{\fplus}{~~~~\fminus}{U}{c}{\quietstack}{\useanchorwidth}{S} \bigg) =
	C \bigg( \stackengine{\Sstackgap}{\fplus}{\fminus~~~~}{U}{c}{\quietstack}{\useanchorwidth}{S} \bigg).
	\end{aligned}
	\end{equation}
\end{proof}

\begin{acknowledgments}
  AA and BP were supported by a Collaborative Research \& Development
  grant \# CRD494074-16 from the Natural Sciences \& Engineering
  Research Council of Canada and a Research Excellence grant \#
  RE-09-051 from the The Ontario Research Fund.  JF was supported by
  the Faraday Institution MultiScale Modelling (MSM) project Grant
  number EP/S003053/1.
\end{acknowledgments}

\pagebreak

%

\end{document}